\documentclass{daj}

\usepackage{amsmath}
\usepackage{amsthm}
\usepackage{amssymb}
\usepackage{cleveref}
\usepackage{wrapfig}
\usepackage{tikz-cd}

\newcommand{\norm}[1]{\left\lVert#1\right\rVert}

\newcommand{\B}{\mathcal{B}}

\newcommand{\R}{\mathbb R}
\newcommand{\Q}{\mathbb Q}
\newcommand{\cQ}{\mathcal Q}
\newcommand{\eps}{\epsilon}

\newcommand{\bx}{\mathbf x}
\newcommand{\bu}{\mathbf u}
\newcommand{\bv}{\mathbf v}
\newcommand{\bw}{\mathbf w}
\newcommand{\ba}{\mathbf a}
\newcommand{\bb}{\mathbf b}
\newcommand{\bt}{\mathbf t}

\newcommand{\cN}{\mathcal N}

\newcommand{\cF}{\mathcal{F}}
\newcommand{\cG}{\mathcal{G}}
\newcommand{\N}{\mathbb N}
\newcommand{\cM}{\mathcal{M}}

\renewcommand{\S}{\mathbb S}
\newcommand{\F}{\mathbb F}

\newcommand{\poly}{\operatorname{poly}}

\newcommand{\SDP}{\operatorname{SDP}}

\newcommand{\MAJ}{\operatorname{MAJ}}
\newcommand{\AT}{\operatorname{AT}}
\newcommand{\PAR}{\operatorname{Parity}}

\newcommand{\NAE}{\operatorname{NAE}}
\newcommand{\OR}{\operatorname{OR}}
\newcommand{\Pol}{\operatorname{Pol}}
\newcommand{\Ham}{\operatorname{Ham}}
\newcommand{\Spread}{\operatorname{Spread}}
\newcommand{\SpreadSingle}{\operatorname{Spread-s}}

\newcommand{\PCSP}{\operatorname{PCSP}}
\newcommand{\CSP}{\operatorname{CSP}}

\newcommand{\id}{\operatorname{id}}

\newcommand{\fiPCSP}{\operatorname{fiPCSP}}
\newcommand{\conv}{\operatorname{conv}}

\newtheorem{theorem}{Theorem}[section]
\newtheorem{claim}[theorem]{Claim}
\newtheorem{lemma}[theorem]{Lemma}
\newtheorem{definition}[theorem]{Definition}

\newtheorem{corollary}[theorem]{Corollary}
\newtheorem{conjecture}[theorem]{Conjecture}
\newtheorem{proposition}[theorem]{Proposition}
\newtheorem{remark}[theorem]{Remark}
\newtheorem{question}[theorem]{Question}

\dajAUTHORdetails{%
  title = {SDPs and Robust Satisfiability of Promise CSP}, %
  author = {Joshua Brakensiek, Venkatesan Guruswami, and Sai Sandeep},
  plaintextauthor = {Joshua Brakensiek, Venkatesan Guruswami, Sai Sandeep},
  plaintexttitle = {SDPs and Robust Satisfiability of Promise CSP}, %
}   %

\dajEDITORdetails{%
   year={2025},
   number={14},
   received={10 June 2023},   %
   revised={23 January 2025},    %
   published={9 September 2025},  %
   doi={10.19086/da.143808},       %
}   %

\begin{document}

\begin{frontmatter}[classification=text]
\title{SDPs and Robust Satisfiability of \\ Promise CSP\thanks{This paper is a significant expansion and revision of a preliminary version of this paper which appeared in the proceedings of the 2023 Symposium on the Theory of Computing (STOC 23).}} %

\author[josh]{Joshua Brakensiek\thanks{Supported in part by an NSF Graduate Research Fellowship and a Microsoft Research PhD Fellowship.}}
\author[venkat]{Venkatesan Guruswami \thanks{Supported in part by NSF grants CCF-2228287 and CCF-2211972 and a Simons Investigator award.}}
\author[sandeep]{Sai Sandeep\thanks{Supported in part by NSF grant CCF-2228287.}}

\begin{abstract}
For a constraint satisfaction problem (CSP), a robust satisfaction algorithm is one that outputs an assignment satisfying most of the constraints on instances that are near-satisfiable. It is known that the CSPs that admit efficient robust satisfaction algorithms are precisely those of bounded width, i.e., CSPs whose satisfiability can be checked by a simple local consistency algorithm (eg., 2-SAT or Horn-SAT in the Boolean case). While the exact satisfiability of a bounded width CSP can be checked by combinatorial algorithms, the robust algorithm is based on rounding a canonical Semidefinite Programming (SDP) relaxation.

In this work, we initiate the study of robust satisfaction algorithms for \emph{promise} CSPs, which are a vast generalization of CSPs that have received much attention recently. The motivation is to extend the theory beyond CSPs, as well as to better understand the power of SDPs. We present robust SDP rounding algorithms under some general conditions, namely the existence of particular high-dimensional Boolean symmetries known as majority or alternating threshold polymorphisms. 
On the hardness front, we prove that the lack of such polymorphisms makes the PCSP hard for all pairs of symmetric Boolean predicates. 
Our approach relies on SDP integrality gaps argued via the absence of certain colorings of the sphere, with connections to sphere Ramsey theory. 

We conjecture that PCSPs with robust satisfaction algorithms are precisely those for which the feasibility of the canonical SDP implies (exact) satisfiability. We also give a precise algebraic condition, known as a minion characterization, of which PCSPs have the latter property.
\end{abstract}
\end{frontmatter}

\section{Introduction}
Horn-SAT and 2-SAT are Boolean constraint satisfaction problems (CSPs) that admit simple combinatorial algorithms for satisfiability.
They are both examples of \emph{bounded width} CSPs: the existence of {a} locally consistent {set of} assignments (which satisfy all local constraints involving some bounded number of variables, and which are consistent on the intersections) implies the existence of a global satisfying assignment.

While the simple local propagation algorithms for Horn-SAT and  2-SAT work when the instance is perfectly satisfiable, they are not robust to errors---if the given instance is almost satisfiable, the local consistency based algorithms do not guarantee solutions that satisfy almost all the constraints. 
In a beautiful work, Zwick~\cite{Zwick98} initiated the study of finding ``robust" algorithms for CSPs, namely algorithms that output solutions satisfying $1-f(\epsilon)$ fraction of the constraints when the instance is promised to be $1-\epsilon$ satisfiable, where $f(\epsilon) \rightarrow 0$ as $\epsilon \rightarrow 0$. Zwick obtained robust algorithms for $2$-SAT using {Semidefinite Programming (SDP)} rounding and for Horn-SAT based on {rounding a Linear Program (LP)}. 
The PCP theorem together with Schaefer's reductions~\cite{Schaefer78} shows that Boolean CSPs that are NP-Hard are also APX-hard with perfect completeness, which in particular means that they do not admit robust satisfiability algorithms.
The only other interesting Boolean CSP besides Horn-SAT and 2-SAT for which satisfiability is polynomial-time decidable is Linear Equations modulo 2. H\aa stad~\cite{Has01} in his seminal work showed that even for $3$-LIN (when all equations involve just three variables), for every $\epsilon, \delta >0$, it is NP-Hard to output a solution satisfying $\frac{1}{2}+\delta$ fraction of the constraints even when the instance is guaranteed to have a solution satisfying $(1-\epsilon)$ fraction of the constraints.   

Unlike Horn-SAT or 2-SAT, the satisfiability algorithm for $3$-LIN is not local, and $3$-LIN does not have bounded width. Together with Schaefer's dichotomy theorem~\cite{Schaefer78}, this yields that for Boolean CSPs, bounded width characterizes robust satisfiability. For CSPs over general domains, a landmark result in the algebraic approach to CSP due to Barto and Kozik~\cite{BartoK14} showed that CSPs that are not bounded width can express linear equations. A reduction from H\aa stad's result then shows that CSPs that are not bounded width do not admit robust algorithms. Guruswami and Zhou~\cite{GuruswamiZ12} conjectured the converse---namely that all bounded width CSPs, over any domain, admit robust algorithms. Another work by Barto and Kozik~\cite{BartoK16} resolved this conjecture in the affirmative, thus giving a \emph{full characterization} of CSPs that have robust algorithms. 

In this work, we study robust algorithms for the class of \emph{Promise} Constraint Satisfaction Problems (PCSPs). PCSPs are a generalization of CSPs where each constraint has a strong form and a weak form. Given an instance that is promised to have a solution satisfying the stronger form of the constraints, the objective is to find a solution satisfying the weaker form of the constraints. A classic example of PCSPs is $(1$-in-$3$-SAT, NAE-$3$-SAT). While both the underlying CSPs are NP-Hard, the resulting PCSP does have a polynomial-time algorithm: given an instance of $1$-in-$3$-SAT that is promised to be satisfiable, we can find an assignment to the variables in polynomial time that satisfies each constraint as an NAE-$3$-SAT instance. PCSPs are a vast generalization of CSPs and capture key problems such as approximate graph and hypergraph coloring. 

Since their formal introduction in \cite{AGH17} and subsequent detailed study in \cite{BrakensiekG21} and \cite{BBKO21}, there has been a flurry of recent works on PCSPs~\cite{AsimiB21,AD21,AustrinBP20,BartoBB21,barto2022combinatorial,BrakensiekG19,BrakensiekG20,BrakensiekGWZ20,BGS21,BrandtsWZ21,ciardo2022clap}. These have led to a rich and still developing theory aimed at classifying the complexity of PCSPs, by tying their (in)tractability to the symmetries associated with their defining relations, and understanding the power and limitations of various algorithmic approaches influential for CSPs in the context of PCSPs.

Against this backdrop, we initiate the study of robust algorithms for PCSPs. The motivation is two-fold. First, as algorithms resilient to a small noise in the input, robust algorithms are important in their own right. Second, in the CSP world, the existence of efficient robust algorithms is equivalent to having bounded width and being decided by $O(1)$ levels of Sherali Adams for CSPs~\cite{ThapperZ17}. It is also proven~\cite{ThapperZ18} to be equivalent to being decided by the basic semidefinite programming (SDP) relaxation~\cite{Raghavendra08}. Here, we say that the basic SDP decides a CSP if for every instance $\Phi$ of the CSP, $\Phi$ has an assignment satisfying all the constraints if and only if there is a vector solution satisfying all the constraints in the SDP relaxation. 
For CSPs, therefore, the study of robust algorithms sheds light on, and in fact, precisely captures, the power of the most popular algorithmic approaches. Robust algorithms for PCSPs provide a rich context to understand how well these algorithmic tools generalize beyond CSPs. 

The main question that we are interested in this work is the following. 
\begin{question}
\label{ques:robust}
Which PCSPs admit polynomial time robust algorithms? 
\end{question}

As is the case with CSPs, a natural approach to characterize which PCSPs have robust algorithms is via the bounded width of PCSPs. However, it turns out that bounded width for PCSPs is weaker than having robust algorithms. 
Concretely, Atserias and Dalmau~\cite{AD21} have proved recently that the PCSP $(1$-in-$3$-SAT, NAE-$3$-SAT) does not have bounded width. Our work implies 
there is a robust algorithm for this PCSP. Atserias and Dalmau also proved that this PCSP is decided by $O(1)$ levels of Sherali-Adams, and as we shall prove later, it is also decided by the basic SDP.\footnote{We say that the basic SDP decides a PCSP (formally defined in~\Cref{sec:prelims}) if for every instance $\Phi$ of the PCSP, if there is a vector solution satisfying all the strong constraints in $\Phi$, then, $\Phi$ has an assignment satisfying all the weak constraints.} On the other hand, Raghavendra's framework of converting integrality gaps of CSPs to the hardness of approximation applies to PCSPs as well~\cite{Raghavendra08}. That is, his result implies that every PCSP that is not decided by the basic SDP does not admit a polynomial-time robust algorithm, assuming the Unique Games Conjecture~\cite{Khot02}. This gives a powerful tool to show the absence of a polynomial time robust algorithm for a {PCSP} (albeit under the Unique Games Conjecture): showing an integrality gap for {its} basic SDP relaxation. With this connection,~\Cref{ques:robust} naturally leads to the following question.  
\begin{question}
\label{ques:sdp}
Which PCSPs are decided by the basic SDP relaxation?
\end{question}

We make progress on~\Cref{ques:robust} and~\Cref{ques:sdp} by studying the \emph{polymorphisms} of PCSPs. Polymorphisms are closure properties of satisfying solutions to (Promise) CSPs. As a concrete example, consider the $2$-SAT CSP: given an instance $\Phi$ of $2$-SAT over $n$ variables $x_1, x_2, \ldots, x_n$, suppose that $\textbf{u}, \textbf{v}, \textbf{w} $ are three assignments to these variables satisfying all the constraints in $\Phi$, then the assignment $\textbf{z}$ that is coordinatewise Majority operation on three bits, i.e., $z_i = \MAJ (u_i, v_i, w_i)$ for every $i \in [n]$, also satisfies all the constraints in $\Phi$. This shows that the Majority function on three variables is a polymorphism of the $2$-SAT CSP. More generally, the Majority function on any odd number of variables is a polymorphism of the $2$-SAT CSP. Similarly, the Parity function on any odd number of variables is a polymorphism of $3$-LIN. On the other hand, there are no non-trivial polymorphisms for $3$-SAT.
Polymorphisms are the central objects in the \textit{{u}niversal algebraic approach} to CSPs~\cite{JeavonsCG97, Jeavons98, BulatovJK05,BartoKW17,Bul17,Zhuk20}, which has then been extended to PCSPs~\cite{BrakensiekG21, BBKO21}. 

At a high level, the existence of non-trivial polymorphisms implies algorithms, and vice-versa. The key challenge is to precisely characterize which polymorphisms lead to algorithms. It is known that the polymorphism family of a PCSP fully captures its computational complexity, i.e., if there are {promise templates }
$\Gamma, \Gamma'$ {(that is, a set of pairs of relations)} such that the polymorphism family of $\Gamma$, $\Pol(\Gamma)$ is contained in $\Pol(\Gamma')$, then ${\PCSP(\Gamma')}$ {(the decision problem for PCSP instances using relations from $\Gamma'$)} is formally easier than ${\PCSP(\Gamma)}$, i.e., there is a \textit{gadget reduction} from ${\PCSP(\Gamma')}$ to ${\PCSP(\Gamma)}$. It turns out that this gadget reduction preserves the existence of robust algorithms as well.  {This was used by Dalmau, Krokhin~\cite{DalmauK13} in the CSP setting, which we extend to PCSPs in \Cref{prop:ppp-defn})}. Thus,~\Cref{ques:robust} and~\Cref{ques:sdp} can be rephrased as \textit{Which polymorphisms lead to robust algorithms for PCSPs? Which polymorphisms lead to being decided by the basic SDP relaxation?}  

We make progress on these questions on two fronts: first, for a large class of Boolean symmetric\footnote{A predicate $P$ is symmetric if for every satisfying assignment $(x_1, \hdots, x_n)$ to $P$, any permutation of that assignment also satisfies $P$. For a Boolean predicate whether an assignment satisfies a predicate depends only on the Hamming weight. A PCSP is said to be symmetric if all the predicates in the template are symmetric.} 
PCSPs where we allow negation of variables, we characterize the polymorphisms that lead to robust algorithms. Our algorithms are based on novel rounding schemes for the basic SDP relaxation, and our hardness results are proved using integrality gaps for the basic SDP relaxation.  
Second, towards understanding the power of basic SDP for promise CSPs, we introduce a minion $\mathcal{M}$ and show that {$\PCSP(\Gamma)$} can be decided by basic SDP if and only if there is a minion homomorphism from $\mathcal{M}$ to {$\Pol(\Gamma)$}.

\subsection{Our results: Robust algorithms and hardness}

As is the case with CSPs, if a PCSP is NP-Hard, then it does not admit polynomial time robust algorithms, assuming $\text{P} \neq \text{NP}$. Thus, Question~\ref{ques:robust} is only relevant for PCSPs that can be solved in polynomial time. A large class of PCSPs for which polynomial time solvability has been fully characterized {is the problems based on Boolean symmetric promise templates}.
In~\cite{BrakensiekG21}, the authors showed that a Boolean symmetric PCSP with folding (i.e., we allow negating the variables) can be solved in polynomial time if and only if it contains at least one of Alternate-Threshold $(\AT)$, Majority ($\MAJ$) or $\PAR$ polymorphisms of all odd arities. {Later an analogous result was shown without the folding restriction in~\cite{FicakKOS19} where in addition one needs to consider all families of threshold polymorphisms.}

\smallskip \noindent \textbf{Robust algorithms.} Our main algorithmic result shows that in two of these cases when the PCSP has $\MAJ$ or $\AT$ polymorphisms of all odd arities, the PCSP admits a robust algorithm. {More strongly, we allow our instances to have negated variables and constant literals, which we call $\fiPCSP$, short for \emph{folded, idempotent} PCSP.}

\begin{theorem}
\label{thm:main-algorithm}
{Let $\Gamma$ be a Boolean promise template. If $\Pol(\Gamma)$} contains $\AT$ or $\MAJ$ polymorphisms of all odd arities{. Then, the robust version of $\fiPCSP(\Gamma)$} admits a polynomial time algorithm. More precisely, 
\begin{enumerate}
    \item If ${\Pol(\Gamma)}$ contains $\MAJ$ polymorphisms of all odd arities {then} there exists a polynomial time algorithm that given an instance of ${\fiPCSP}(\Gamma)$ {and $\eps > 0$} that is promised to have a solution satisfying $1-\epsilon$ fraction of the constraints, outputs a solution satisfying $1-\tilde O(\epsilon^{\frac{1}{3}})$ fraction of the constraints.\footnote{Here, $\tilde O$ hides multiplicative poly logarithmic factors.}
    \item If ${\Pol(\Gamma)}$ contains $\AT$ polymorphisms of all odd arities {then} there exists a polynomial time algorithm that outputs a solution satisfying $1-O\left(\frac{\log \log \frac{1}{\epsilon}}{\log \frac{1}{\epsilon}}\right)$ fraction of the constraints on an instance of ${\fiPCSP}(\Gamma)$ promised to have a solution satisfying $1-\epsilon$ fraction of the constraints. 
\end{enumerate}
\end{theorem}

Similar to the robust algorithms for CSPs~\cite{Zwick98,CharikarMM09,BartoK16}, our robust algorithms for {promise templates} with $\MAJ$ and $\AT$ polymorphisms are based on rounding the basic SDP relaxation. 
The main challenge here is to obtain robust algorithms for a large class of {promise templates} without access to the predicates and just using the properties of their polymorphisms. We achieve this using a combination of polymorphic tools where we use the fact that the {promise templates} contain $\AT$ or $\MAJ$ polymorphisms to deduce structural properties of the underlying predicate pairs, and SDP rounding tools where we then use these structural properties to get a robust algorithm. 

For $\MAJ$ polymorphisms, we first reduce the problem to the case when every weak predicate is of the form $k$-SAT. We then show that the robust algorithm of Charikar, Makarychev, and Makarychev~\cite{CharikarMM09} for the $2$-SAT {predicate} generalizes to these classes of {promise templates}.
While the analysis of~\cite{CharikarMM09} is tailored towards $2$-SAT, we give a completely different analysis that is not based on predicates and instead uses the existence of $\MAJ$ polymorphisms as a black box. Similarly, for the $\AT$ polymorphisms, we first use the properties of $\AT$ polymorphisms to reduce to the \emph{weighted hyperplane} {promise template} that generalizes the $(1$-in-$3$-SAT, NAE-$3$-SAT) {promise template}. We then give a robust algorithm for the weighted hyperplane {promise template} based on a random threshold rounding technique. A detailed overview of our algorithmic ideas appears in Sections~\ref{sec:maj-overview} and~\ref{sec:at-overview}.

\smallskip\noindent\textbf{Hardness results.}
Unlike our robust algorithms, which work for general Boolean {promise template} with the said polymorphisms, in our hardness results, we rely on the symmetry of the predicates of {the promise template}. 
Furthermore, we assume that the {promise template} contains a single predicate pair $\Gamma =(P,Q)$ that does not admit $\AT$ or $\MAJ$ polymorphisms of all odd arities. We show that for such Boolean symmetric folded PCSPs, the basic SDP relaxation has an integrality gap with perfect completeness, i.e., there is a finite instance on which the SDP relaxation satisfies all the strong constraints with zero error but the instance is not satisfiable even using the weak constraints. 
By Raghavendra's framework connecting SDP gaps and Unique-Games hardness~\cite{Raghavendra08}, the integrality gap rules out robust satisfaction algorithms (under the Unique Games conjecture (UGC)~\cite{Khot02}). 

\begin{theorem}
\label{thm:main-hardness}
Let $\Gamma={\{}(P,Q){\}}$ {be a promise template consisting of} a pair of symmetric Boolean predicates such that  $\AT_{L_1},\MAJ_{L_2} \notin \Pol(\Gamma)$ for some odd integers $L_1, L_2$. Then, under the UGC, unless $\text{P}=\text{NP}$, there is no polynomial time robust algorithm for the PCSP associated with {$\fiPCSP(\Gamma)$.}
\end{theorem}

Similar to our algorithmic result, we first use the properties of the polymorphisms to reduce to a small set of {promise templates}. We obtain integrality gaps for these {fiPCSP}s for the basic SDP relaxation, which then implies robust hardness under the UGC.  For CSPs, strong integrality gaps~\cite{Schoenebeck08,Tulsiani09,SchoenebeckTT07} are known for the basic SDP relaxation and its strengthenings such as the Lasserre hierarchy, almost all of them being random constructions. For the case of PCSPs, analyzing the random constructions is trickier since we need to sample the constraints with a precise density such that there is a vector solution to the strong constraints, but the weak constraints are not satisfied. Instead, we take the opposite approach where we first construct the vector solution and then add all the constraints that the vector solution satisfies. This is similar in spirit to Feige and Schechtman's integrality gap~\cite{FeigeS01} for MAX-CUT where they first sampled $n$ uniformly random points on a $d$-dimensional sphere and then added edges between every pair of points whose distance falls within a preset range. 

Toward this approach, we first construct an infinite integrality gap instance where the vertex set corresponds to the $n$-dimensional sphere $\S^n$ for a large integer $n$ and there are constraints for every tuple of vertices whose corresponding vectors satisfy the SDP constraints. For the set of {promise templates} that we study, we show that this instance is not satisfiable, even using weak constraints. A compactness argument then implies the existence of a finite integrality gap instance. 
As we shall see later, by using our minion characterization result, showing that the infinite instance has no satisfiable assignment is a necessary step to obtain a finite integrality gap instance. 
Toward showing that the infinite instance does not have an assignment satisfying all the weak constraints, we study colorings of the sphere $f:\S^n \rightarrow \{-1,+1\}$ and use a result of Matoušek and Rödl~\cite{matouvsek1995ramsey} from \emph{sphere Ramsey theory} where the existence of monochromatic configurations in colorings of the sphere are studied. While their result directly applies to some {promise templates}, for others, we combine their result with new techniques to prove the existence of structured configurations in sphere colorings. 
A more detailed overview appears in Section~\ref{sec:sdpgap-overview}. 

\smallskip\noindent\textbf{The power of SDPs and robust PCSP algorithms.}
Both our algorithmic and hardness results crucially use the basic SDP relaxation.
As our algorithms for the $\AT$ and $\MAJ$ polymorphisms are based on rounding the basic SDP, we get that every Boolean {promise template} that contains $\AT$ or $\MAJ$ polymorphisms is decided by the basic SDP. On the hardness front,~\Cref{thm:main-hardness}, shows that {many} Boolean symmetric folded PCSPs without $\AT$ or $\MAJ$ polymorphisms cannot be decided by the basic SDP. This suggests a more general relation between the basic SDP and robust algorithms for PCSPs. At an intuitive level, for both the existence of robust algorithms and being decided by the basic SDP, the underlying requirement seems to be the existence of polymorphism families that are robust to noise.  While our results show that this is true for the PCSPs that we study in this paper (noise stability is one crucial aspect that distinguishes $\MAJ$ and $\AT$ from $\PAR$.), we believe this is a more general phenomenon and motivates us to make the following conjecture. 
\begin{conjecture}
\label{conj:sdp-robust}
A {promise template} $\Gamma$ has a polynomial time robust algorithm if {$\PCSP(\Gamma)$} is decided by the basic SDP relaxation. Else, there is no polynomial time robust algorithm for {$\PCSP(\Gamma)$}, unless $\text{P}=\text{NP}$.
\end{conjecture}

As mentioned earlier, if there is an integrality gap for $\Gamma$ with respect to the basic SDP relaxation, then by Raghavendra's~\cite{Raghavendra08} result, we get that $\Gamma$ does not have a polynomial time robust algorithm, assuming the Unique Games Conjecture. This already proves one direction of~\Cref{conj:sdp-robust}. The other direction is more interesting: can we obtain robust algorithms for PCSPs just using the fact that basic SDP decides them? 
We remind the reader that the conjecture is already proven for CSPs, where the existence of robust algorithms~\cite{BartoK16} and decidability by basic SDP~\cite{ThapperZ18} are both shown to be equivalent to having bounded width. 

\subsection{Minion characterization of basic SDP}

In addition to our concrete characterization of robust algorithms for a subfamily of PCSPs, we also present a novel algebraic characterization of which PCSPs can be decided via basic SDP. Originally, in the study of CSPs, such algebraic characterizations were structured as follows (e.g., \cite{Bul17,Zhuk20}).

\begin{itemize}
\vspace{-1ex}
\item \emph{``Algorithm $\mathcal A$ solves $\operatorname{CSP}(\Gamma)$, if and only if there is a polymorphism $f \in \Pol(\Gamma)$ with specific properties.''}
\end{itemize}

Since the early days of PCSPs,  it has been known that a single polymorphism cannot dictate hardness  (c.f., \cite{BrakensiekG21}), and thus one must instead consider a \emph{sequence} of polymorphisms (e.g., \cite{BrakensiekGWZ20}):
\begin{itemize}
\vspace{-1ex}
\item \emph{``Algorithm $\mathcal A$ solves $\PCSP(\Gamma)$, if and only if there is an infinite sequence of polymorphism $f_1, f_2, \hdots  \in \Pol(\Gamma)$ with specific properties.''}
\end{itemize}
However, in many cases, such a characterization is unfeasible or unwieldy. Instead, a more general approach, pioneered by \cite{BBKO21}, captures the structure of polymorphism via a \emph{minion} (formally defined in Section~\ref{sec:minion}). 
A key property of the polymorphisms of {promise template} $\Gamma$ is that the function family $\Pol(\Gamma)$ is closed under taking \emph{minors}\footnote{A function $f:D_1^n \rightarrow D_2$ of arity $n$ is said to be a minor of another function $g:D_1^m \rightarrow D_2$ of arity $m$ with respect to a mapping $\pi:[m]\rightarrow [n]$ such that $f(x_1,x_2,\ldots,x_n)=g(x_{\pi(1)},x_{\pi(2)},\ldots,x_{\pi(m)})$ for every $\textbf{x}\in D_1^n$.}
A minion is an abstraction based on this: it is a collection of objects each with an arity, and for every object $a$ of arity $m$, and a mapping $\pi : [m]\rightarrow [n]$, there is a unique object $b$ of arity $n$ that is said to be a minor of $a$ {with respect to} $\pi$. A \emph{minion homomorphism} is a mapping between minions that preserves the minor operation. 
A powerful way to capture the limits of algorithms for PCSPs is via minion homomorphisms: 
\begin{itemize}
\vspace{-1ex}
\item \emph{``Algorithm $\mathcal A$ solves $\PCSP(\Gamma)$, if and only if there is minion homomorphism from $\mathcal M_{\mathcal A}$ to $\Pol(\Gamma)$.''}
\end{itemize}

Many recent papers \cite{BrakensiekGWZ20,ciardo2022clap,ciardo2022sherali} have proven such characterizations in various contexts. Our contribution to this line of work is showing that the basic SDP can be captured by a minion, which we call $\mathcal M_{\SDP}$.

\begin{theorem}\label{thm:minion}
The basic SDP decides {$\PCSP(\Gamma)$} if and only if there is a minion homomorphism from $\mathcal M_{\SDP}$ to $\Pol(\Gamma)$.
\end{theorem}

We note that a similar minion was concurrently and independently discovered by Ciardo-Zivny \cite{cz22-minion}. The theorem applies equally to Boolean and non-Boolean PCSPs.

The construction of the $\cM_{\SDP}$ minion is inspired by the vector interpretation of solutions to the basic SDP. Each object in the minion is a collection of orthogonal vectors which sum to a reference vector $\bv_0$. The minors involve adding groups of vectors together. Having a minion homomorphism from $\cM_{\SDP}$ to $\Pol(\Gamma)$ implies that there are polymorphisms of $\Gamma$ whose minors behave exactly like combining orthogonal vectors.

Proving Theorem~\ref{thm:minion} has a few technical hurdles. One challenge is that SDP solutions may require vectors of an arbitrarily large dimension. For these arbitrarily-large dimensional relationships to be captured in our minion, we have that the families of vectors making up $\cM_{\SDP}$ reside in a (countably) infinite-dimensional vector space. Similar techniques have been used in other minion constructions \cite{ciardo2022clap,ciardo2022sherali}.

Another challenge that appears specifically unique to this paper is that a Basic SDP solution gives a vector corresponding to each variable, but for the proof to go through additional vectors are needed which correspond to the constraints. (The variable vectors are "projections" of the constraint vectors.) Obtaining such constraints would typically be done via Sum-of-Squares or a related routine, but we prove that including such vector constraints are without loss of generality. That is, any basic SDP solution can be extended to a solution that includes constraint vectors without modifying the original variable vectors. This gives us enough vector structure to prove that the minion homomorphism corresponds to the basic SDP solution. 

 \medskip\noindent\textbf{Relation with sphere colorings.} 
 By a result of \cite{BBKO21}, there is a minion homomorphism $\cM_{\SDP} \to \Pol(\Gamma)$ if and only if there is an assignment satisfying all the constraints in a ``universal'' instance of $\PCSP(\Gamma)$ known as a \emph{free structure}. In the case that $\Gamma$ is a Boolean folded PCSP, this free structure for $\cM_{\SDP}$ turns out to be an instance where every possible unit vector is a variable. The constraints correspond to collections of vectors that satisfy the corresponding basic SDP constraints. This is precisely the same infinite instance that we use to show integrality gaps. 
 Thus, the result of~\cite{BBKO21} translates to {Boolean promise templates} as stating that a Boolean folded ${\PCSP(\Gamma)}$ is decided by the basic SDP if and only if there is an assignment satisfying all the constraints in the infinite integrality gap instance. 
 For the general theory of approximation of basic SDPs, similar constructs with sphere coloring being a `universal' gap have appeared in the literature (eg. in~\cite{brakensiek2021mysteries}). 

\subsection{Connections to Discrete Analysis} 
Before concluding the introduction, we want to outline some viewpoints related to Discrete Analysis, and specifically the analysis of Boolean functions, that underlie the polymorphic approach to CSPs in general, and the study of robust satisfiability in particular. 

First, polymorphisms embody a discrete analogue of convexity, as they are operations under which the space of satisfying assignments to a CSP are closed. The existence of certain symmetric Boolean functions (such as OR, Majority, Parity) as polymorphisms drives the tractable cases of Boolean CSP~\cite{Schaefer78,Chen09}. The study of Boolean promise CSPs is much richer, and includes polymorphisms such as the Alternating Threshold polymorphism which govern efficient satisfiability algorithms~\cite{BrakensiekG21,BrakensiekGWZ20,BBKO21}. 

Among variants of promise CSPs, robust PCSPs (and more broadly approximate PCSPs) are particularly suited to the tools of discrete analysis. 
Satisfiability algorithms can be very brittle to noise, with Gaussian elimination to solve linear equations being the quintessential example. In fact this is inherent as solving linear systems becomes highly intractable under noise~\cite{Has01}, and this can be attributed to the fact that the underlying polymorphism (Parity function) is highly sensitive to noise.
The robustness criterion seems to require that the associated polymorphisms are essentially ``smooth'' objects (i.e., concentrate on low-degree Fourier coefficients and demonstrate some noise stability). Understanding such phenomenon has been pivotal in developments to the theory of approximation algorithms over the last few decades, including the formulation of the Unique Games Conjecture \cite{Khot02a,KKMO07,MOO,Mossel,analysisOdonnell,KMS18}. The more brittle algebraic methods used throughout the study of (P)CSPs therefore seem less useful. This is in part why the breakthrough of Raghavendra~\cite{Raghavendra08} for approximate CSP satisfaibility algorithms relied heavily on analytical methods whereas the proof of the CSP Dichotomy Theorem needed a number of breakthroughs in universal algebra~\cite{JeavonsCG97, Jeavons98, BulatovJK05,BartoKW17,Bul17,Zhuk20}. 
Even with this ``smoothness'' condition simplifying the landscape, proving a dichotomy for all promise CSPs will be much trickier than the corresponding result for CSPs due to a much richer family of smooth polymorphisms such as Alternating Threshold.
Thus, robust algorithms for promise CSPs provides an excellent avenue to further develop the analytical machinery that has shaped the approximation algorithms field in CSPs, and potentially uncover deeper connections with universal algebra tools.

As mentioned earlier, our hardness results are obtained by connections to \emph{sphere Ramsey theory}, and specifically the existence of certain monochromatic vector configurations in colorings of high-dimensional spheres. Such continuous analogs of questions normally posed in discrete/combinatorial settings seem interesting in their own right, and hold many exciting challenges. For instance, are there density analogs of such sphere Ramsey statements which bound the measure of subsets that avoid monochromatic configurations (see Section~\ref{sec:conclusion} for a discussion)?

\bigskip\noindent\textbf{Organization of the paper.}
We first start by introducing formal definitions and some general observations in~\Cref{sec:prelims} (experts familiar with SDPs and CSPs can skip or just skim this section). We then give a detailed technical overview of our results in~\Cref{sec:overview}. We provide our algorithmic results (\Cref{thm:main-algorithm}) in~\Cref{sec:alg} and prove the hardness results (\Cref{thm:main-hardness}) in~\Cref{sec:ug-hardness}. We propose and establish properties of the basic SDP minion in ~\Cref{sec:minion}.
Finally, we conclude in~\Cref{sec:conclusion} with several intriguing challenges and open problems raised by our work.

\section{Preliminaries}
\label{sec:prelims}

\noindent \textbf{Notations.} We use $[n]$ to denote the set $\{1,2,\ldots, n\}$. {We let $\B := \{-1,+1\}$} A predicate or a relation over a domain $D$ of arity $k {\ge 0}$ is a subset of $D^k$.
For a relation $P \subseteq {D}^k$ of arity $k$, we abuse the notation and use $P$ both as a subset of ${D}^k$, and also as a function $P:{D}^k \rightarrow \{0,1\}$. We use boldface letters to denote vectors and roman letters to denote their elements, e.g., $\textbf{x}=(x_1, x_2,\ldots,x_k)$. We have $\S^n := \{ \textbf{v} \in \R^{n+1}: \norm{\textbf{v}}_2=1\}$.
For a vector $\textbf{v} \in D_1^k$ and a function $f:D_1\rightarrow D_2$, we use $f(\textbf{v})\in D_2^k$ to denote $(f(v_1),f(v_2),\ldots,f(v_k))$.

For a vector $\textbf{x} \in {\B}^k$, we use $\textsf{hw}(\textbf{x})$ to denote the number of $+1$s in $\textbf{x}$, i.e., $\textsf{hw}(\textbf{x})=\frac{k+\sum_{i=1}^nx_i}{2}$. For $S \subseteq \{0,1,\ldots,k\}$, we use $\Ham_k S$ to denote $\{\textbf{x} \in \{-1,+1\}^k : \textsf{hw}(\textbf{x})\in S\}$. We use $\NAE_k$ to denote the set $\Ham_k \{1,2,\ldots,k-1\}$, and $k$-SAT to denote the set $\Ham_k \{1,2,\ldots,k\}$.
For vectors $\textbf{x}, \textbf{y} \in \R^n$, { we use  $\langle \textbf{x}, \textbf{y}\rangle$ to denote $\sum_i x_i y_i$}. 

\subsection{{(P)CSP instances, templates,} and polymorphisms}

{We now work toward defining the main computational problem we seek to study: symmetric idempotent Boolean folded promise constraint satisfaction problems (PCSP) as well as the main tool we shall use for analyzing the complexity of such problems: \emph{polymorphisms}. To build toward these definitions, we begin by defining constraint satisfaction problems (CSP). The main ``parameter'' defining a (P)CSP is its set of allowable constraint types, which we call a \emph{template}. We can then use the template to build \emph{instances} of the (P)CSP which have associated \emph{assignments} for which one can define multiple computational \emph{problems}.}

\begin{definition}[CSP {template/instance/assignment}] {We define a \emph{CSP template} to be a finite set of predicates $\Gamma =\{P_1, P_2, \ldots, P_l\}$ over a finite domain $D$, where for all $i \in [\ell]$, $P_i$ has arity $k_i \ge 0$.}
{An \emph{instance} $\Phi=(V,\mathcal{C})$ of the constraint satisfaction problem (CSP) associated with the template $\Gamma$ (i.e., ``$\Phi$ is an instance of $\CSP(\Gamma)$''), consists of a variable set $V$ and a constraint set $\mathcal C$. The variable set consists of} $n$ variables $V =\{u_1, u_2, \ldots, u_n\}.$ {The} $m$ constraints $\mathcal{C}=\{C_1, C_2, \ldots, C_m\}$ 
each {consist} of a tuple of variables $C_j = (u_{j,1}, u_{j,2}, \ldots, u_{j,{l_j}}) \in V^{l_j}$ and an associated predicate $P^{(j)} \in \Gamma$ of the same arity $l_j$. 
An assignment $\sigma : V \rightarrow D$ is said to satisfy the constraint $C_j$ if $\sigma(C_j)=(\sigma(u_{j,1}), \sigma(u_{j,2}), \ldots, \sigma(u_{j,{l_j}})) \in P^{(j)}$.
\end{definition}

{We now define two computational problems associated with $\CSP(\Gamma)$ where an instance $\Phi$ of $\CSP(\Gamma)$ is given as input.

\begin{definition}[{CSP decision/search problems}] In the \emph{decision version} of $\CSP(\Gamma)$, the objective is to decide if there is an assignment to a given instance of $\CSP(\Gamma)$ that satisfies all the constraints. In the \emph{search version} of $\CSP(\Gamma)$, the objective is to output an assignment to a given instance of $\CSP(\Gamma)$ that satisfies all the constraints.
\end{definition}}

We next define {promise constraint satisfaction problems} (PCSP).  {We begin by defining what a promise template is.}\footnote{{Other works, such as \cite{BBKO21}, parameterize the promise CSP as a pair of a relational structures with a homomorphism between them.}}

{
\begin{definition}[promise template]
    Let $D_1, D_2$ be finite domains. A promise template consists of a finite collection $\Gamma = \{ (P_1,Q_1), (P_2,Q_2),\ldots, (P_l,Q_l)\}$ of pairs of predicates, where for all $i \in [\ell]$, $P_i$ has arity $k_i\ge 0$ over domain $D_1$ and each $Q_i$ has arity $k_i$ over domain $D_2$. Furthermore, we assert there exists a map $h:D_1\rightarrow D_2$ such that for all $i\in[l]$ and $\textbf{x} \in D_1^{k_i}$, $\textbf{x} \in P_i$ implies $h(\textbf{x}) \in Q_i$. In other words, $h$ is a homomorphism from $P_i$ to $Q_i$ for all $i \in [l]$.
\end{definition}
}

We say that $(D_1,D_2)$ is the domain of the promise template $\Gamma$.

\begin{definition}[{PCSP instance/assignments}]
	In an instance $\Phi=(V,\mathcal{C})$ of $\PCSP(\Gamma)$, we have a set of $n$ variables $V =\{u_1, u_2, \ldots, u_n\}$ and $m$ constraints $\mathcal{C}=\{C_1, C_2, \ldots, C_m\}$ 
each consisting of a tuple of variables $C_j = (u_{j,1}, u_{j,2}, \ldots, u_{j,{l_j}}) \in V^{l_j}$ and an associated predicate pair $(P^{(j)},Q^{(j)}) \in \Gamma$ of the same arity $l_j$.

{A $P$-assignment is a map $\sigma_1 : V \rightarrow D_1$ and a $Q$-assignment is a map $\sigma_2 : V \rightarrow D_2$.} A {$P$}-assignment {$\sigma_1$} is said to strongly satisfy the constraint $C_j$ if $\sigma(C_j)\in P^{(j)}$, and {a $Q$-assignment $\sigma_2$} is said to weakly satisfy the constraint $C_j$ if $\sigma(C_j) \in Q^{(j)}$.
\end{definition}

{Crucially, due to the existence of the homomorphism $h$ connecting each $P_i$ with its corresponding $Q_i$, any $P$-assignment strongly satisfying a collection of constraint can be converted into a $Q$-assignment weakly satisfying those constraints by composing $h$ with the $P$-assignment.} {If $\Gamma = \{(P,Q)\}$ consists a single pair of predicates, we often state $\Gamma = (P,Q)$ and write $\PCSP(P, Q)$ for $\PCSP(\Gamma)$.} The following are computational problems associated with {instances of} $\PCSP(\Gamma)$. 

\begin{definition}[{PCSP problems}]\label{def:pcsp-prob} In the decision version of $\PCSP(\Gamma)$, given an input instance $\Phi=(V,\mathcal{C})$ of $\PCSP(\Gamma)$, the objective is to distinguish between the two cases.
  \begin{enumerate}
\item  There is a $P$-assignment $\sigma_1 : V \rightarrow D_1$ that strongly satisfies all the constraints.
\item There is no $Q$-assignment $\sigma_2 : V \rightarrow D_2$ that weakly satisfies all the constraints.  
\end{enumerate}

In the search version of $\PCSP(\Gamma)$, given an input instance $\Phi=(V,\mathcal{C})$ of $\PCSP(\Gamma)$ with the promise that there is an assignment $\sigma_1: V\rightarrow D_1$ that strongly satisfies all the constraints, the objective is to find an assignment $\sigma_2: V \rightarrow D_2$ that weakly satisfies all the constraints.\footnote{{This explains the terminology \emph{promise} CSPs to refer to such problems. For the decision version, a promise problem consists of disjoint sets of Yes and No instances, which together need not exhaust all instances.}}
 
{An algorithm for $\PCSP(\Gamma)$ is said to be $f$-robust algorithm for a function $f:\R \rightarrow \R$ if the following holds. The inputs to the algorithm are an instance $\Phi=(V,\mathcal{C})$ of $\PCSP(\Gamma)$ and a parameter $\epsilon$} with the promise that there is an assignment $\sigma_1: V\rightarrow D_1$ that strongly satisfies $1-\epsilon$ fraction of the constraints. {The objective of the algorithm is to output an assignment} $\sigma_2: V \rightarrow D_2$ that weakly satisfies at least $1-f(\epsilon)$ fraction of the constraints. {We call an algorithm a robust algorithm for $\PCSP(\Gamma)$ if it is $f$-robust for some monotone and non-negative function $f$ with $f(\epsilon)\rightarrow 0$ as $\epsilon \rightarrow 0$.}
\end{definition}

In this paper, we restrict ourselves to PCSP \emph{templates} where both the domains are equal to the set ${\B :=} \{-1,+1\}$. {We call such templates \emph{Boolean} templates.} {For technical reasons, to more closely align with the conventions in the robust CSP literature, we consider a different notion of an PCSP instance which incorporates \emph{variable folding}, i.e., the negation of variables. We also assume that the PCSP instance allows for the setting of constants, which we call \emph{idempotence}.}

\begin{definition}[Boolean folded idempotent {instance/assignment}] {Given a Boolean template $\Gamma=\{(P_1,Q_1),$ $(P_2,Q_2),\ldots,$ $(P_l,Q_l)\}$ where $P_i \subseteq Q_i \subseteq \B^{k_i}$ for every $i \in [l]$, an instance $\Phi=(V,\mathcal{C})$ of $\fiPCSP(\Gamma)$ consists of a set of $n$ variables $V = \{u_1, \hdots, u_n\}$ and a set of $m$ constraints $\mathcal{C} = \{C_1, \hdots, C_m\}$. For $i \in [n]$, each variable $u_i$ has two corresponding \emph{literals}: $u_i$ and $\overline{u_i}$, the positive and negative literals, respectively. {Observe that we identify the positive literal with the variable itself.} We also have two constant literals corresponding to the elements of $\B$. For each $j \in [m]$, $C_j$ consists} of a tuple of literals $C_j = (x_{j,1}, x_{j,2}, \ldots, x_{j,{l_j}}) \in {(}V\cup \overline{V} \cup \B{)}^{l_j}$ and an associated predicate pair $(P^{(j)},Q^{(j)}) \in \Gamma$ of the same arity $l_j$. 

Given an assignment\footnote{{Since the families $P_i$ and $Q_i$ of predicates have a common domain $\B$, we need not distinguish between ``$P$-assignments'' and ``$Q$-assignments'' in the Boolean setting.}} $\sigma : V \rightarrow \B,$ {define its extended assignment $\sigma' : V \cup \overline{V} \cup \B \rightarrow \B$ to be} $\sigma'(u_i)=\sigma(u_i)$ and $\sigma'(\overline{u_i})=-\sigma(u_i)$ for every $i \in [n]$ {and $\sigma'(b) = b$ for $b \in \B$}. The assignment $\sigma$ is said to strongly (and {respectively} weakly) satisfy the constraint $C_j$ in $\Phi$ if $\sigma'(C_j) \in P^{(j)}$ (and {respectively} $\sigma'(C_j) \in Q^{(j)}$). 
\end{definition}

{We now restate the portion of Definition~\ref{def:pcsp-prob} for robust PCSP in the fiPCSP setting.
{
\begin{definition}
    
An algorithm for $\fiPCSP(\Gamma)$ is said to be $f$-robust algorithm for a function $f:\R \rightarrow \R$ if the following holds. The inputs to the algorithm are an instance $\Phi=(V,\mathcal{C})$ of $\fiPCSP(\Gamma)$ and a parameter $\epsilon$ with the promise that there is an assignment $\sigma_1: V\rightarrow D_1$ that strongly satisfies $1-\epsilon$ fraction of the constraints. The objective of the algorithm is to output an assignment $\sigma_2: V \rightarrow D_2$ that weakly satisfies at least $1-f(\epsilon)$ fraction of the constraints. We call an algorithm a robust algorithm for $\fiPCSP(\Gamma)$ if it is $f$-robust for some monotone and non-negative function $f$ with $f(\epsilon)\rightarrow 0$ as $\epsilon \rightarrow 0$.
\end{definition}
}

Associated with every PCSP {template, the notion of polymorphism captures the various} closure properties of the satisfying solutions {to any instance of the corresponding PCSP.} More formally, we can define the polymorphisms of a PCSP {template} as follows. 

\begin{definition}[Polymorphisms of PCSPs]
	{Given a promise template} $\Gamma = \{ ((P_1,Q_1), (P_2, Q_2), \ldots,$ $(P_l,Q_l))\}$ such that for every $i \in [l]$, $P_i\subseteq D_1^{k_i}, Q_i\subseteq D_2^{k_i}$, a polymorphism of arity ${L}$ is a function $f:D_1^{{L}} \rightarrow D_2$ that satisfies the below property for every $i \in [l]$.
	For all $\textbf{v}_1,\textbf{v}_2,\ldots,\textbf{v}_{k_i}\in D_1^{L}$ satisfying $( (\textbf{v}_1)_j, (\textbf{v}_2)_j, \ldots, (\textbf{v}_{k_i})_j ) \in P_i$ for each $j \in [{L}]$, we have 
	\[
	(f(\textbf{v}_1), f(\textbf{v}_2), \ldots , f(\textbf{v}_{k_i})) \in Q_i
	\]
    {If $\Gamma$ is Boolean (i.e., $D_1 = D_2 =\B$), we say that $f:\B^n\rightarrow \B$ is \emph{folded} if $f(-\textbf{v})=-f(\textbf{v})$ for all $\textbf{v}\in \B^n$.} {Likewise, we say that that $f$ is idempotent if $f(1,1,\ldots,1)=1$ and ${f(-1,-1,\ldots,-1)=-1}$.}
	We use $\emph{Pol}(\Gamma)$ to denote the family of all the polymorphisms of {$\Gamma$}.
\end{definition}

{Given a predicate $P \subseteq D_1^k$ and a function $f : D_1^L \to D_2$, we let $f(P) \subseteq D_2^k$ denote the set
\[
    f(P) := \bigcup_{\bx_1, \hdots, \bx_L \in P} f(\bx_1, \hdots, \bx_L),
\]
where $f(\bx_1,\hdots, \bx_L) := (f(x_{1,1},\hdots, x_{L,1}), \hdots, f(x_{1,k}, \hdots, x_{L,k}))$ is the component-wise application of $f$.}
Given a family $\cF$ of functions $f_i : D_1^{L_i} \to D_2$ 
we let $O_{\cF}(P) := \bigcup_{f \in \cF} f(P).$

{For our algorithmic results, we focus on studying Boolean promise templates $\Gamma$ whose corresponding polymorphisms contain particular families of functions.} {Specifically, we} extensively study the {a}lternate-{t}hreshold $(\AT)$ and {m}ajority $(\MAJ)$ polymorphisms in this paper:
\begin{enumerate}
    \item For an odd integer $L\geq 1$ and $\textbf{x}\in{\B}^L$, we have $
    \AT_L(\textbf{x})=+1$, if $x_1-x_2+x_3-\ldots+x_L >0$, and $-1$, otherwise.
    \item For an odd integer $L\geq 1$ and $\textbf{x}\in{\B}^L$, we have $
    \MAJ_L(\textbf{x})=+1$, if $x_1+x_2+x_3+\ldots+x_L >0$, and $-1$, otherwise.
\end{enumerate}
{We let $\AT$ denote the set of $\AT_L$ for $L \ge 1$ odd and $\MAJ$ the set of $\MAJ_L$ for $L \ge 1$ odd.}

\smallskip \noindent \textbf{Relaxations of PCSPs.} We say that a {promise template} $\Gamma'$ is a relaxation of another {promise template }$\Gamma$ if $\Pol(\Gamma)\subseteq \Pol(\Gamma')$. If $\Gamma'$ is a relaxation of $\Gamma$, then there is a \textit{gadget reduction} from $\Gamma' \cup \{(=, =)\}$ to $\Gamma \cup \{(=, =)\}$, where $=$ denotes the equality predicate in the relevant domain. More formally, it is referred to as $\Gamma' \cup \{(=, =)\}$ is \textit{positive primitive promise (ppp)-definable} from $\Gamma \cup \{(=, =)\}$. 

\begin{definition}[ppp-definability of PCSPs~\cite{Chen09, BrakensiekG21}]
\label{def:ppp}
We say that a {pair of predicates $(P',Q')$ of the same arity} is ppp-definable from a {promise template} $\Gamma$ over the same domain pair if there exists a fixed constant $l$ and an instance $\Phi$ of $\PCSP(\Gamma)$ over $k+l$ variables $u_1, u_2, \ldots, u_k, v_1, v_2, \ldots, v_l$ such that 
\begin{enumerate}
    \item If $(x_1, x_2, \ldots, x_k) \in P'$, then there exist $y_1, y_2, \ldots, y_l$ such that the assignment $(x_1, \ldots, x_k,y_1, \ldots, y_l)$ strongly satisfies all the constraints in $\Phi$. 
    \item If there is an assignment $(z_1, \ldots, z_{k+l})$ weakly satisfying all the constraints in $\Phi$, then $(z_1, z_2, \ldots, z_k)\in Q'$. 
\end{enumerate}
More generally, we say that {a promise template} $\Gamma'$ is ppp-definable from {another promise template} $\Gamma$ if every predicate pair in $\Gamma'$ is ppp-definable from $\Gamma$. 
\end{definition}
Note that if $\Gamma'$ is ppp-definable from $\Gamma$, then the decision version of $\PCSP(\Gamma')$ can be reduced to $\PCSP(\Gamma)$ in polynomial time. 
We now observe that the same holds for the robust version as well. More formally, we have the following proposition.\footnote{We previously claimed falsely that the reduction holds more generally when we are given that $\Gamma'$ is a relaxation of $\Gamma$. We are grateful to Andrei Krokhin for pointing this mistake to us. Under the general assumption that $\Gamma'$ is a relaxation of $\Gamma$, we can reduce the robust version of $\Gamma'$ to the robust version of $\Gamma$ together with binary equality constraints. For the sake of exact satisfiability, the equality constraints do not matter~\cite{BartoKW17,BrakensiekG21} but we do not know a way to obtain a robust algorithm for {$\PCSP(\Gamma)$} together with binary equality constraints just using the fact that $\Gamma$ has a robust algorithm. Even for CSPs, {a} direct way to show this is not known~\cite{DalmauK13}, without using the full robust dichotomy result of~\cite{BartoK16}.}

\begin{proposition}
\label{prop:ppp-defn}
Suppose that the {promise template} $\Gamma'$ over {$(D_1, D_2)$} is ppp-definable from $\Gamma$ over the same domain{.} {If there exists an $f$-robust algorithm for $\PCSP(\Gamma)$, then there exists an $O_{\Gamma,\Gamma'}(1)\cdot f$-robust algorithm for $\PCSP(\Gamma')$.}
\end{proposition}

\begin{proof}[{Proof}]
Given an instance $\Phi'=(V',\mathcal{C}')$ of ${\PCSP(\Gamma')}$ over a set of variables $V'$ and containing $m$ constraints $\mathcal{C}'=\{C'_1,C'_2,\ldots,C'_m\}$, we output an instance $\Phi$ of ${\PCSP(\Gamma)}$ containing $|V'|$ original variables and a set of dummy variables. 
For every constraint $C'_j$ using $(P', Q')$ of arity $k_j$ involving the {tuple of literals} $(u_{j,1}, u_{j,2}, \ldots, u_{j,k_j}) $ in $\Phi'$, we have a set of dummy variables $w_{j,1}, w_{j,2}, \ldots, w_{j,l_j}$ and a set of constraints ${\mathcal C}_j$ {whose variable scope is} ${S_j \cup \{}w_{j,1}, w_{j,2}, \ldots, w_{j,l_j}\}$ as in~\Cref{def:ppp}. Let $W=\cup_j \{ w_{j,1}, w_{j,2}, \ldots, w_{j,l_j}\}$ and let $V=V'\cup W$ be the set of variables of $\Phi$, and $\mathcal{C}=\cup_{j\in[m]}{\mathcal{C}}_j$ {is} the set of constraints in $\Phi$.
We claim that this reduction preserves robust algorithms. \begin{enumerate}
    \item (Completeness). Suppose that there exists an assignment $\sigma': V' \rightarrow D_1$  strongly satisfying $1-\epsilon$ fraction of the constraints in $\Phi'$. For every $j \in [m]$ such that $\sigma'$ strongly satisfies the constraint $C'_j$, there is an assignment $\sigma_j:\{ w_{j,1}, w_{j,2}, \ldots, w_{j,l_j}\}\rightarrow D_1$ such that $\sigma \cup \sigma_j : C_j\rightarrow D_1$ strongly satisfies all the constraints {of $\mathcal{C}_j$}. Consider an assignment $\sigma: V \rightarrow D_1$ where we set $\sigma(u_i)=\sigma'(u_i)$ for $u_i \in V'$, $\sigma(w_{j,i})=\sigma_j(w_{j,i})$ for all $w_{j,i}\in W$ such that $\sigma$ strongly satisfies the constraint $C'_j$. If $\sigma$ does not strongly satisfy the constraint $C'_j$, we set $\sigma(w_{j,i})$ arbitrarily. The assignment $\sigma$ strongly satisfies at least $1-O_{\Gamma,\Gamma'}(\epsilon)$ fraction of the constraints in $\Phi$. 
    \item (Soundness). Suppose that there is an assignment $\sigma : V \rightarrow D_2$
    weakly satisfying $1-\epsilon$ fraction of the constraints in $\Phi$. {By Markov's inequality, f}or at least $1-O_{\Gamma,\Gamma'}(\epsilon)$ values of $j \in [m]$, all the constraints in ${\mathcal C}_j$ are weakly satisfied by $\sigma$. This shows that the assignment $\sigma$ restricted to $V'$ weakly satisfies $1-O_{\Gamma,\Gamma'}(\epsilon)$ fraction of the constraints in $\Phi'$. \qedhere
\end{enumerate}
\end{proof}

{In the folded, idempotent setting, we also allow for negations of variables and the setting of constants, we represent these additional operations via what we coin as ``fippp''-definitions.}

{
\begin{definition}[fippp-definability of fiPCSPs]
\label{def:fippp}
We say that a {pair of predicates $(P',Q')$ of the same arity} is fippp-definable from a {promise template} $\Gamma$ over the same domain pair if there exists a fixed constant $l$ and an instance $\Phi$ of $\fiPCSP(\Gamma)$ over $k+l$ variables $u_1, u_2, \ldots, u_k, v_1, v_2, \ldots, v_l$ such that 
\begin{enumerate}
    \item If $(x_1, x_2, \ldots, x_k) \in P'$, then there exist $y_1, y_2, \ldots, y_l$ such that the assignment $(x_1, \ldots, x_k,y_1, \ldots, y_l)$ strongly satisfies all the constraints in $\Phi$. 
    \item If there is an assignment $(z_1, \ldots, z_{k+l})$ weakly satisfying all the constraints in $\Phi$, then $(z_1, z_2, \ldots, z_k)\in Q'$. 
\end{enumerate}
More generally, we say that {a promise template} $\Gamma'$ is fippp-definable from {another promise template} $\Gamma$ if every predicate pair in $\Gamma'$ is fippp-definable from $\Gamma$. 
\end{definition}
}

{
\begin{proposition}
\label{prop:fippp-defn}
Suppose that the {promise template} $\Gamma'$ over {$(D_1, D_2)$} is fippp-definable from $\Gamma$ over the same domain{.} {If there exists an $f$-robust algorithm for $\fiPCSP(\Gamma)$, then there exists an $O_{\Gamma,\Gamma'}(1)\cdot f$-robust algorithm for $\fiPCSP(\Gamma')$.}
\end{proposition}

The proof of Proposition~\ref{prop:fippp-defn} is essentially identical to that of Proposition~\ref{prop:ppp-defn}. 

\begin{proof}[{Proof}] 
Given an instance $\Phi'=(V',\mathcal{C}')$ of ${\fiPCSP(\Gamma')}$ over a set of variables $V'$ and containing $m$ constraints $\mathcal{C}'=\{C'_1,C'_2,\ldots,C'_m\}$, we output an instance $\Phi$ of ${\fiPCSP(\Gamma)}$ containing $|V'|$ original variables and a set of dummy variables. 
For every constraint $C'_j$ using $(P', Q')$ of arity $k_j$ involving the {tuple of literals} $(u_{j,1}, u_{j,2}, \ldots, u_{j,k_j}) $ in $\Phi'$, we have a set of dummy variables $w_{j,1}, w_{j,2}, \ldots, w_{j,l_j}$ and a set of constraints ${\mathcal C}_j$ {whose variable scope is} $S_j \cup \{w_{j,1}, w_{j,2}, \ldots, w_{j,l_j}\}$ as in~\Cref{def:fippp}. Let $W=\cup_j \{ w_{j,1}, w_{j,2}, \ldots, w_{j,l_j}\}$ and let $V=V'\cup W$ be the set of variables of $\Phi$, and $\mathcal{C}=\cup_{j\in[m]}{\mathcal{C}}_j$ {is} the set of constraints in $\Phi$.
We claim that this reduction preserves robust algorithms. \begin{enumerate}
    \item (Completeness). Suppose that there exists an assignment $\sigma': V' \rightarrow D_1$  strongly satisfying $1-\epsilon$ fraction of the constraints in $\Phi'$. For every $j \in [m]$ such that $\sigma'$ strongly satisfies the constraint $C'_j$, there is an assignment $\sigma_j:\{ w_{j,1}, w_{j,2}, \ldots, w_{j,l_j}\}\rightarrow D_1$ such that $\sigma \cup \sigma_j : C_j\rightarrow D_1$ strongly satisfies all the constraints {of $\mathcal{C}_j$}. Consider an assignment $\sigma: V \rightarrow D_1$ where we set $\sigma(u_i)=\sigma'(u_i)$ for $u_i \in V'$, $\sigma(w_{j,i})=\sigma_j(w_{j,i})$ for all $w_{j,i}\in W$ such that $\sigma$ strongly satisfies the constraint $C'_j$. If $\sigma$ does not strongly satisfy the constraint $C'_j$, we set $\sigma(w_{j,i})$ arbitrarily. The assignment $\sigma$ strongly satisfies at least $1-O_{\Gamma,\Gamma'}(\epsilon)$ fraction of the constraints in $\Phi$. 
    \item (Soundness). Suppose that there is an assignment $\sigma : V \rightarrow D_2$
    weakly satisfying $1-\epsilon$ fraction of the constraints in $\Phi$. {By Markov's inequality, f}or at least $1-O_{\Gamma,\Gamma'}(\epsilon)$ values of $j \in [m]$, all the constraints in ${\mathcal C}_j$ are weakly satisfied by $\sigma$. This shows that the assignment $\sigma$ restricted to $V'$ weakly satisfies $1-O_{\Gamma,\Gamma'}(\epsilon)$ fraction of the constraints in $\Phi'$. \qedhere
\end{enumerate}
\end{proof}
}

\subsection{The basic SDP}

We now describe the Basic SDP relaxation of an instance of a PCSP, similar to how it is presented in \cite{Raghavendra08}. Let {$\Gamma$ be a promise template over domain $(D,D')$, and let} $\Phi$ be an instance of {$\PCSP(\Gamma)$} over $n$ variables $V=\{u_1, u_2, \ldots, u_n\}$ and $m$ constraints $C_1, C_2, \ldots, C_m$. Suppose that the constraint $C_j$ {subjects the variables} $( u_{j,1}, u_{j,2}, \ldots, u_{j,l_j})$ {to the predicate pair} $(P^{(j)}, Q^{(j)})$.
In the basic SDP relaxation of ${\PCSP(\Gamma)}$ corresponding to $\Phi$, we have a vector $\textbf{v}_{i,a}$ corresponding to each variable $u_i$, $i \in [n]$, along with a label $a \in D$. We also have a unit vector $\textbf{v}_0$. 
For each constraint $C_j$, $j \in [m]$, there is a probability distribution (referred to as the \emph{local distribution} of the constraint $C_j$) {over the set of possible local assignments to that constraint.} 
{Since $C_j$ may contain repeated variables, we let $S_j := \{u_{j,1}, \hdots, u_{j,l_j}\}$ denote the set of variables appearing in $C_j$. Our probability distribution is then over the set of functions $\{f : S_j \to D\}$.} We represent this using a variable $\lambda_j(f)$ for every $j\in [m]$ and {local} assignment $f: {S_j} \rightarrow D$. {Note that we use $f(C_j)$ as shorthand for $( f(u_{j,1}), f(u_{j,2}), \ldots, f(u_{j,l_j}))$} Finally, we have an error parameter $\eps_j$ corresponding to the constraint $j, j \in [m]$, equal to the probability that $\lambda_j$ is supported outside $P^{(j)}$. We refer to $\epsilon_j$ as the error of the basic SDP relaxation on the constraint $j$, and $\sum_{j=1}^m \eps_j$ as the error of the basic SDP relaxation of the instance $\Phi$.
\begin{align*}
    \textbf{minimize: } & \sum_{j=1}^m \eps_j & \\
\textbf{subject to: } \eps_j &\geq 0  & \forall j \in [m]\\
\lambda_j(f) &\geq 0  & \forall j\in [m], ~ f:{S}_j \rightarrow {D} \\ 
\sum_{f:{S}_j \rightarrow D}\lambda_j(f)&=1  & \forall j \in [m] \\ 
 \sum_{f:{S}_j \rightarrow D, f(C_j)\in P^{(j)}} \lambda_j(f) &= 1-\eps_j \quad & \forall j \in [m]\\
{ \langle \bv_0, \bv_0\rangle} &{= 1}\\
\text{ (First moments.)  } \qquad { \langle \textbf{v}_{i,a},   \textbf{v}_0 \rangle } &= \sum_{\substack{f:{S}_j \rightarrow D\\f({u_i}) = a}}\lambda_j(f)&  \forall j\in [m],~{u_i} \in {S}_j,~ a \in D \\ 
\text{(Second moments.)  } \quad { \langle \textbf{v}_{i,a} , \textbf{v}_{i',a'} \rangle } &= \sum_{\substack{f:{S}_j \rightarrow D\\f({u_i}) = a, f({u_{i'}}) = a'}}\lambda_j(f)  & \forall j\in [m],~{u_i,u_{i'}} \in {S}_j,~ a,a' \in D 
\end{align*}
{We implicitly assume that every variable appears in at least one clause, so the second moment condition also implies that $\langle v_{i,a}, v_{i,a'} \rangle = 0$ for $a \neq a'$ for all $i \in [n]$.}

\begin{proposition}\label{prop:basic-SDP-completeness}
Let $\Gamma$ be a promise template over domain $(D_1, D_2)$. Suppose that $\Phi=(V,\mathcal{C})$ is an instance of a $\PCSP(\Gamma)$ such that there is an assignment $\sigma : V \rightarrow D_1$ that strongly satisfies all the constraints. Then, the basic SDP {relaxation of $\Phi$} is feasible {with objective value $0$}.
\end{proposition}
\begin{proof}
    {We embed our vectors in the real line.} We set $\textbf{v}_0=1$, and for every $i \in [n]$ and $a \in D_1$, we set
    \[
        \textbf{v}_{i,a} := \begin{cases}
        1 & a = \sigma(u_i),\\
        0 & \text{otherwise.}
        \end{cases}
    \] {For any} $j \in [m]$ and $f:{S}_j \to D_1$, we set $\lambda_j(f)=1$ if $f(x)=\sigma(x)$ for every $x \in {S}_j$, and we set $\lambda_j(f)=0$ otherwise. These variables satisfy all the constraints in the basic SDP relaxation with $\eps_j=0$ for all $j \in [m]$.
\end{proof}

\noindent \textbf{Boolean folded {idempotent} variant.}
In Sections~\ref{sec:alg} and \ref{sec:ug-hardness} of the paper, we prefer to consider an alternative formulation of the Basic SDP for the {Boolean folded idempotent PCSP.} {We let $\Gamma$ be a promise template and c}onsider an instance $\Phi=(V,\mathcal{C})$ of {$\fiPCSP(\Gamma)$} where $V=\{u_1,u_2,\ldots,u_n\}$ and $\mathcal{C}=\{C_1,C_2,\ldots,C_m\}$ with the constraint $C_j$ using the predicate pair $(P^{(j)},Q^{(j)})$ for $j\in[m]$. {Let $\mathcal L = V \cup \overline{V} \cup \B$ be the set of literals. Each clause is a tuple of literals $C_j = (x_{j,1}, x_{j,2}, \ldots, x_{j,{l_j}}) \in \mathcal L^{l_j}$ bound by the predicate pair $(P^{(j)}, Q^{(j)})$.}

{For each variable $u_i \in V$, we have an associated vector $\textbf{v}_i$ along with a truth vector $\bv_0$}
{We extend this association to all literals by assigning a map $\bv : \mathcal L \to \{\bv_1, \hdots, \bv_n\}$ where
\begin{align}
    \bv(x) = \begin{cases}
        \bv_i & x = u_i \in V\\
        -\bv_i & x = \overline{u_i} \in \overline{V}\\
        x\cdot \bv_0 & x \in \B.
        \end{cases}\label{eq:vl}
\end{align}}

{For each clause $C_j \in \mathcal C$, we let $S_j$ denote the set of variables appearing in $C_j$. Namely, $S_j = \{u_i \in V : \exists k \in [l_j], x_{j,l_j} = u_i \vee x_{j,l_j} = \overline{u_i}\}$. Given a local assignment of variables $f : S_j \to \B$, we let $f(C_j) = (f(x_{j,1}), f(x_{j,2}), \ldots, f(x_{j,{l_j}}))$, where we define $f(\overline{u_i}) = -f(u_i)$ and $f(\pm 1) = \pm 1$.}

{The relaxation is as follows.}
\begin{align*}
    \textbf{minimize: } & \sum_{j=1}^m \eps_j & \\
\textbf{subject to: } \eps_j &\geq 0 & \forall j \in [m]\\
\lambda_j(f) &\geq 0 & \forall j\in [m], ~f:{S}_j \rightarrow {\B} \\ 
\sum_{f:{S}_j \rightarrow {\B}}\lambda_j(f)&=1 & \forall j \in [m] \\ 
 \sum_{f:{S}_j \rightarrow {\B}, f(C_j)\in P^{(j)}} \lambda_j(f) &= 1-\eps_j & \forall j \in [m]\\
\norm{\textbf{v}_i }_2^2 &= 1 & \forall i \in \{0,1,\ldots,n\} \\ 
\text{ (First moments.)  } \qquad { \langle \textbf{v}(x) , \textbf{v}_0 \rangle }&= \sum_{f:{S}_j \rightarrow {\B}}\lambda_j(f) f(x) & \forall j\in [m],~x \in C_j \\
\text{(Second moments.)  } \quad { \langle \textbf{v}(x) , \textbf{v}(x') \rangle } &= \sum_{f:{S}_j \rightarrow {\B}}\lambda_j(f) f(x)f(x') & \forall j\in [m],~x,x' \in C_j
\end{align*}

We say that basic SDP is feasible on $\Phi$ if the above objective function is zero on $\Phi$. We show that the SDP is feasible if there is an assignment that strongly satisfies all the constraints of $\Phi$. 
\begin{proposition}
    Suppose that $\Phi=(V,\mathcal{C})$ is an instance of a Boolean folded {idempotent} PCSP such that there is an assignment $\sigma : V \rightarrow {\B}$ that strongly satisfies all the constraints. Then, the basic SDP {relaxation of $\Phi$} is feasible {with objective value $0$}.
\end{proposition}
\begin{proof}
    {We embed our vectors in the real line.} We set $\textbf{v}_0=1$, and $\textbf{v}_i := \sigma(u_i) \in \R$ for every $i \in [n]$. {For any} $j \in [m]$ and $f:{S}_j \rightarrow {\B}$, we set $\lambda_j(f)=1$ if $f(x)=\sigma(x)$ for every $x \in C_j$, and we set $\lambda_j(f)=0$ otherwise. These variables satisfy all the constraints in the basic SDP relaxation with $\eps_j=0$ for all $j \in [m]$.
\end{proof}

More generally, we get that if there is an assignment that strongly satisfies $1-\epsilon$ fraction of the constraints in $\Phi$, the objective value of the above relaxation is at most $\epsilon m$, for every $\epsilon \geq 0$.
On the other hand, if the basic SDP is feasible for an instance $\Phi$ of a PCSP, it doesn't necessarily imply that $\Phi$ has an assignment weakly satisfying all the constraints. For some PCSPs however, this is indeed the case, and we say that such PCSPs are decided by the basic SDP.

\begin{definition}
{For a promise template $\Gamma$, w}e say that the basic SDP \textit{decides} {$\PCSP(\Gamma)$} if for every instance $\Phi$ such that the basic SDP is feasible on $\Phi$, there is an assignment to $\Phi$ that weakly satisfies all the constraints. 
\end{definition}

{\paragraph{Efficient solvability.} Polynomial-time SDP solving algorithms can only solve the problem ``approximately.'' In the SDP literature, this notion of approximation can vary widely. For our purposes, we are satisfied if the objective is approximate (computed to within $1/\poly(n)$ precision), but the solution returned is a ``feasible'' solution to the SDP.

\begin{theorem}
\label{thm:sdp}
    Suppose an instance $\Phi$ of $\fiPCSP(\Gamma)$ has a solution strongly satisfying $1-\epsilon$ fraction of the constraints. Then, for any input $\delta >0$, we can find a vector solution to the basic SDP whose coefficients are square roots of rational numbers and all inner products are rational numbers. The objective value of the SDP solution is at most $(\eps + \delta)m$ and it can be computed in time $\poly(n, m, 1/\delta)$.
\end{theorem}
This precise formulation does not seem to appear in the literature~\cite{grotschel1993Geometric,freund2004introduction,gartner2012approximation,odonnell2016SOS}.
As such, we show how this result can be proved using known methods. See Appendix~\ref{app:efficient-sdp} for details.
}
{\subsection{Gram and orthogonal matrices}\label{subsec:gram-ortho}
Recall that a real symmetric matrix $M \in \R^{n \times n}$ is positive semidefinite if $\bv^{T} M \bv \ge 0$ for all $\bv \in \R^n$. Since $M$ is real and symmetric, we know that $M$ has a full eigenspace: there exists $\lambda_1, \hdots, \lambda_n \in \R$ and unit vectors $\bv_1, \hdots, \bv_n \in \R^n$ such that 
\[
M = \sum_{i=1}^n \lambda_i \bv \bv_i^{T}.
\]
Since $M$ is positive semidefinite, we know that each $\lambda_i \ge 0$. That is, if we let $V \in \R^{n \times n}$ be the matrix such that the $i$th row is $\sqrt{\lambda_i} \bv_i$, then we have that $M = V^{T} V$. In other words, $M$ is also a \emph{Gram matrix}. It is also straightforward that any Gram matrix is positive semidefinite since if $M = V^{T} V$, then for any $\bv \in \R^n$, we have that $\bv^{T} M \bv = \|V \bv\|_2^2 \ge 0$. We call $V$ a \emph{realization} of the Gram matrix, and we call the columns of $V$ the vectors of the realization.

A crucial fact about realizations is that they generate the same Gram matrix if and only if they are related by an \emph{orthogonal matrix}, a matrix $Q \in \R^{n \times n}$ for which $Q^{T} Q = Q Q^{T} = I_n$. We prove a generalization of this equivalence where the number of vectors differs from the ambient dimension.

\begin{proposition}\label{prop:gram-ortho}
Let $n, m$ be positive integers and let $\bu_1, \hdots, \bu_n, \bv_1, \hdots, \bv_n \in \R^m$ be vectors. The following are equivalent.
\begin{itemize}
\item [(1)] For all $i,j \in [n]$, we have that $\langle \bu_i, \bu_j\rangle = \langle \bv_i, \bv_j \rangle$.
\item [(2)] There exists an orthogonal matrix $Q \in \R^{m \times m}$ such that for all $i \in [n]$, $\bv_i = Q \bu_i$
\end{itemize}
\end{proposition}
\begin{proof}
Proving that (2) implies (1) is straightforward for if $\bv_i = Q\bu_i$ for all $i \in [n]$, we have that
\[
    \langle \bv_i, \bv_j\rangle = \langle Q\bu_i, Q\bu_j\rangle = (Q\bu_i)^{T} Q\bu_j = \bu_i^{T} Q^T Q \bu_j = \langle \bu_i, \bu_j\rangle,
\]
where we use that $Q$ is an orthogonal matrix.

We now prove that (1) implies (2). Let $U := \operatorname{span}\{\bu_i : i \in [n]\} \subset \R^{m}$. Using the Gram-Schmidt algorithm, we can construct an orthonormal basis $\ba_1, \hdots, \ba_m$ of $\R^n$ such that $\ba_1, \hdots, \ba_{\dim U}$ is a basis of $U$. (Assume $\dim U \ge 1$ or else every vector equals $0$ for which we can pick $Q$ to be the identity matrix.)

Also pick $S \subseteq [n]$ such that $\{\bu_{s} : s \in S\}$ is a basis of $U$. There exists a unique linear map $L : U \to \R^m$ for which $L(\bu_{s}) = \bv_s$ for all $s \in S$. We claim the following properties about $L$,
\begin{itemize}
\item[(1)] $L(\bu_i) = \bv_i$ for all $i \in [n]$, so the image of $L$ is $V:= \operatorname{span}\{\bv_i : i \in [n]\}$.
\item[(2)] $L(\ba_1), \hdots, L(\ba_{\dim U})$ is an orthonormal basis of $V$.
\end{itemize}
To verify (1), for any $i \in [n]$, let $\lambda : S \to \R$ be such that $\bu_i = \sum_{s \in S} \lambda(s)\bu_s$. Then, $L(\bu_i) = \sum_{s \in S} \lambda(s)\bv_s$. Note then that
\begin{align*}
\langle L(\bu_i), \bv_i\rangle &= \sum_{s \in S} \lambda(s) \langle \bv_s, \bv_i\rangle = \sum_{s\in S} \lambda(s) \langle \bu_s, \bu_i\rangle = \langle \bu_i, \bu_i\rangle = \langle \bv_i, \bv_i\rangle, \text{ and}\\
\langle L(\bu_i), L(\bu_i)\rangle &= \sum_{s,s' \in S} \lambda(s)\lambda(s') \langle \bv_s, \bv_{s'}\rangle = \lambda(s)\lambda(s') \langle \bu_s, \bu_{s'}\rangle = \langle \bu_i, \bu_i\rangle = \langle \bv_i, \bv_i\rangle.
\end{align*}
This can only happen if $L(\bu_i) = \bv_i$. To verify (2), for each $i \in [\dim U]$, let $\kappa_i : S \to \R$ be such that $\ba_i = \sum_{s\in S} \kappa_i(s) \bu_s$. Then, note that for any $i,j \in [\dim U]$ we have that
\[
\langle \ba_i, \ba_j \rangle = \sum_{s,s'\in S}\kappa_i(s)\kappa_i(s') \langle \bu_s, \bu_{s'}\rangle = \sum_{s,s'\in S}\kappa_i(s)\kappa_i(s') \langle L(\bu_s), L(\bu_{s'})\rangle = \langle L(\ba_i), L(\ba_j), \rangle,
\]
so $L(\ba_1), \hdots, L(\ba_{\dim U})$ is an orthonormal basis of $V$.

To finish, we can extend $L(\ba_1), \hdots, L(\ba_{\dim U})$ into an orthonormal basis $\bb_1, \hdots, \bb_m$ of $\R^m$ such that $\bb_i = L(\ba_i)$ for all $i \in [\dim U]$. We can thus construct a linear map $Q : \R^m \to \R^m$ such that $Q(\ba_i) = \bb_i$ for all $i \in [m]$. The matrix representation of $Q$ is an orthogonal matrix and it sends $\bu_i$ to $\bv_i$ for all $i \in [n]$ because $Q$ extends $L$.
\end{proof}

As an immediate corollary, we recover the equivalence of realizations.
\begin{corollary}\label{cor:gram-ortho}
Let $M = U^{T} U$ and $N = V^{T} V$ be $n \times n$ real matrices. We have that $M= N$ if and only if there exists an orthogonal matrix $Q$ such that $V = QU$. 
\end{corollary}
\begin{proof}
Apply Proposition~\ref{prop:gram-ortho} with $m = n$, $\bu_1, \hdots, \bu_n$ as the columns of $U$, and $\bv_1, \hdots, \bv_n$ as the columns of $V$.
\end{proof}

}

\subsection{Elementary properties of Gaussians}
We {state} a couple of elementary properties of Gaussian distribution that we use later. 
First, we prove the following anti-concentration inequality for the standard Gaussian random variable. 
{
\begin{proposition}
\label{prop:gaussian-anticoncentration-general}
Suppose that \( X \sim \mathcal{N}(0, \sigma^2) \) is a Gaussian random variable with zero mean and variance \( \sigma^2 \). Then, for every \( c \in \mathbb{R} \) and \( \epsilon \geq 0 \),
\[
\Pr\left( X \in [c-\epsilon, c+\epsilon] \right) \leq \frac{\epsilon}{\sigma}.
\]
\end{proposition}

\begin{proof}
By the definition of the Gaussian density, we have:
\[
\Pr\left( X \in [c-\epsilon, c+\epsilon] \right) = \int_{c-\epsilon}^{c+\epsilon} \frac{1}{\sqrt{2\pi \sigma^2}} e^{-\frac{(x-\mu)^2}{2\sigma^2}} dx.
\]

By substituting $ Z = \frac{X}{\sigma}$, we get
\begin{align*}
\Pr\left( X \in [c-\epsilon, c+\epsilon] \right) &=    \Pr\left( Z \in \left[\frac{c-\epsilon}{\sigma}, \frac{c+\epsilon}{\sigma}\right] \right) \\ 
&= \int_{\frac{c-\epsilon}{\sigma}}^ {\frac{c+\epsilon}{\sigma}} \frac{1}{\sqrt{2\pi}} e^{-\frac{z^2}{2}} dz \\
&\leq \frac{2\epsilon}{\sqrt{2\pi}\sigma} \\ 
&\leq \frac{\epsilon}{\sigma}. \qedhere
\end{align*}
\end{proof}
}

\begin{proposition}
\label{prop:gaussian-anticoncentration}
Suppose that $X\sim \mathcal{N}(0,1)$ has the standard Gaussian distribution. Then, for every { $c$ and} $\epsilon \geq 0$, 
\[
{\text{Pr}\left( X \in [c-\epsilon, c+\epsilon]\right)} \leq \epsilon.
\]
\end{proposition}

\begin{proof}
We have 
\[
    {\text{Pr}\left( X \in [c-\epsilon, c+\epsilon]\right) = \int_{c-\epsilon}^{c+\epsilon} \frac{1}{\sqrt{2\pi}}e^{-\frac{x^2}{2}}dx  \leq \int_{c-\epsilon}^{c+\epsilon} \frac{1}{\sqrt{2\pi}} dx} \le \epsilon. \qedhere
\]
\end{proof}
We also need the following concentration inequality for 1-dimensional Gaussian. 
\begin{proposition}
\label{prop:gaussian-concentration} \cite{vershynin2018high}
Suppose that $X\sim \mathcal{N}(0,\sigma^2)$ has Gaussian distribution with variance $\sigma^2$. Then, for every $t \geq 0$, 
\[
\text{Pr}\left( X\geq t\right) \leq e^{-\frac{t^2}{2\sigma^2}}.
\]
\end{proposition}

{
We will need the following properties of the multivariate Gaussian distribution:

\begin{proposition}
\label{prop:gaussian-multivariate}
    Let $\zeta \sim \mathcal{N}(\textbf{0}, \Sigma)$ be an $n$-dimensional multivariate Gaussian random vector.
    \begin{enumerate}
        \item For any fixed vector $\mathbf{v} \in \mathbb{R}^n$, the random variable $\langle \zeta, \mathbf{v} \rangle$ follows $\mathcal{N}(0, \mathbf{v}^T\Sigma \mathbf{v})$.
        \item For any two fixed vectors $\mathbf{v}, \mathbf{w} \in \mathbb{R}^n$, the covariance of the random variables $\langle \zeta, \mathbf{v} \rangle$ and $\langle \zeta, \mathbf{w} \rangle$ is equal to $\mathbf{v}^T \Sigma \mathbf{w}$.
    \end{enumerate}
\end{proposition}

\begin{proof}
    The random variable \( r = \langle \zeta, \mathbf{v} \rangle = \sum_{i \in [n]} v_i \zeta_i \) is a linear combination of Gaussian random variables, and therefore, follows a Gaussian distribution. We calculate its mean and variance as follows:
    \[
    \mathbb{E}[r] = \mathbb{E}\left[\sum_{i \in [n]} v_i \zeta_i\right] = \sum_{i \in [n]} v_i \mathbb{E}[\zeta_i] = 0.
    \]

    The variance of \( r \) is:
    \[
    \mathrm{Var}(r) = \mathbb{E}[r^2] = \mathbb{E}\left[\left(\sum_{i \in [n]} v_i \zeta_i\right)^2\right] = \sum_{i,j \in [n]} v_i v_j \mathbb{E}[\zeta_i \zeta_j].
    \]
    Since \( \mathbb{E}[\langle \zeta, \zeta^T] = \Sigma \), 
    \[
    \mathrm{Var}(r) = \sum_{i,j \in [n]} v_i v_j \Sigma_{ij} = \mathbf{v}^T \Sigma \mathbf{v}.
    \]

    Therefore, \( r \sim \mathcal{N}(0, \mathbf{v}^T \Sigma \mathbf{v}) \), which proves the first part.

    Let \( r' = \langle \zeta, \mathbf{w} \rangle = \sum_{i \in [n]} w_i \zeta_i \). The covariance of \( r \) and \( r' \) is:
    \[
    \mathrm{Cov}(r, r') = \mathbb{E}[rr'] = \mathbb{E}\left[\left(\sum_{i \in [n]} v_i \zeta_i\right)\cdot \left(\sum_{j \in [n]} w_j \zeta_j\right)\right].
    \]
    Expanding the product, we have:
    \[
    \mathbb{E}[rr'] = \sum_{i,j \in [n]} v_i w_j \mathbb{E}[\zeta_i \zeta_j] = \sum_{i,j \in [n]} v_i w_j \Sigma_{ij}.
    \]
    This simplifies to:
    \[
    \mathrm{Cov}(r, r') = \mathbf{v}^T \Sigma \mathbf{w}. \qedhere
    \] 
\end{proof}
}
\section{Overview of {T}echniques}
\label{sec:overview}

{Moving forward, unless otherwise specified, all promise templates $\Gamma$ are Boolean, and we consider the folded, idempotent variant of promise constraint satisfaction, i.e., $\fiPCSP(\Gamma)$.}

\subsection{Robust algorithm for MAJ polymorphisms}
\label{sec:maj-overview}
We obtain our robust algorithms by first solving the basic SDP relaxation and then \textit{rounding} the vectors. We first illustrate the SDP rounding idea with a warm-up algorithm (originally appeared in~\cite{Zwick98}) to solve the decision version of $2$-SAT CSP predicate $P=\{(-1,+1),(+1,-1),(+1,+1)\}$. Consider an instance $\Phi$ of {2-SAT (i.e., $\fiPCSP(P,P)$)} over a set of $n$ variables and $m$ constraints. We first solve the basic SDP relaxation of $\Phi$. If the basic SDP relaxation has a strictly positive error, then the instance is clearly not satisfiable. {For the moment}, suppose that the relaxation has zero error and we found a set of vectors $\textbf{v}_0, \textbf{v}_1, \ldots,\textbf{v}_n$, and the local probability variables $\lambda_j(f)$ for all $j \in [m], f: {S}_j \rightarrow {\B}$ that satisfy all the constraints in the basic SDP relaxation with $\eps_j=0$ for every $j \in [m]$. Consider an arbitrary constraint $C_j$ using the literals $x_1, x_2$. We abuse the notation and let $\textbf{v}_1, \textbf{v}_2$ denote the vectors assigned by the SDP to $x_1$ and $x_2$ respectively {(i.e., $\bv(x_1)$ and $\bv(x_2)$)}.
Using the first moment and second moment properties satisfied by the local probabilities $\lambda_j(f)$, we get the following. 
\begin{enumerate}
\item $\langle \textbf{v}_1, \textbf{v}_0 \rangle + \langle \textbf{v}_2, \textbf{v}_0 \rangle \geq 0.$
\item If $\langle \textbf{v}_1, \textbf{v}_0 \rangle + \langle \textbf{v}_2,\textbf{v}_0 \rangle=0$, $\textbf{v}_1+\textbf{v}_2 =0$.
\end{enumerate}

These properties motivate the following simple rounding algorithm that outputs the assignment $\sigma: V \rightarrow {\B}$. We sample a random vector $\zeta \sim \mathcal{N}(0,\textbf{I})$, and set
\[
\sigma(u_i)=\begin{cases}
+1, \text{ if }\langle \textbf{v}_i , \textbf{v}_0 \rangle > 0. \\ 
-1, \text{ if }\langle \textbf{v}_i , \textbf{v}_0 \rangle < 0. \\ 
+1, \text{ if }\langle \textbf{v}_i , \textbf{v}_0 \rangle = 0 \text{ and } \langle \zeta, \textbf{v}_i \rangle >0.\\ 
-1, \text{ if }\langle \textbf{v}_i , \textbf{v}_0 \rangle = 0 \text{ and } \langle \zeta, \textbf{v}_i \rangle <0.
\end{cases}
\]

This ensures that if a constraint uses the literals $x_1, x_2$, at least one of $x_1$ or $x_2$ is rounded to $+1$.

While this algorithm finds a satisfying solution when the underlying instance $\Phi$ has a solution satisfying all the constraints, it does not give any non-trivial guarantees when the instance is only promised to have a solution satisfying $1-\epsilon$ fraction of the constraints.
{By setting $\delta = \epsilon$ in \Cref{thm:sdp}, we get that the average error of the basic SDP relaxation is at most O($\epsilon$) on any instance that is promised to satisfy $1-\epsilon$ fraction of the constraints.}
Zwick~\cite{Zwick98} gave a robust algorithm {based on rounding the basic SDP relaxation} for $2$-SAT which was later improved by Charikar, Makarychev, and Makarychev~\cite{CharikarMM09}. They sample a random vector $\zeta \sim \mathcal{N}(0,\textbf{I})$, and set 
\[
\sigma(u_i) = \begin{cases}
+1, \text{ if }\langle \textbf{v}_i, \zeta \rangle \geq - \frac{\langle \textbf{v}_i, \textbf{v}_0 \rangle}{\sqrt{\epsilon}}. \\ 
-1, \text{ otherwise. }
\end{cases}
\]
One can view their algorithm as a smoothed version of the earlier discussed algorithm: { Using~\Cref{prop:gaussian-multivariate}, the random variable $r_i = \langle \textbf{v}_i, \zeta \rangle$ follows $ \mathcal{N}(0,1)$ for all $i \in [n]$. Thus, if $|\langle \textbf{v}_i, \textbf{v}_0 \rangle| > \lambda\sqrt{\epsilon}$ 
for a positive real number $\lambda$, then $\sigma(u_i)$ is equal to the sign of $\langle \textbf{v}_i, \textbf{v}_0 \rangle $ if $r_i \leq \lambda$, which occurs with probability at least $1 - e^{-\frac{\lambda^2}{2}}$(~\Cref{prop:gaussian-concentration}). }On the other hand, if $\langle \textbf{v}_i, \textbf{v}_0 \rangle =0 $, we set $\sigma(u_i)=+1$ if and only if $\langle \zeta, \textbf{v}_i \rangle \geq 0$.

We use this algorithm in our proof of~\Cref{thm:main-algorithm} but we give a completely different analysis. 
Consider a constraint ${C_j}$ in the $2$-SAT instance and let $\textbf{v}_1, \textbf{v}_2$ denote the vectors assigned by the basic SDP to the literals in a constraint.
Let ${r_1=\langle \textbf{v}_1, \zeta \rangle}$ and ${r_2=\langle \textbf{v}_2, \zeta \rangle}$. {From~\Cref{prop:gaussian-multivariate}, both the random variables $r_1$ and $r_2$} are standard Gaussian variables with covariance ${\text{Cov}(r_1,r_2)}=\langle \textbf{v}_1, \textbf{v}_2 \rangle$. To upper bound the probability that the output assignment $\sigma$ violates the constraint ${C_j}$,~\cite{CharikarMM09} calculate the probability that ${r_1} < - \frac{\langle \textbf{v}_i, \textbf{v}_0 \rangle}{\sqrt{\epsilon}}$ and ${r_2} < - \frac{\langle \textbf{v}_i, \textbf{v}_0 \rangle}{\sqrt{\epsilon}}.$ They do so by computing the probability that a pair of Gaussian random variables with known covariance lie in a given intersection of two half-spaces. While this analysis works for $2$-SAT, these calculations turn out to be significantly harder when there are more than two Gaussian random variables.

Instead, we take a conceptually different, and arguably simpler route. As a concrete example, consider the {promise template} $\Gamma=(P, Q)$ where $P=
\Ham_4\{2,3,4\}$, $Q=\Ham_4\{1,2,3,4\}$, i.e.,
\[
P=\bigl\{ \textbf{x} \in {\B}^4 : \sum_{i \in 4}x_i \geq 0\bigr\}, \quad Q = {\B}^k \setminus \{(-1,-1,-1,-1)\}.
\]
Consider a constraint ${C_j}$ of an instance $\Phi$ of ${\fiPCSP(\Gamma)}$. As the average error of SDP over all the constraints is at most $\epsilon$, for at least $1-\epsilon^{1/4}$ fraction of the constraints, the error is at most $\epsilon^{3/4}$. We restrict ourselves to these constraints and let $\textbf{v}_1,\ldots, \textbf{v}_4$ denote the vectors assigned by the SDP to the literals in the constraint. We have that $\sum_{i \in [4]}\langle \textbf{v}_i, \textbf{v}_0 \rangle\geq -\epsilon^{3/4}$. {Recall that if for some $i \in [4]$, $\langle \textbf{v}_i, \textbf{v}_0 \rangle  \geq C\sqrt{\eps}$ for a positive real number $C$, the corresponding literal is rounded to $+1$ with probability at least $1 - e^{-\frac{C^2}{2}}$. On the other hand, if for some $i \in [4]$, $\langle \textbf{v}_i, \textbf{v}_0 \rangle  \leq -4C\sqrt{\eps}$, there is some $i'\in [4]$ with $\langle \textbf{v}_{i'}, \textbf{v}_0 \rangle  \geq C\sqrt{\eps}$, which again ensures that there is at least one literal that is rounded to $+1$ with probability at least $1 - e^{-\frac{C^2}{2}}$. Hence, if there is at least one $i\in[4]$ with $|\langle \textbf{v}_{i}, \textbf{v}_0 \rangle|  \geq 4C\sqrt{\epsilon}$, at least one literal is rounded to $+1$ with probability at least $1 - e^{-\frac{C^2}{2}}$. }

Thus the interesting case is when $|\langle \textbf{v}_i,\textbf{v}_0 \rangle | \leq O(\sqrt{\epsilon})$ for every $i\in [4]$. In this case, using the first and second moment properties satisfied by these vectors, we get that $\norm{\sum_{i \in [4]} \textbf{v}_i}_2$ is at most $O(\epsilon^{1/4})$. 
The output assignment $\sigma$ violates $Q$ on this constraint only if ${\langle  \textbf{v}_i, \zeta \rangle} \leq -\frac{\langle \textbf{v}_i, \textbf{v}_0 \rangle}{\sqrt{\epsilon}}$ for every $i \in [4]$, or equivalently, ${\langle  \textbf{v}_i, \zeta \rangle} +\frac{\langle \textbf{v}_i, \textbf{v}_0 \rangle}{\sqrt{\epsilon}} \leq 0$ for every $i \in [4]$.
However, using $\norm{\sum_{i \in [4]} \textbf{v}_i}_2 \leq O(\epsilon^{1/4})$ and $\sum_{i \in [4]}\langle \textbf{v}_i, \textbf{v}_0 \rangle\geq -\epsilon^{3/4}$, we get that
\[
\sum_{i \in [4]} \left({\langle  \textbf{v}_i, \zeta \rangle} +\frac{\langle \textbf{v}_i, \textbf{v}_0 \rangle}{\sqrt{\epsilon}}\right) \geq -O(\epsilon^{1/4}).
\]
If $\sigma$ does not satisfy $Q$, then for some $i \in [4]$, we have that 
\[
{\langle  \textbf{v}_i, \zeta \rangle} +\frac{\langle \textbf{v}_i, \textbf{v}_0 \rangle}{\sqrt{\epsilon}} \in \left[ -O(\epsilon^{1/4}),0\right]
\]
Finally, we can upper bound the probability that this occurs to be at most $\tilde{O}\left(\epsilon^{1/4}\right)$ using anti-concentration of the Gaussian ${\langle  \textbf{v}_i, \zeta \rangle} \sim \mathcal{N}(0,1)$ {(\Cref{prop:gaussian-anticoncentration})}. Thus, we can infer that the assignment $\sigma$ satisfies at least $1-\tilde{O}(\epsilon^{1/4})$ fraction of the constraints in expectation. A careful analysis of the parameters gives a guarantee of $1-\tilde{O}(\epsilon^{1/3})$.

For an arbitrary {Boolean promise template} $\Gamma$ with $\MAJ\subseteq \Pol(\Gamma)$, we obtain a robust algorithm {for $\fiPCSP(\Gamma)$} by first reducing to the case when all the predicate pairs are of the form $(P,Q)$ with $Q={\B}^k \setminus \{(-1,-1,\ldots,-1)\}$, generalizing the above two examples of {the} $2$-SAT {template} {(i.e., $P=Q=\{(-1,+1),(1,-1),(1,1)\}$)} and $(\Ham_{4}\{2,3,4\}$, $\Ham_4\{1,2,3,4\})$. Then, we find a weight vector $\textbf{w}$ which satisfies that ${ \langle \textbf{w}, \textbf{x}\rangle }\geq 0$ for all $\textbf{x}\in P$ ($\textbf{w}=(1,1,\ldots,1)$ suffices for the previous two examples). We prove the existence of such a vector $\textbf{w}$ by using a Linear Programming relaxation, and we crucially use the fact that ${\MAJ \subseteq {\Pol}(P,Q)}$ in the analysis of this LP relaxation.
Once we find the vector $\textbf{w}$, the above analysis of $(\Ham_{4}\{2,3,4\}, \Ham_4\{1,2,3,4\})$ can be generalized, the main change being that we study the properties of the weighted sum of the $\textbf{v}_i$s with weights being given by the vector $\textbf{w}$.  

\subsection{Robust algorithm for AT polymorphisms}
\label{sec:at-overview}

For the Alternating-Threshold (AT) case, we combine these ideas with a random geometric sampling trick. 
As a concrete example, consider the {templates for $1$-in-$3$-SAT and NAE-SAT: $P = \{\bx \in \B^3 : x_1+x_2+x_3 = -1\}$ and $Q = \{\bx \in \B^3 : x_1+x_2+x_3 \in \{-1,1\}\}$, respectively. We have that $\AT \subseteq \Pol(P,Q)$~\cite{BrakensiekG21}.}
{Assume that an instance $(V,\mathcal C)$ of $\fiPCSP(P,Q)$ is perfectly satisfiable.} For {this} case, we can {efficiently find an exact solution} using the basic SDP relaxation via random hyperplane rounding as follows. Consider an arbitrary constraint ${C_j}$ and let $\textbf{v}_1, \textbf{v}_2, \textbf{v}_3$ denote the vectors assigned by the basic SDP to the literals in ${C_j}$. Using the fact that these vectors satisfy the first and second moment constraints of the basic SDP relaxation with zero error, we can infer that their sum $\textbf{v}_s = \textbf{v}_1 + \textbf{v}_2 + \textbf{v}_3$ is equal to $-\textbf{v}_0$ for every constraint $C$.  Let $\textbf{v}_i \perp \textbf{v}_0 = \textbf{v}_i - \langle \textbf{v}_i, \textbf{v}_0 \rangle \textbf{v}_0$ for $i \in [3]$. Note that $\sum_{i \in [3]} \textbf{v}_i \perp \textbf{v}_0 =0$. Using this observation, we can design a rounding scheme.
We first sample $\zeta \sim \mathcal{N}(0,\textbf{I})$, and set $\sigma(u_i)=+1$ if $\langle \textbf{v}_i \perp \textbf{v}_0 , \zeta \rangle > 0$, and $-1$ if $\langle \textbf{v}_i \perp \textbf{v}_0 , \zeta \rangle < 0$. {Each variable $u_i$ is assigned a value with probability $1$ unless $\bv_i \in \{\bv_0, -\bv_0\}$ since $\bv_i$ is a unit vector. In that undetermined case, we set $\sigma(u_i) = +1$ if $\bv_i = \bv_0$ and $\sigma(u_i) = -1$ otherwise.\footnote{{Observe this choice of assignment is consistent with the default rounding of the constant literals $\pm 1$.}}} As $\textbf{v}_1+\textbf{v}_2+\textbf{v}_3=-\textbf{v}_0$, the rounding scheme  ensures that at least one literal associated with these vectors is set to $+1$, and at least one literal is set to $-1$.

For the robust setting where we are only guaranteed that there is {an assignment strongly} satisfying $1-\epsilon$ fraction of the constraints. {By setting $\delta = \epsilon$ in \Cref{thm:sdp},} we get that the average SDP error is at most {O($\epsilon$)}. By Markov's inequality, we are guaranteed that for at least $1-\sqrt{\epsilon}$ fraction of the constraints, the SDP error is at most {O($\sqrt{\epsilon}$)}. For these constraints, we get that the sum vector $\textbf{v}_s$'s component orthogonal to $\textbf{v}_0$ has $\ell_2$ norm at most $O(\eps^{1/4})$, i.e., $\norm{\sum_{i \in 3}\textbf{v}_i \perp \textbf{v}_0}_2 \leq O(\epsilon^{1/4})$.
Using this, we design a rounding scheme that is similar to the above, with the addition that when $\norm{\textbf{v}_i \perp \textbf{v}_0}_2$ is very small, we want to round it to $+1$ or $-1$ depending on its component along $\textbf{v}_0$, similar to how we were rounding $\textbf{v}_0$ to $+1$ and $-\textbf{v}_0$ to $-1$ in the exact algorithm earlier. We have the following compact algorithm based on this idea.
 \[
        \sigma(u_i) = 
        \begin{cases}
        +1, \text{ if  }\langle \textbf{v}_i \perp \textbf{v}_0, \zeta \rangle \ge -\delta \langle \textbf{v}_i, \textbf{v}_0 \rangle. \\ 
        -1, \text{ otherwise.}
        \end{cases}
        \]

{
We choose \( \delta \) to be \( \epsilon^\kappa \), where \( \kappa \) is a constant such that \( \kappa < \frac{1}{4} \), to be set later. Note that when \( |\langle \textbf{v}_i \perp \textbf{v}_0, \zeta \rangle| > \delta \), our new rounding scheme is equivalent to the earlier exact algorithm, i.e., \( \sigma(u_i) = +1 \) if \( \langle \textbf{v}_i \perp \textbf{v}_0, \zeta \rangle > 0 \), and \( -1 \) otherwise. 

If \( |\langle \textbf{v}_i \perp \textbf{v}_0, \zeta \rangle| > 3\delta \) for some \( i \in [3] \), the condition \( \norm{\sum_{i \in [3]} \textbf{v}_i \perp \textbf{v}_0}_2 \leq O(\epsilon^{1/4}) < \delta \) implies that there exist distinct \( i, i' \in [3] \) such that \( \langle \textbf{v}_i \perp \textbf{v}_0, \zeta \rangle > \delta \) and \( \langle \textbf{v}_{i'} \perp \textbf{v}_0, \zeta \rangle < -\delta \). In this case, \( \sigma(u_i) = +1 \) and \( \sigma(u_{i'}) = -1 \), ensuring that the output assignment satisfies the NAE-\( 3 \)-SAT constraint. On the other hand, when \( |\langle \textbf{v}_i \perp \textbf{v}_0, \zeta \rangle| \) is much smaller than \( \delta \) for every \( i \in [3] \), our rounding function sets \( \sigma(u_i) = +1 \) if \( \langle \textbf{v}_i, \textbf{v}_0 \rangle > 0 \) and \( -1 \) otherwise. Since \( \sum_{i \in [3]} \textbf{v}_i \) is close to \( -\textbf{v}_0 \), the output satisfies the NAE-\( 3 \)-SAT constraint even in this case.

Finally, we employ a geometric sampling trick where \( \delta \) is sampled uniformly at random from a geometric series (by choosing \( \kappa \) from an arithmetic series) to ensure that, with high probability, either \( |\langle \textbf{v}_i \perp \textbf{v}_0, \zeta \rangle| > 3\delta \) for some \( i \in [3] \), or \( |\langle \textbf{v}_i \perp \textbf{v}_0, \zeta \rangle| \) is much smaller than \( \delta \) for every \( i \in [3] \).

}

For an arbitrary Boolean {promise template} $\Gamma$ with $\AT \subseteq \Pol(\Gamma)$, we first show that {$\fiPCSP(\Gamma)$ reduces to $\fiPCSP(P,Q)$ where} 
\[
P=\{ \textbf{v} \in {\B}^k : \langle \textbf{v}, \textbf{w}\rangle = b \}, \quad Q = {\B}^k \setminus\{\textbf{x},-\textbf{x}\}
\]
for some $\textbf{w} \in \R^k, b \in \R$ with $w_i \neq 0\,~\forall i \in [k]$, where we set $\textbf{x}$ as $x_i = +1$ if $w_i>0$, and $-1$ otherwise.
This generalizes {our previous $1$-in-$3$-SAT versus NAE-$3$-SAT template} which corresponds to the case when $\textbf{w}=(1,1,1),b=-1$. We show that the algorithm above works for this general predicate pair as well, thereby obtaining {for every promise template $\Gamma$ with $\AT \subseteq \Pol(\Gamma)$} a robust algorithm for {$\fiPCSP(\Gamma)$.}

\subsection{UG-hardness of robust algorithms}
\label{sec:sdpgap-overview}

As mentioned earlier, our Unique Games hardness for {$\fiPCSP(\Gamma)$} (\Cref{thm:main-hardness}) is based on an integrality gap for the basic SDP relaxation, i.e., we need to show that there is a finite instance $\Phi$ of ${\fiPCSP(\Gamma)}$ that has SDP error of zero, yet there is no assignment that weakly satisfies all the constraints.  
In pursuit of this goal, we develop a general recipe for showing integrality gaps with respect to basic SDP for Promise CSPs via colorings of the $n$-dimensional unit sphere $\S^n$. 

{We {focus on Boolean promise templates which are symmetric.} We say that a predicate $P \subseteq {\B}^k$ is symmetric if for every $\textbf{x}, \textbf{y}\in {\B}^k$ such that $\textsf{hw}(\textbf{x})=\textsf{hw}(\textbf{y})$, we have $\textbf{x} \in P$ if and only if $\textbf{y} \in P$}.

We first start by showing an integrality gap instance for $3$-LIN. Recall that the $3$-LIN CSP has the predicate $P$ with 
\begin{align}
P=\{ \textbf{x} \in {\B}^3 : x_1 + x_2 + x_3 = -1 \text{ or }x_1 + x_2 + x_3=+3\}.\label{eq:3-lin}
\end{align}
Consider the instance $\Phi$ {of $\fiPCSP(P,P)$} that uses three variables and uses two constraints $C_1, C_2$ with $C_1=\{(x_1,x_2,x_3)\},C_2=\{ (\overline{x_1}, \overline{x_2},\overline{x_3})\}$. The instance has no assignment that satisfies both the constraints, and we now show that the basic SDP solution has zero error on $\Phi$. We first describe the local probability variables: {for each $j \in [2]$ and each assignment $f : S_j \to \B$, where $S_j = \{u_1,u_2,u_3\}$,} we set 
\[
\lambda_j(f)={\begin{cases}
1/4 & f(C_j) \in P\\
0 & \text{otherwise.}
\end{cases}}
\]
That is, we set each local distribution to be the uniform distribution over {satisfying assignments}. Substituting these in the first and second moment constraints, we get the following requirements that the vectors $\textbf{v}_1, \textbf{v}_2, \textbf{v}_3$ need to satisfy. 
\begin{align}
\label{eq:3lin-1}     { \langle \textbf{v}_i , \textbf{v}_0 \rangle} &= 0 \quad \forall i \in [3] \\ 
    \label{eq:3lin-2} { \langle \textbf{v}_i , \textbf{v}_{i'} \rangle } &= 0\quad \forall i \neq i' \in [3]
\end{align}
We can find such three vectors by picking three orthogonal vectors that are all orthogonal to $\textbf{v}_0$. This shows that there is a solution to the basic SDP relaxation of $\Phi$ that has zero error, thus finishing the proof of the existence of an integrality gap for $\Phi$.

While the simple example gives an integrality gap for {3-LIN}, it is a challenging task to find such explicit integrality gap instances for general predicates. 
We develop a non-explicit approach where {for a given promise template $\Gamma$} we first construct an infinite integrality gap instance $\mathcal{I}^n(\Gamma)$ {for $\fiPCSP(\Gamma)$ and} then use it to show the existence of a finite integrality gap instance.
The variable set $V$ of $\mathcal{I}^n(\Gamma)$ to be the set of unit vectors in $\R^{n+1}$: $V=\{ u_{\textbf{v}} : \textbf{v} \in \S^{n} \}$. {We fix an arbitrary vector to be assigned $\textbf{v}_0$. We let $\mathcal L = V \cup \overline{V} \cup \B$ be our set of literals, and let $\hat{\bv} : \mathcal L \to 
\S^n$ be the map}
\[
    {\hat{\bv}(x) = \begin{cases}\bv & x = u_{\bv}\\
    -\bv & x = \overline{u_{\bv}}\\
    x \cdot \bv_0 & x \in \B.\end{cases}}
\]

{We now describe the clauses $\mathcal C$ of our integrality gap instance. Consider $(P,Q) \in \Gamma$ of arity $k$. For a potential clause $C = (x_1, \hdots, x_k) \in \mathcal L^k$, let $S$ be the set of variables making up $C$. We add $C$ to $\mathcal C$ if there exists a probability distribution $\lambda$ over assignments $f : S(C) \to \B$ satisfying the following constraints.}
{\begin{align*}
 \sum_{f:{S} \rightarrow {\B}, f(C)\in P} \lambda_j(f) &= 1\\
\text{ (First moments.)  } \qquad { \langle \hat{\textbf{v}}(x) , \textbf{v}_0 \rangle }&= \sum_{f:{S} \rightarrow {\B}}\lambda(f) f(x) & ~x \in C\\
\text{(Second moments.)  } \quad { \langle \hat{\textbf{v}}(x) , \hat{\textbf{v}}(x') \rangle } &= \sum_{f:{S} \rightarrow {\B}}\lambda(f) f(x)f(x') & x,x' \in C.
\end{align*}} 
We refer to {$\hat{\bv}(C) := (\hat{\bv}(x_1), \hdots, \hat{\bv}(x_k))$ satisfying the above as a \emph{$P$-configuration}} with respect to $\textbf{v}_0$. 
The way we have added the constraints ensures that setting $\textbf{a}$ to the variable $u_{\textbf{a}}$ satisfies all the basic SDP constraints with zero error. We then show that the instance $\mathcal{I}^n(\Gamma)$ does not have any assignment $\sigma$ that weakly satisfies all the constraints to obtain the integrality gap. Towards this, we study the sphere colorings $f_n : \S^{n} \rightarrow \B$ that weakly satisfy all the constraints of $\mathcal{I}^n(\Gamma)$. 

{We now return to examining $3$-LIN} in terms of $\mathcal{I}^n(P)$ {for $P$ as in (\ref{eq:3-lin})}. As mentioned earlier, a set of three vectors $\textbf{v}_1, \textbf{v}_2, \textbf{v}_3$ are a $P$-configuration with respect to $\textbf{v}_0$ if they satisfy~\Cref{eq:3lin-1} and~\Cref{eq:3lin-2}. 
Thus, to show that the instance $\mathcal{I}^n(P)$ does not have any assignment satisfying all the constraints, it suffices to show that for some positive integer $n$, there is no function $f:\mathbb{S}^n \rightarrow {\B}$ that satisfies the following condition:
For all vectors $\textbf{v}_1, \textbf{v}_2, \textbf{v}_3 \in \mathbb{S}^n$ are mutually orthogonal and are orthogonal to $\textbf{v}_0$, we have 
\[
f(\textbf{v}_1)+f(\textbf{v}_2)+f(\textbf{v}_3) \in \{-1,+3\} \ . 
\]
{Since we allow negations of variables in our instance}, we also require such a function $f$ to be folded, i.e., $f(-\textbf{v})=-f(\textbf{v})$ for every $\textbf{v}\in \mathbb{S}^{{n}}$. Such a coloring $f$ trivially does not exist: consider a set of three mutually orthogonal vectors that are all orthogonal to $\textbf{v}_0$, $V=(\textbf{v}_1, \textbf{v}_2, \textbf{v}_3)$ and their negations, $V'=(-\textbf{v}_1, -\textbf{v}_2, -\textbf{v}_3)$. Such a set of vectors is guaranteed to exist if $n \geq 4$.
Note that both these are valid $P$-configurations, but at least one of $f(\textbf{v}_1)+f(\textbf{v}_2)+ f(\textbf{v}_3)$, $f(-\textbf{v}_1)+ f(-\textbf{v}_2)+f(-\textbf{v}_3)$ does not belong to $\{-1,+3\}$, thus completing the proof that there is no assignment satisfying all the constraints of $\mathcal{I}^n(P)$. Hence, $\mathcal{I}^n(P)$ is an integrality gap instance for 3-LIN, i.e., ${\fiPCSP(P,P)}$, and this implies the existence of a finite integrality gap instance as well. While there is a direct integrality gap instance for the $3$-LIN, for an arbitrary {promise template} $\Gamma$, to show the existence of an integrality gap for the basic SDP relaxation, it is more convenient and practical to show the absence of any assignment $f:\S^{n}\rightarrow {\B}$ satisfying all the constraints of $\mathcal{I}^n(\Gamma)$ for some $n$ which then implies the existence of a finite integrality gap instance by a compactness argument.

While the $P$-configurations in the above proof for $3$-LIN are easy to study, in general, proving the absence of sphere coloring is challenging. For example, consider the {promise template} $\Gamma=(P,Q)$ where $P=\Ham_5 \{2,5\}$, $Q=\Ham_5 \{1,2,3,4,5\}$:
\[
P = \bigl\{ \textbf{x} \in {\B}^5 : |\{ i \in [5] : x_i = +1\} | \in \{2,5\} \bigr\}, \quad Q =\bigl\{ \textbf{x} \in  {\B}^5 : \sum_{i=1}^5 x_i {\neq}-5\bigr\}. 
\]
Here, a set of $P$-configurations with respect to a vector $\textbf{v}_0$ are five unit vectors $\{\textbf{v}_1, \textbf{v}_2, \ldots, \textbf{v}_5\}$
such that every two distinct vectors have an inner product equal to ${-\frac{1}{5}}$, and ${ \langle \textbf{v}_i , \textbf{v}_0 \rangle} = 0$ for every $i \in [5]$. The sphere coloring problem is then to show that there exists $n$ such that for any folded $f:\mathbb{S}^n \rightarrow {\B}$, there exists a set of five vectors in $\mathbb{S}^n$ with every pair of them having inner product equal to ${-\frac{1}{5}}$ that are all colored $-1$. 

Such problems where the goal is to find a monochromatic structure in sphere colorings are studied in a topic called \emph{sphere Ramsey theory.} In a striking result using tools from combinatorics, linear algebra, and Banach space theory, Matoušek and Rödl~\cite{matouvsek1995ramsey} proved that every set of affinely independent vectors $V$ whose circumradius is smaller than $1$ is sphere Ramsey---i.e., for every $r$, there exists $n$ large enough such that every $r$-coloring of $\mathbb{S}^n$ must have a monochromatic set $U$ that is congruent to $V$. This directly answers the above question regarding sphere coloring of {$\mathcal{I}^n(P,Q)$} where $P=\Ham_5 \{2,5\}$, $Q=\Ham_5 \{1,2,3,4,5\}$.

For an arbitrary Boolean symmetric {promise template} $\Gamma=(P,Q)$, to prove~\Cref{thm:main-hardness}, we first reduce the problem into a fixed number of templates using the properties of $\AT$ and $\MAJ$ polymorphisms, following~\cite{BrakensiekG21}. For some of these templates, the result of Matoušek and Rödl~\cite{matouvsek1995ramsey} directly answers the sphere coloring problem associated with them. For others, we need extra work built on the sphere Ramsey result. We highlight one such {promise} template $\Gamma =  (P,Q)$ where $P=\Ham_k \{1,k\}$, $Q=\Ham_k \{0,1,\ldots,k\}\setminus\{b\}$ for positive integers $k, b :0\leq b \leq k$. In the sphere coloring problem associated with this template, we need to show that for any real $\alpha \in [0,1]$, in any folded $f:\S^n \rightarrow {\B}$, there are $k$ vectors all of whose pairwise inner products are equal to $\alpha$, and exactly $b$ of them are assigned $+1$ according to $f$. We refer to such a set of vectors $S=\{\textbf{v}_1,\ldots,\textbf{v}_k\}\subseteq \S^n$ as $\alpha$-configuration if ${ \langle \textbf{v}_i , \textbf{v}_j \rangle }=\alpha$ for all $i \neq j \in [k]$.

The sphere Ramsey result shows that there is an $\alpha$-configuration with $b$ vectors being colored $+1$ only when $b=0$ or $b=k$. To extend to general $b$, we prove the following connectivity lemma: Between any two arbitrary $\alpha$-configurations $S, T$, there is a path $U_1, U_2,\ldots, U_L$ of $\alpha$-configurations with length $L:=L(\alpha,k,n)$ where $U_1=S$, $U_L=T$, and any two consecutive configurations $U_i$, $U_{i+1}$ differ in at most one element. The sphere Ramsey result shows that there is an $\alpha$-configuration with all the vectors being colored $+1$, and by negating these vectors, we get an $\alpha$-configuration with all the vectors being colored $-1$. Finally, using the connectivity lemma between these two configurations, we get that for every $b \in \{0,1,\ldots,k\}$, there is an $\alpha$-configuration where in exactly $b$ vectors are colored $+1$, thereby finishing the proof of the integrality gap. We also remark that while our integrality proofs are non-constructive in general, we get an explicit integrality gap instance for ${\fiPCSP(\Gamma)}$ inspired by the connectivity lemma.

\section{Robust Algorithms}
\label{sec:alg}

\subsection{CMM is a robust algorithm when {MAJ} is a polymorphism}

We restate~\Cref{thm:main-algorithm} for the case of $\MAJ$ polymorphisms. 
\begin{theorem}
\label{thm:maj-alg}
Let $\Gamma=\{(P_1,Q_1), (P_2, Q_2), \ldots, (P_l, Q_l)\}$ be a Boolean {promise template} with $\MAJ \subseteq \Pol(\Gamma)$. For every $\epsilon>0$, there is a randomized polynomial time algorithm that given an instance $\Phi$ of ${\fiPCSP}(\Gamma)$ that is promised to have an assignment satisfying $1-\epsilon$ fraction of the constraints, finds an assignment to $\Phi$ that satisfies $1-\tilde{O}_{\Gamma}(\epsilon^{\frac{1}{3}})$ fraction\footnote{We use $O_\Gamma$ to denote a hidden constant which depends on the specific template $\Gamma$.} of the constraints in expectation.
\end{theorem}

In the rest of this subsection, we prove~\Cref{thm:maj-alg}. Our strategy is to reduce the problem into a special case {in which} every predicate pair {of} $\Gamma$ is of the form $(P,{\B}^k \setminus \{(-1,-1,\ldots, -1)\})$, and then {solve the robust version of $\fiPCSP(\Gamma)$ using} the algorithm of Charikar, Makarychev, and Makarychev~\cite{CharikarMM09}.

For ease of notation, we use $O(\cdot)$ instead of $O_{\Gamma}(\cdot)$ when $\Gamma$ is clear from the context. 
We first {remove from $\Gamma$ all predicate pairs $(P,Q)$ for which $Q=\B^k$} for some integer $k$ since these constraints are trivially satisfied by any assignment. 
{To justify why we can remove these predicate pairs in the robust setting, let $\Phi$ be our instance of $\fiPCSP(\Gamma)$.} Suppose that there are $m$ constraints in $\Phi$, and $m'=\alpha m$ of them use predicates of the form $(P,Q)$ where $P\subseteq Q={\B}^k$. We consider the instance $\Phi'$ containing $m-m'$ constraints obtained by deleting the constraints that use predicates of the form $(P,Q)$ where $P\subseteq Q={\B}^k$. In the instance $\Phi'$, we are promised that there is a solution satisfying $m-m'-\epsilon m$ constraints, i.e., $1-\frac{\epsilon}{1-\alpha}$ fraction of the constraints. 
We use the algorithm that we will present later in the subsection on the instance $\Phi'$ to get an assignment weakly violating at most $\tilde{O}\left( \left(\frac{\epsilon}{1-\alpha}\right)^{\frac{1}{3}}\right)(m-m')$ constraints. The same assignment weakly violates at most 
\[
\tilde{O}\left( \left(\frac{\epsilon}{1-\alpha}\right)^{\frac{1}{3}}\right)(m-m') = \tilde{O}\left( \left( \frac{\epsilon}{1-\alpha}\right)^{\frac{1}{3}}\right) (1-\alpha)m \leq \tilde{O}( \epsilon^{\frac{1}{3}}) m 
\]
constraints in $\Phi$. 
Thus, it suffices to study Boolean {promise templates} where no predicate pair is of the form $(P, {\B}^k)$. 

We further transform the instance into one in which every predicate pair is of the form $(P,{Q)}$ {with $Q = {\B}^k\setminus \{(-1,-1,\ldots,-1)\})$.}  

\begin{lemma}
\label{lem:maj-transform}
Consider a Boolean {promise template} $\Gamma=\{(P_1,Q_1), \ldots, (P_l, Q_l)\}$ where $P_i\subseteq  Q_i \subsetneq {\B}^{k_i}$ for every $i \in [l]$. {There exits a promise template $\Gamma'=\{(P'_1,Q'_1), \ldots, (P'_{l'}, Q'_{l'})\}$ and a polynomial time algorithm such that given an instance $\Phi$ of $\fiPCSP(\Gamma)$ over a set of variables $V$, the algorithm} outputs an instance $\Phi'$ of {$\fiPCSP(\Gamma')$} such that the following hold. 
\begin{enumerate}
    \item (Completeness.) {For every $\epsilon>0$,} if an assignment $\sigma : V \rightarrow {\B}$ strongly satisfies $1-\epsilon$ fraction of the constraints in $\Phi$, then $\sigma$ strongly satisfies at least {$1-O_{\Gamma}(\epsilon)$} fraction of the constraints in $\Phi'$ as well.
    \item (Soundness.) {For every $\epsilon>0$, }if an assignment $\sigma : V \rightarrow {\B}$ weakly satisfies $1-\epsilon$ fraction of the constraints in $\Phi'$, then $\sigma$ weakly satisfies at least {$1-O_{\Gamma}(\epsilon)$} fraction of constraints in $\Phi$. 
    \item The resulting {promise template} $\Gamma'$ satisfies the below two properties:
    \begin{enumerate}
        \item For every $i \in [l']$, $Q'_i$ is equal to ${\B}^{k'_i}\setminus \{(-1,-1,\ldots,-1)\}$ for some positive integer $k'_i$.
        \item If $\MAJ \subseteq \Pol(\Gamma)$, then, $\MAJ \subseteq \Pol(\Gamma')$.
    \end{enumerate} 
\end{enumerate}
\end{lemma}

\begin{proof}
We obtain the above transformation in two steps. First, we construct a Boolean {promise template} $\Gamma^*$ from $\Gamma$ as follows: 
\[
\Gamma^* :=  \left\{ \left(P_i, {\B}^{k_i}\setminus \{\textbf{{b}}\}\right) : i \in [l], \textbf{{b}} \in {\B}^{k_i} \setminus Q_i \right\}
\]
Note that for every predicate pair $(P,Q) \in \Gamma^*$, there is a predicate pair $(P,Q') \in \Gamma$ with $Q' \subseteq Q$, and thus, $\Pol(\Gamma) \subseteq \Pol(\Gamma^*)$. 

Given an instance $\Phi$ of ${\fiPCSP(\Gamma)}$ over a set of variables $V$, we obtain an instance $\Phi^*$ of ${\fiPCSP(\Gamma^*)}$ over the same set of variables $V$ as follows. We order the constraints of $\Phi$ as $C_1, C_2, \ldots, C_m$. Consider a constraint $C_j {= (x_{j,1}, x_{j,2}, \ldots, x_{j,k_i})}$ in $\Phi$ using the predicate pair $(P_i,Q_i)${.} 
In the instance $\Phi^*$, we add $2^{k_i}-|Q_i|$ constraints $C_{j,\textbf{{b}}}$ associated with every $\textbf{{b}} \in {\B}^{k_i} \setminus Q_i$. The constraint $C_{j,\textbf{{b}}}$ uses the same tuple of literals as $C_j$ but uses the predicate pair $(P_i, {\B}^k_i \setminus \{\textbf{{b}}\}).$ 
We analyze the completeness and soundness of this reduction. 
\begin{enumerate}
    \item (Completeness.) If $\sigma : V \rightarrow {\B}$ strongly satisfies a constraint $C_j$, then $\sigma$ strongly satisfies $C_{j,\textbf{{b}}}$ for every $\textbf{{b}} \in {\B}^{k_i}\setminus {Q_i}$. 
    If $\sigma$ strongly satisfies $(1-\epsilon)m$ constraints in $\Phi$, then the number of constraints that $\sigma$ does not strongly satisfy in $\Phi^*$ is at most $2^K\epsilon m$, where $K := \max_{i \in [l]}k_i$. 
    Thus, $\sigma$ strongly satisfies at least $m^*-2^K\epsilon m \geq (1-2^K\epsilon)m^*$ constraints, where $m^*$ is the number of constraints in $\Phi^*$.
    Here, we are using the fact that $m^* \geq m$, as we have $Q_i \neq {\B}^{k_i}$ for every $i \in [k]$.
    \item (Soundness.) Suppose that $\sigma : V \rightarrow {\B}$ weakly satisfies $1-\epsilon$ fraction of the constraints in $\Phi^*$. As $\sigma$ weakly violates at most $\epsilon m^* \leq 2^K \epsilon m$ constraints, for at least $1-2^K\epsilon$ fraction of the original constraints $C_j$, $\sigma$ weakly satisfies $C_{j,\textbf{{b}}}$ for every $\textbf{{b}} \in {\B}^{k_i}\setminus Q_i$. For these constraints $C_j$, $\sigma$ weakly satisfies $C_j$ as well, and thus, $\sigma$ weakly satisfies at least $1-2^K\epsilon$ fraction of the constraints in $\Phi$.
\end{enumerate}

Next, we transform $\Gamma^*$ to $\Gamma'$ to ensure that every weak predicate is of the form ${\B}^k\setminus\{(-1,\ldots,-1)\}$. 
We use the following entry-wise product notation: for a pair of vectors $\textbf{u},\textbf{v}\in\R^k$, we let
    \[
    \textbf{u}\odot\textbf{v} := (u_1v_1, u_2v_2, \ldots,u_kv_k).
    \]
    For a predicate $P \subseteq {\B}^k$, we let 
    \[
    P \odot \textbf{v} := \{ \textbf{u} \odot \textbf{v} : \textbf{u} \in P \}.
    \]
We let $\Gamma'$ be the following. 
\[
\Gamma' := \left\{ \left(P \odot (-\textbf{{b}}), {\B}^k \setminus \{(-1,-1,\ldots,-1)\}\right) : \left(P,{\B}^k\setminus \{\textbf{{b}}\}\right) \in \Gamma^*\right\}
\]
Given the instance $\Phi^*$ of ${\fiPCSP(\Gamma^*)}$ that is obtained from $\Phi$,  we construct an instance $\Phi'$ of ${\fiPCSP(\Gamma')}$ on the same set of variables as follows. Consider a constraint $C_{j,\textbf{{b}}}$ using the predicate pair $(P_i, {\B}^{k_i}\setminus \{\textbf{{b}}\})$. We have a constraint $C'_{j,\textbf{{b}}}$ in $\Phi'$ using the predicate pair $(P_i \odot \textbf{{b}}, {\B}^{k_i}\setminus \{(-1,-1,\ldots,-1)\}$. However, in the constraint $C'_{j,\textbf{{b}}}$, we negate the literals corresponding to the indices $p\in[k_i]$ where ${b}_p=1$. More formally, {we define the} constraint $C'_{j,\textbf{{b}}}{= (y_{j,1},y_{j,2},\ldots,y_{j,k_i})}$ {as follows.}
\[
y_{j,p} := \begin{cases}
x_{j,p} , \text{ if } {b}_p = -1. \\
\overline{x_{j,p}} , \text{ if } {b}_p = 1.
\end{cases}
\]
{For the constant literals, $\overline{x_{j,p}}$ refers to a change of sign.} An assignment $\sigma : V \rightarrow {\B}$ strongly (and {respectively} weakly) satisfies $C'_{j,\textbf{{b}}}$ if and only if $\sigma$ strongly (and {respectively} weakly) satisfies $C_{j,\textbf{{b}}}$. Thus, we get that $\Phi'$ satisfies the completeness and soundness properties of the lemma. Finally, the operation of negating a subset of coordinates preserves all polymorphisms that are folded, and thus, if $\MAJ \subseteq \Pol(\Gamma)$, $\MAJ \subseteq \Pol(\Gamma')$ as well. 
\end{proof}

Given an instance $\Phi$ of ${\fiPCSP(\Gamma)}$, we transform it to an instance $\Phi'$ of ${\fiPCSP(\Gamma')}$ using~\Cref{lem:maj-transform}. 
If $\Phi$ is promised to have an assignment strongly satisfying at least $1-\epsilon$ fraction of the constraints, then $\Phi'$ has an assignment strongly satisfying $1-O(\epsilon)$ fraction of the constraints as well. If there is a polynomial time robust algorithm that outputs an assignment weakly satisfying $1-f(\epsilon)$ fraction of the constraints, then we can use this assignment to obtain a robust algorithm for ${\fiPCSP(\Gamma)}$ as well. Thus, a polynomial time robust algorithm for ${\fiPCSP(\Gamma')}$ gives a polynomial time robust algorithm for ${\fiPCSP(\Gamma)}$ as well. 

{Given a Boolean promise template} $\Gamma$ {in which} every predicate pair is of the form $(P,$\linebreak ${\B}^{k}\setminus \{(-1,-1,\ldots,-1)\})$ with $\MAJ \subseteq \Pol(\Gamma)$, we show that the robust algorithm of Charikar, Makarychev, and Makarychev~\cite{CharikarMM09} for $2$-SAT generalizes and gives a robust algorithm for the Boolean folded PCSP $\Gamma$ as well.   
First, we state their algorithm.

\smallskip 
\noindent\fbox{%
    \parbox{\textwidth}{%
    \begin{enumerate}
        \item Given an instance $\Phi$ of ${\fiPCSP(\Gamma)}$ containing $n$ variables $u_1,u_2,\ldots,u_n$,
        solve the {Boolean folded idempotent} basic SDP and obtain a set of vectors $\textbf{v}_0, \textbf{v}_1, \ldots, \textbf{v}_n$ {and objective value $\eps > 0$}. Let $\boldsymbol \mu \in \R^n$ denote the first moments and $\Sigma \in \R^{n\times n}$ be the {Gram} matrix of these vectors. 
        \begin{align*}
            \mu_i &= { \langle \textbf{v}_i , \textbf{v}_0 \rangle } \quad \forall i \in [n]\\ 
            \Sigma_{i,j} &= { \langle \textbf{v}_i , \textbf{v}_j \rangle } \quad \forall i,j \in [n]
        \end{align*}
        \item Sample an $n$ dimensional Gaussian $\zeta \sim \mathcal N(\textbf{0}, \Sigma)$. Note that $\Sigma$ is positive semidefinite. {We can also equivalently sample $\zeta' \sim \mathcal{N}(\textbf{0}, \textbf{I})$ and use $\langle \textbf{v}_i, \zeta' \rangle$ in the place of $\zeta_i$ in step 4 below. We stick with $\zeta \sim \mathcal{N}(\textbf{0},\Sigma)$ for ease of notation in the analysis.}
        \item Set\footnotemark  $\gamma = \epsilon^{\frac{2}{3}}$.
        \item For each $i \in [n]$, round as follows
        \[
            \sigma(u_i) = \begin{cases}
            +1 & \zeta_i \ge -\mu_i/\gamma.\\
            -1 & $\text{otherwise}$.
            \end{cases}
        \]
    \end{enumerate}
            }%
}        \footnotetext{We change the parameter slightly -- in the original CMM algorithm, $\gamma$ is set to be $\sqrt{\epsilon}$.}
\smallskip 

We shall prove the following guarantee about the algorithm.

\begin{theorem}\label{thm:maj}
Let $\Gamma=\{(P_1,Q_1),(P_2,Q_2),\ldots,(P_l,Q_l)\}$ be a Boolean {promise template} such that $\MAJ \subseteq \Pol(\Gamma)$ where $P_i \subseteq Q_i = {\B}^{k_i}\setminus \{(-1,-1,\ldots,-1)\} $ for every $i \in [l]$. Let $\Phi$ be an instance of ${\fiPCSP}(\Gamma)$ over $n$ variables and using $m$ constraints for which the basic SDP relaxation has a solution with error value at most $\eps m$. Then, the assignment $\sigma$ output by the above CMM algorithm weakly satisfies $1 - \tilde{O}_{\Gamma}(\eps^{1/3})$ fraction of the constraints in expectation.
\end{theorem}

{\Cref{thm:maj}, together with \Cref{lem:maj-transform} and \Cref{thm:sdp}, completes the proof of \Cref{thm:maj-alg}. Suppose that $\Gamma = \{(P_1, Q_1), (P_2, Q_2), \ldots, (P_l, Q_l)\}$ is a Boolean {promise template} with $\MAJ \subseteq \Pol(\Gamma)$ with $Q_i = {\B}^{k_i} \setminus \{(-1, -1, \ldots, -1)\}$ for every $i \in [l]$. Given $\epsilon > 0$ and an instance $\Phi$ of ${\fiPCSP}(\Gamma)$ that is promised to have an assignment satisfying $1 - \epsilon$ fraction of the constraints, we solve the basic SDP relaxation of the instance. Without loss of generality, we assume that $\epsilon > \frac{1}{m}$, where $m$ is the number of constraints in the instance. By setting $\delta = \epsilon$ in \Cref{thm:sdp}, we obtain that the basic SDP relaxation has an error of at most $2 \epsilon m$. Finally, \Cref{thm:maj} implies that the output assignment of the above algorithm weakly satisfies $1 - \tilde{O}_{\Gamma}(\epsilon^{1/3})$ of the constraints in expectation.}

{Towards proving~\Cref{thm:maj},} we analyze the probability that the output assignment weakly satisfies each constraint separately. Fix a constraint $C_j$ using the predicate pair $(P,Q)$ of arity $k$ with $P\subseteq Q={\B}^k \setminus \{(-1,-1,\ldots,-1)\}.$ Suppose that the basic SDP solution has error equal to $c$ on this constraint, i.e., $\eps_j=c$. 
Our goal is to upper bound the probability that the rounded solution violates the constraint $Q$ by a function of $\epsilon$ and $c$. Using the fact that the expected value of $c$ over all the constraints is at most $\epsilon$, we get our required robustness guarantee. More formally, we prove the following.

\begin{lemma}\label{lem:maj}
Fix $j \in [m]$, and suppose that the basic SDP solution has an error value equal to $c$ on $C_j$, i.e., $\eps_j=c$. Then, the probability that $\sigma$ does not weakly satisfy $C_j$ is at most
\[
O\left(\sqrt{\epsilon}+\sqrt{\left(\gamma \sqrt{\log \frac{1}{\epsilon}} + 2c\right)\log \frac{1}{\epsilon}}+\frac{c}{\gamma}\right).
\]
\end{lemma}

By summing over all the constraints, and using linearity of expectation, the expected number of constraints that are not weakly satisfied by $\sigma$ is at most 
\begin{equation}
\label{eq:maj-error}
O\left(\sqrt{\epsilon}m+\sum_{j \in [m]}\left(\sqrt{\left(\gamma \sqrt{\log \frac{1}{\epsilon}} + 2\eps_j\right)\log \frac{1}{\epsilon}}+\frac{\eps_j}{\gamma}\right)\right).
\end{equation}
As the basic SDP has a total error at most $\epsilon m$, the average value of $\epsilon_j$ over $j\in [m]$ is at most $\epsilon$. Also note that the expression in~\Cref{eq:maj-error} is a concave function of $\eps_j$. Thus, using Jensen's inequality, we get that the expected number of constraints that are not weakly satisfied by $\sigma$ is at most 
\[
O\left(\sqrt{\epsilon}m+m\cdot \left(\sqrt{\left(\gamma \sqrt{\log \frac{1}{\epsilon}} + 2\eps\right)\log \frac{1}{\epsilon}}+\frac{\eps}{\gamma}\right)\right) \leq \tilde{O}\left(m\eps^{1/3}\right)
\]
This completes the proof of Theorem~\ref{thm:maj}. 
In the rest of the subsection, we prove Lemma~\ref{lem:maj}.

{For our predicate $P \subseteq \B^k$, l}et $\mathcal P$ be the convex hull of $P$, where the tuples are viewed as vectors in $\R^k$. 
\[
\mathcal{P} := \left\{ \sum_{\textbf{a} \in P} \lambda_{\textbf{a}} \textbf{a} : \lambda_\textbf{a}\geq 0\,~\forall \textbf{a}\in P,~ \sum_{\textbf{a} \in P}\lambda_\textbf{a}=1\right\}.
\]
We prove the following property about $\mathcal P$ using the fact that {$\MAJ \subseteq \Pol(P, \B^k \setminus \{(-1, \hdots, -1)\}$}. We recall that $O_{\MAJ}(P)$ denotes the set $\bigcup_{\textbf{x}_1,\ldots,\textbf{x}_L\in P,L \in \mathbb{N},\text{ odd}}\MAJ_L(\textbf{x}_1,\ldots,\textbf{x}_L)$ for a predicate $P \subseteq {\B}^k$.

\begin{lemma}
\label{lem:separating-hyperplane}
Let $P \subseteq {\B}^k$ be a predicate such that $(-1,-1,\hdots, -1) \not\in O_{\MAJ}(P)$. Then, there is a hyperplane separating $\mathcal{P}$ from the origin: there exists $\textbf{w} \in \mathbb{R} ^k$, $\textbf{w} \geq 0$ and $\norm{\textbf{w}}_1 = 1$ such that for every $\textbf{a} \in \mathcal{P}$, $\langle \textbf{a}, \textbf{w} \rangle \geq 0$. 
\end{lemma}

\begin{proof}
Consider the following linear program,
\begin{align*}
\textbf{maximize: } & \eta\\
\textbf{subject to: } & \sum_{i=1}^k w_i = 1\\
                      & \forall a \in P, \ \eta - \sum_{i=1}^k a_i w_i \le 0 \\ 
                      & \bw \ge 0
\end{align*}

It suffices to prove that the objective of this linear program is non-negative. To do this, we consider the dual program on variables $\nu \in \R$ and $\lambda_a \in \R$ for $a \in P$:
\begin{align*}
\textbf{minimize: } & \nu \\
\textbf{subject to: } & \lambda \ge 0 \wedge \sum_{a \in P} \lambda_a = 1 \tag{dual of $\eta$}\\
    &\forall i \in [k], \nu - \sum_{a \in P} a_i \lambda_a \ge 0 \tag{dual of $\bw$}
\end{align*}

As all the coefficients used in the LP are rational, we may assume that $\nu$ and $\lambda$ are rational. Assume for sake of contradiction that there is a solution to the dual LP with $\nu < 0$. Then, we have that for all $i \in [k]$,
\[
\sum_{a \in P} a_i \lambda_a < 0.
\]
Let $N$ be the least common denominator of the $\lambda_a$s. Consider the set of satisfying assignments to $P$ where we take $2N \lambda_a$ copies of $a$ for each $a \in P$. We also add an arbitrary element $b$ of $a$ to our set. As $\sum_{a \in P}a_i \lambda_i N <0$ for every $i \in [k]$, and $\sum_{a \in P}a_i \lambda_i N$ is an integer, we get that for every $i\in [k]$, $2\sum_{a \in P}a_i \lambda_a N + b \leq 0$. In other words, when we apply $MAJ_{2N+1}$ coordinatewise to this set of assignments in $P$, we get $(-1, -1, \ldots, -1)$. As $(P, Q)$ contains Majority of all odd arities as polymorphisms, this implies that the resulting output $(-1,-1,\ldots, -1)$ is in $Q$, a contradiction. 

Thus, the objective $\eta$ of the original LP is non-negative, completing the proof. 
\end{proof}

Now we use this lemma to complete the proof of Lemma~\ref{lem:maj}. Recall that our goal is to lower bound the probability that the assignment $\sigma$ output by the above algorithm weakly satisfies the constraint $C_j$. Suppose that ${C_j=(x_{j,1},x_{j,2},\ldots,x_{j,k})}$, and uses the predicate pair $(P,Q)$ where $P\subseteq Q = {\B}^k\setminus\{(-1,-1,\ldots,-1)\}$. {Let $S_j$ denote the set of variables appearing in $C_j$.} For each $l \in [k]$, let $\hat{\bv}_l := \bv(x_{j,l})$, where the function $\bv$ is defined as in (\ref{eq:vl}). {Let \( \mu^{(j)} \) and \( \Sigma^{(j)} \) be defined as follows:
\begin{align*}
\mu^{(j)}_l &:= \langle \hat{\bv}_l , \bv_0 \rangle & \forall l \in [k],\\
\Sigma^{(j)}_{l,l'} &:= \langle \hat{\bv}_l , \hat{\bv}_{l'} \rangle & \forall l,l' \in [k].
\end{align*}}

{Since $\eps_j = c,$ we have that
\[
    \sum_{\substack{f : S_j \to \B\\f(C_j) \in P}} \lambda_j(f) = 1 - c.
\]}
Using~\Cref{lem:separating-hyperplane}, we get $\bw \in \R^k$ with $\bw \geq 0$ and $\norm{\bw}_1=1$ such that ${\langle \bw,  \textbf{a}_i \rangle} \geq 0$ for all $\textbf{a}_i \in P$.
Combining this with the above properties of the basic SDP solution, we get the following. 
\begin{enumerate}
\itemsep=-0.5ex
\vspace{-1ex}
    \item (First moment). We have 
    \begin{align*}
    {\langle \bw , \mu^{(j)} \rangle} &=  {\sum_{f : S_j \to \B} \lambda_j(f) {\langle \bw, f(C_j)\rangle}}\\
    & \geq -c \text{ (Using \Cref{lem:separating-hyperplane} and $-1 \leq {\langle \bw, \ba \rangle} \leq 1 \,\forall {\ba \in \B^k}$ )}
    \end{align*}
    \item (Second moment). We have
    \begin{align*}
    \bw ^T {\Sigma^{(j)}} \bw &= {\sum_{f : S_j \to \B} \lambda_j(f)} {\langle \bw,  f(C_j) \rangle ^2}\\
    &\leq  {\sum_{\substack{f : S_j \to \B\\f(C_j) \in P}} \lambda_j(f) (\langle \bw,  f(C_j) \rangle ^2} + c\\
    &\leq {\sum_{\substack{f : S_j \to \B\\f(C_j) \in P}} \lambda_j(f) \langle \bw,  f(C_j) \rangle } + c \leq { \langle \bw, \mu \rangle} + 2c \ . \end{align*}
\end{enumerate}
We do casework on the value of ${\langle \bw, \mu^{(j)}\rangle}$. First, consider the case that ${\langle \bw, \mu^{(j)}\rangle} \geq \kappa = \gamma \sqrt{\log \frac{1}{\epsilon}}$. As $\norm{\bw}_1$=1, and $\bw \geq 0$, there exists $i \in [k]$ such that ${\mu^{(j)}_i} \geq \kappa$. As $\zeta_i \sim \mathcal{N}(0,1)$, using~\Cref{prop:gaussian-concentration}, with probability at least $1-\sqrt{\epsilon}$, we have $\zeta_i \geq -{\frac{\mu^{(j)}_i}{\gamma}}$. Thus, with probability at least $1-\sqrt{\epsilon}$, the rounded solution satisfies $Q$.  

Henceforth, we assume ${\bw^{T}\mu^{(j)}} < \kappa $.
For notational convenience let $\bt = -{\mu^{(j)} / \gamma }$.  We have 
\begin{equation}
\label{eq:wt-atmost}
\bw^{T} \bt \le \frac{c}{\gamma}
\end{equation}
and  ${\bw^{T}\Sigma^{(j)} \bw }\le \kappa +2c$. {From ~\Cref{prop:gaussian-multivariate}}, ${\langle \bw , \zeta \rangle} \sim {\mathcal{N}(0, \bw^T \Sigma^{(j)} \bw)}$. Thus, using~\Cref{prop:gaussian-concentration},  
with probability at least $1-\sqrt{\epsilon}$, we have that 
\begin{equation}
\label{eq:wzeta-atmost}
|\bw^{T} \zeta| \le O\left( \sqrt{(\kappa + 2c)\log \frac{1}{\epsilon}}\right)
\end{equation}

Note that the rounded solution does not satisfy $Q$ only if $\bt \geq \zeta$. We now upper bound the probability that this {occurs}. Together with~\Cref{eq:wt-atmost} and~\Cref{eq:wzeta-atmost},
 $\bt \geq \zeta$ implies that  
\[0 \le \bw^{T}(\bt - \zeta) \le O\left(\sqrt{(\kappa + 2c)\log \frac{1}{\epsilon}}+\frac{c}{\gamma}\right).\]

Take some coordinate with $w_i \ge 1/k$ and note that
\[
    {t_i}- \zeta_i \in \left[0, O\left(\sqrt{(\kappa + 2c)\log \frac{1}{\epsilon}}+\frac{c}{\gamma}\right)\right].
\]
{Since $t_i$ is a constant, this occurs with probability 
\[
O\left(\sqrt{(\kappa + 2c)\log \frac{1}{\epsilon}} + \frac{c}{\gamma}\right),
\]
using~\Cref{prop:gaussian-anticoncentration} applied to \( \zeta_i \sim \mathcal{N}(0, 1) \).} Thus, the probability that the assignment $\sigma$ does not satisfy $Q$ is at most 
\[
O\left(\sqrt{\epsilon}+\sqrt{(\kappa + 2c)\log \frac{1}{\epsilon}}+\frac{c}{\gamma}\right).
\]

This completes the proof of Lemma~\ref{lem:maj}.

\subsection{Warm-up for \texorpdfstring{$\AT$}{AT}: Random threshold rounding algorithm for Dual-Horn SAT}

As a stepping-stone for our algorithm for $\AT$ presented in the next section, we present a robust algorithm for the Dual-Horn SAT CSP. {The formal template Dual-Horn SAT allows for constraints of arbitrarily large arity, but we consider a finite fragment. That is, our template} $\Gamma={\{}P_1, P_2, \ldots, P_l{\}}$ {has} $P_i = {\B}^{k_i}\setminus\{(-1,-1,\ldots,-1)\}$ {for $i \in [l]$.} {Specifically to emulate Dual-Horn SAT, we restrict ourself to instances of $\fiPCSP(\Gamma)$ such that each clause has at most one negated literal and no constant literals appear.}

A robust algorithm for Dual-Horn SAT was found previously by Zwick~\cite{Zwick98}, and a matching hardness result was obtained by Guruswami and Zhou\cite{GuruswamiZ12}. We now present a robust algorithm for Dual-Horn SAT achieving similar guarantees as~\cite{Zwick98}. The random thresholding idea that we use here serves as a good warmup for the AT algorithm presented in the next subsection.   

\smallskip 
\noindent\fbox{%
    \parbox{\textwidth}{%
    \begin{enumerate}
                \item Given an instance $\Phi$ of ${\fiPCSP(\Gamma)}$ containing $n$ variables $\{u_1,u_2,\ldots,u_n\}$,
        solve the basic SDP and obtain a set of vectors $\textbf{v}_0, \textbf{v}_1, \ldots, \textbf{v}_n$ {and objective value $\eps > 0$}. Let $\mathbf{\mu} \in \R^n$ denote the first moments.  
        \begin{align*}
            \mu_i &= { \langle \textbf{v}_i , \textbf{v}_0\rangle } \quad \forall i \in [n]
        \end{align*}        
        \item Let $T$ be a geometric progression with first term $\sqrt{\eps}$, last term $1/K$ and spacing between terms is at least $K$, where $K$ is the maximum arity of the predicates in $\Gamma$. 
        \item Sample a uniformly random threshold $t \in T$.
        \item For each $i \in [n]$, round as follows
        \[
            \sigma(x_i) = \begin{cases}
            +1 & \text{ if } {\mu_i \ge t-1.}\\
            -1 & $\text{otherwise}$.
            \end{cases}
        \]
    \end{enumerate}
            }%
}
\smallskip 

\begin{theorem}\label{thm:or}
Let $\Phi$ be an instance of {$\fiPCSP(\Gamma)$ such that each clause has at most one negated literal and no constant literals} for which there is a basic SDP solution with error at most $\eps$. Then, the assignment $\sigma$ output by the above algorithm satisfies $1 - O_{\Gamma}(1/\log(1/\eps))$ fraction of the constraints of $\Phi$ in expectation. 
\end{theorem}

\begin{proof}
As with the analysis of {the} $\MAJ$ {polymorphism}, we fix a single constraint $C_j$ using the predicate $P$ and analyze the probability that $\sigma$ satisfies $C_j$.  Since the average SDP error over the constraints is equal to $\eps$, by Markov's inequality, for at least $1 - \sqrt{\eps}$ fraction of the constraints, the SDP error value is at most $\sqrt{\epsilon}$. Henceforth, we assume that $c \leq \sqrt{\epsilon}$. 

Suppose that the constraint $C_j$ uses the {literals} $x_1, \hdots, x_k$ { with some potentially equal. As an abuse of notation, we let $\mu_i$ denote the inner product of $\bv_0$ with $\bv(x_i)$.} First, consider the case when none of the variables are negated. 
Using the fact that the local assignments in the SDP solution have weight at least $1-\sqrt{\epsilon}$ fraction on assignments in $P$, we get that 
\[
\sum_{i=1}^k \mu_i \geq (1-\sqrt{\epsilon})(-(k-2))+\sqrt{\epsilon}(-k) = -(k-2)-2\sqrt{\epsilon}
\]
By the pigeonhole principle, we must have some $i$ with 
\[
\mu_i \ge \frac{-(k-2)-2\sqrt{\epsilon}}{k} =-1+\frac{2-2\sqrt{\epsilon}}{k} \ge -1+\frac{1}{K}
\]
where in the final step, we assumed that $\epsilon$ is small enough that $2\sqrt{\epsilon}\leq 1$.
Thus, $\mu_i \ge t-1$ for all $t \in T$. In this case, $\sigma(x_i)=+1$, and the constraint is satisfied by $\sigma$.

{Next, c}onsider the case {in which exactly one} literal is a negated variable in the constraint. Without loss of generality, assume that $x_1$ is negated. 
Then, we have that 
\[
-\mu_1 + \mu_2 + \cdots + \mu_k \ge  (1-\sqrt{\epsilon})(-(k-2))+\sqrt{\epsilon}(-k) = -(k-2)-2\sqrt{\epsilon}.\]
Thus, there exists an integer $i \in \{2, 3, \hdots, k\}$ such that 
\[\mu_i \ge\frac{ \mu_1-(k-2)-2\sqrt{\epsilon}}{(k-1)} = -1+\frac{\mu_1+1-2\sqrt{\epsilon}}{k-1}\]
The algorithm rounds $x_i$ to $-1$ only if $t > \frac{\mu_1+1-2\sqrt{\epsilon}}{k-1}$. On the other hand, the algorithm rounds $x_1$ to $+1$ only if $t \leq \mu_1+1$. Since there is at most one element in $T$ in $(\frac{\mu_1+1-2\sqrt{\epsilon}}{k-1},\mu_1+1)$, there exists at most one $t \in T$ which would round $x_i$ to $-1$ but $x_1$ to $+1$. Therefore, the probability the constraint is satisfied is at least $1 - 1/|T|$ = $1 - {O_{\Gamma}(1/\log(1/\eps))}$. By linearity of expectation over all the constraints with error at most $\sqrt{\epsilon}$, we get the required claim. 
\end{proof}

{~\Cref{thm:or} implies a polynomial time robust algorithm for Dual-Horn SAT: Suppose that the template $\Gamma={\{}P_1, P_2, \ldots, P_l{\}}$ {has} $P_i = {\B}^{k_i}\setminus\{(-1,-1,\ldots,-1)\}$for $i \in [l]$. Consider an instance $\Phi$ of $\fiPCSP(\Gamma)$ such that each clause has at most one negated literal and no constant literals appear, with the property that there is an assignment that strongly satisfies $(1-\epsilon)$ fraction of the constraints. Without loss of generality, we can assume that $\epsilon \geq \frac{1}{m}$, where $m$ is the number of constraints in $\Phi$. By setting $\delta = \sqrt{\epsilon}-\epsilon$ in ~\Cref{thm:sdp}, we can obtain a solution to basic SDP relaxation of $\Phi$ with error at most $\sqrt{\epsilon}m$. ~\Cref{thm:or} then implies that the output of the above algorithm satisfies $1 - O_{\Gamma}(1/\log(1/\eps))$ fraction of the constraints of $\Phi$ in expectation.}

\begin{remark}
We remark that the above algorithm applies directly to Horn-SAT as well. Furthermore, for an arbitrary {Boolean promise template} $\Gamma=\{(P_1,Q_1),(P_2,Q_2),\ldots,(P_l,Q_l)\}$ with OR polymorphisms\footnote{$\OR(x_1,x_2,\ldots,x_L)$ is equal to $+1$ if there is $i \in [L]$ with $x_i=+1$, and $-1$ otherwise.}
$\OR \subseteq \Pol(\Gamma)$, it's been proved~\cite{BrakensiekG21} that $\OR \subseteq \Pol(\Gamma')$ where $\Gamma'$ is the CSP containing the predicates $\{Q_1,Q_2,\ldots,Q_l\}$. By Schaefer's theorem~\cite{Schaefer78}, any {template $\Gamma$} with $\OR \subseteq \Pol(\Gamma)$ is ppp-definable from $k$-Horn-SAT, and thus, our algorithm proves that there is a polynomial time robust algorithm {for $\PCSP(\Gamma)$ (i.e., negated literals are \emph{not} allowed)} for every {Boolean promise template} $\Gamma$ with $\OR \subseteq \Pol(\Gamma)$.
\end{remark}

\subsection{Algorithm for \texorpdfstring{$\AT$}{AT}}
\label{subsec:at-symmetric}
\newcommand{\sgn}{\operatorname{sgn}}
\newcommand{\Aff}{\operatorname{Aff}}

We now show how to combine the ideas in the previous two subsections to obtain a polynomial time robust algorithm for every Boolean {promise template} $\Gamma=\{(P_1,Q_1),\ldots,(P_l,Q_l)\}, P_i \subseteq Q_i \subseteq {\B}^{k_i}$ with $\AT\subseteq \Pol(\Gamma)$. 
Similar to the $\MAJ$ polymorphisms case, we can without loss of generality assume that $Q_i \neq {\B}^{k_i}$ for every $i \in [l]$.
We first reduce an arbitrary {promise template} $\Gamma$ with $\AT \subseteq \Pol(\Gamma)$ to a specific {template} that we will work with. 
For a vector $\textbf{w} \in \R^k$, define $\sgn(\textbf{w})_i$ to be $-1$ if {$w_i < 0$} and $+1$ otherwise. Define $\Gamma_{AT}$ to be the following {infinite} family of \emph{weighted hyperplane predicates}:
\begin{align*}
    \Gamma_{\AT} &:= \{(P_{\textbf{w},b} := \{{\textbf{x}} \in {\B}^k: { \langle \textbf{w} , \textbf{x} \rangle }= b\}, Q_{\textbf{w},b} := {\B}^k \setminus \{\sgn(\textbf{w}), -\sgn(\textbf{w})\}) :\\&{k \in \N,} b \in \Q, {\textbf{w} \in \Q^k\setminus \textbf{0}},\textbf{w}.\sgn(\textbf{w})>b, -\textbf{w}.\sgn(\textbf{w})<b\}
\end{align*}
We {prove} that these predicates indeed have AT of all odd arities as polymorphisms. 
\begin{claim}
$\AT \subseteq \Pol(\Gamma_{AT})$
\end{claim}
\begin{proof}
Fix $b \in \Q$ and {$\textbf{w} \in \Q^k\setminus \textbf{0}$}. Let $(P_{{\bw},b}, Q_{{\bw},b})$ be the corresponding predicate for these values. It suffices to show that $\AT \subseteq \Pol(P_{{\bw},b},Q_{{\bw},b})$. Fix an odd arity $L$ and pick $\textbf{x}_1, \hdots, \textbf{x}_L \in P_{{\bw},b}$.  Observe that
\[
    \AT(\textbf{x}_1, \hdots, \textbf{x}_L) = \sgn(\textbf{x}_1 - \textbf{x}_2 + \cdots + \textbf{x}_L).
\]
Further, ${ \langle \textbf{w} , (\textbf{x}_1 - \textbf{x}_2 + \cdots + \textbf{x}_L) \rangle }= b$. 
This implies that $\sgn(\textbf{x}_1 - \textbf{x}_2 + \cdots + \textbf{x}_L) \neq \sgn(\textbf{w})$ as otherwise,
\[
    b = { \langle \textbf{w} , (\textbf{x}_1 - \textbf{x}_2 + \cdots + \textbf{x}_L) \rangle  \ge \langle \textbf{w} , \sgn(\textbf{w}) \rangle} > b.
\]
where we used the fact that the absolute value of each entry in $\textbf{x}_1-\textbf{x}_2+\ldots +\textbf{x}_L$ is at least $1$.
By a similar argument, $\sgn(\textbf{x}_1 - \textbf{x}_2 + \cdots + \textbf{x}_L) \neq - \sgn(\textbf{w})$. Thus, $\AT(\textbf{x}_1, \hdots, \textbf{x}_L) \in Q_{{\bw},b}$, as desired.
\end{proof}

\newcommand{\const}{\operatorname{const}}

Let $\Gamma_{\const}$ be the {Boolean promise template} $\{(\{-1\},\{-1\}),(\{+1\},\{+1\})\}$.
{Since $\AT$ is an idempotent family of polymorphisms, we have that $AT \subseteq \Pol(\Gamma_{\const})$.}
We show that an arbitrary Boolean {promise template} with $\AT$ polymorphisms can be reduced to {a finite fragment of} ${\Gamma_{\AT}\cup \Gamma_{const}}.$

\begin{lemma}\label{thm:at-general}
Let $\Gamma$ be any Boolean {promise template} with $\Gamma=\{(P_1,Q_1),\ldots,(P_l,Q_l)\}, P_i \subseteq Q_i \subsetneq {\B}^{k_i}$ with $\AT\subseteq  \Pol(\Gamma)$.
 Then, $\Gamma$ is ppp-definable from a {finite} Boolean {promise template} $\Gamma'$ with $\Gamma' \subseteq \Gamma_{\AT}\cup \Gamma_{const}$.
\end{lemma}
We prove~\Cref{thm:at-general} in~\Cref{subsec:at-general}. Recall that if a {promise template} $\Gamma$ is ppp-definable from another {promise template} $\Gamma'$ {and} if ${\fiPCSP(\Gamma')}$ has a polynomial time robust algorithm, then ${\fiPCSP(\Gamma)}$ has a polynomial time robust algorithm as well (up to losing constant factors in the error parameter).
In the rest of this section, {we prove that for any finite} Boolean {promise template} $\Gamma$ with $\Gamma \subseteq \Gamma_{\AT}$, {$\fiPCSP(\Gamma)$ has a robust algorithm}, thereby obtaining a robust algorithm for ${\fiPCSP(\Gamma)}$ {for} every Boolean {promise template} $\Gamma$ with $\AT \subseteq\Pol(\Gamma)$. 
We state our algorithm below. \smallskip 

\noindent\fbox{%
    \parbox{\textwidth}{%
    \begin{enumerate}
        \item Given an instance $\Phi$ of ${\fiPCSP(\Gamma)}$ containing $n$ variables $V=\{u_1,u_2,\ldots,u_n\}$,
        solve the basic SDP and obtain a set of vectors $\textbf{v}_0, \textbf{v}_1, \ldots, \textbf{v}_n$ {and objective value $\epsilon > 0$}.  \item {Sample an $n$ dimensional Gaussian $\zeta \sim \mathcal N(\textbf{0}, \textbf{I})$.}
        \item Choose $\delta$ uniformly at random from $\{ p, r p, \ldots, r^{\kappa} p\}$ where $p =\epsilon^{c_1}$, and $r = \kappa = \Theta\left(\frac{\log \frac{1}{\epsilon}}{\log \log \frac{1}{\epsilon}}\right)$ such that $r^{\kappa}p=\epsilon^{c_2}$. Here, $c_1$ and $c_2$ are two arbitrary real constants satisfying $0<c_2<c_1<0.25$.
        \item For every $i \in [n]$, let $\textbf{v}_i = \alpha_i \textbf{v}_0 + \textbf{v}'_i$, where $\textbf{v}'_i$ is orthogonal to $\textbf{v}_0$. We output an assignment $\sigma$ as follows.
        \[
        \sigma(u_i) = 
        \begin{cases}
        -1, \text{ if  }\langle \zeta, \textbf{v}'_i\rangle \ge \delta\alpha_i \left| \langle \zeta,  \textbf{v}_0\rangle \right|. \\ 
        +1, \text{ otherwise.}
        \end{cases}
        \]
    \end{enumerate}
            }%
}

\smallskip 

We now analyze the algorithm.
\begin{theorem}
\label{thm:at}
Let $\Gamma$ be a {finite} Boolean {promise template} such that $\Gamma \subseteq \Gamma_{\AT}\subseteq\Gamma_{const}$. Let $\Phi$ be an instance of ${\fiPCSP}(\Gamma)$ for which there is a basic SDP solution with average error at most $\epsilon {> 0}$. Then, the assignment output by the above algorithm satisfies at least {$1-O_{\Gamma}\left(\frac{\log \log \frac{1}{\epsilon}}{\log \frac{1}{\epsilon}}\right)$} fraction of constraints of $\Phi$ in expectation. 
\end{theorem}
{~\Cref{thm:at} together with ~\Cref{thm:at-general} implies a polynomial time robust algorithm for Boolean promise templates with AT polymorphisms. Consider an instance $\Phi$ of $\fiPCSP(\Gamma)$ with $\Gamma \subseteq \Gamma_{AT} \cup \Gamma_{const}$ with the property that there is an assignment that strongly satisfies $(1-\epsilon)$ fraction of the constraints. Without loss of generality, we can assume that $\epsilon \geq \frac{1}{m}$, where $m$ is the number of constraints in $\Phi$. By setting $\delta = \sqrt{\epsilon}-\epsilon$ in ~\Cref{thm:sdp}, we can obtain a solution to basic SDP relaxation of $\Phi$ with error at most $\sqrt{\epsilon}m$. ~\Cref{thm:at} then implies that the output of the above algorithm satisfies $1-O_{\Gamma}\left(\frac{\log \log \frac{1}{\epsilon}}{\log \frac{1}{\epsilon}}\right)$ fraction of the constraints of $\Phi$ in expectation. Together with~\Cref{lem:maj}, this completes} the proof of~\Cref{thm:main-algorithm}.

For ease of notation, we just use $O()$ instead of $O_\Gamma()$ when $\Gamma$ is clear from the context. As before, we prove~\Cref{thm:at} by proving a lower bound on the probability that a particular constraint is satisfied. Consider the constraint $C {= (x_1, \hdots, x_k)}$ using the predicate pair $(P,Q)$ where $P \subseteq Q \subseteq {\B}^k$, and let $c$ denote the error of the SDP solution on constraint ${C_j}$. 
As the average value of $c$ over all the constraints is at most $\epsilon$, using Markov's inequality, at least $1-\sqrt{\epsilon}$ fraction of the constraints have SDP error at most $\sqrt{\epsilon}$. 
We restrict our analysis to these constraints with SDP error $c \leq \sqrt{\epsilon}$ and show that the rounded solution violates the predicate $Q$ with probability at most $O\left(\frac{\log \log \frac{1}{\epsilon}}{\log \frac{1}{\epsilon}}\right)$, thereby proving~\Cref{thm:at}. {Finally, we note that the rounding algorithm preserves the properties of idempotence and folding: specifically, $\pm \textbf{v}_0$ are rounded to $\pm 1$, respectively. Moreover, for any vector $\textbf{v}$, the algorithm ensures that $\textbf{v}$ and $-\textbf{v}$ are rounded to distinct values.}

{We break the analysis into cases based on the predicate pair $(P,Q)$.} We first consider the case when $P=Q$ and $P \subseteq \Ham_k \{0,k\}$, which is a generalization of $\Gamma_{const}$.

\begin{lemma}
\label{lem:equal-case}
Let {$\Gamma = (P, Q)$ } be a {finite} Boolean {promise template} with $\AT \subseteq \Pol(\Gamma)$ and {$P=Q$, $P\subseteq \Ham_k \{0,k\}$ for an integer $k \geq 1$}.
Let $\Phi$ be an instance of ${\fiPCSP}(\Gamma)$, and in the basic SDP solution of $\Phi$, suppose that the constraint ${C_j}$ using the predicate pair $(P,Q)$ has error at most $\sqrt{\epsilon}$. Then, the assignment $\sigma$ output by the above algorithm does not weakly satisfy ${C_j}$ with probability at most $O\left(\frac{\log \log \frac{1}{\epsilon}}{\log \frac{1}{\epsilon}} \right)$.
\end{lemma}

We defer the proof of~\Cref{lem:equal-case} to~\Cref{sec:missing-proofs}. 

We now consider the case when $P=P_{\textbf{w},b}$ and $Q=Q_{\textbf{w},b}$ {for some $k \in \N, b \in \Q, \textbf{w} \in \Q^k \setminus \textbf{0},\textbf{w}.\sgn(\textbf{w})>b, -\textbf{w}.\sgn(\textbf{w})<b$}.
For ease of notation, we let {$x_1, x_2, \ldots, x_k$ denote the literals assigned to ${C_j}$, with $\bv_1, \hdots, \bv_k$ the corresponding values of $\bv(x_1), \hdots, \bv(x_k).$ We let $S$ denote the set of variables in the scope of $C_j$.} We have that
\[
    {\sum_{\substack{f : S \to \B\\f(C_j) \in P}} \lambda_j(f) \ge 1 - \sqrt{\eps}}.
\]

Let $\textbf{v}_s = {\sum_{i \in [k]}\textbf{w}_i\textbf{v}_i}$, and let the component of $\textbf{v}_s$ along $\textbf{v}_0$ be $\alpha \textbf{v}_0$, and the component normal to $\textbf{v}_0$ be equal to $\textbf{v}'_s$.
\[
\textbf{v}_s = \alpha \textbf{v}_0 + \textbf{v}'_s,\langle \textbf{v}_0, \textbf{v}'_s\rangle = 0.\] 
We use the first and second moments properties of the vectors $\textbf{v}_1, \textbf{v}_2, \ldots,\textbf{v}_k$ to get the following. 

\begin{enumerate}
    \item (First moments). We have 
    \begin{align*}
            \alpha &=  {\sum_{f : S \to \B}\lambda_j(f)\langle \textbf{w}, {f(C_j)}\rangle}\\
        &= {\sum_{\substack{f : S \to \B\\f(C_j) \in P}}} {\lambda_j(f)}\langle \textbf{w}, {f(C_j)}\rangle + {\sum_{\substack{f : S \to \B\\f(C_j) \not\in P}}} {\lambda_j(f)}\langle \textbf{w}, {f(C_j)}\rangle\\ 
        &= {\sum_{\substack{f : S \to \B\\f(C_j) \in P}}} {\lambda_j(f)}b + {\sum_{\substack{f : S \to \B\\f(C_j) \not\in P}}} {\lambda_j(f)}\langle \textbf{w}, {f(C_j)}\rangle\\
        &= b + {\beta} 
        \end{align*}
    where
    \[
    {\beta} = {\sum_{\substack{f : S \to \B\\f(C_j) \not\in P}}} {\lambda_j(f)}({b} - \langle \textbf{w}, {f(C_j)}\rangle)
    \]
    We have $ {|\beta|}=O(\sqrt{\epsilon})$. 
    \item (Second moments). We have 
    $$
        \norm{\textbf{v}_s}_2^2 = 
        {\sum_{i, j \in [k]}\textbf{w}_i \textbf{w}_j} \langle \textbf{v}_i, \textbf{v}_j \rangle = {\sum_{f : S \to \B}} {\lambda_j(f)} ({ \langle {f(C_j)} , \textbf{w}\rangle})^2 = b^2 + {\beta'}
        $$
    where $|{\beta'}|=O(\sqrt{\epsilon})$.
\end{enumerate}
Thus, we get 
$\norm{\textbf{v}'_s}_2^2 = \norm{\textbf{v}_s}_2^2 - \alpha^2 
    = {(b^2+\beta')-(b+\beta)^2}$
    which is at most $O(\sqrt{\epsilon})$. 

We are now ready to analyze the algorithm. We consider two cases separately:

\smallskip \noindent \textbf{Case $1$}. Suppose that there exists $i \in [k]$ such that $\norm{\textbf{v}'_i}_2 > k \delta r^2$. We claim that in this case, the rounded solution satisfies $Q$ with probability at least $1-O(\frac{1}{r})$. 

{From ~\Cref{prop:gaussian-multivariate}, $\langle \zeta, \textbf{v}'_j \rangle \sim \mathcal{N}(0, \norm{\textbf{v}'_j}_2^2)$ for every $j \in [k]$. 
Since $\norm{\textbf{v}'_i}_2 > k \delta r^2$, using~\Cref{prop:gaussian-anticoncentration}, with probability at least $1 - O\left(\frac{1}{r}\right)$, we have that $|\langle \zeta, \textbf{v}'_i \rangle| > k \delta r$.

Recall that $\norm{\textbf{v}'_s}_2 \leq \epsilon^{0.25} \leq \delta$. Thus, $|\langle \zeta, \textbf{v}'_s \rangle| \leq r \delta$ with probability at least $O\left(\frac{1}{r}\right)$. 
As $|\langle \zeta, \textbf{v}'_i \rangle| > k \delta r$, there exists $i' \in [k], i' \neq i$, such that $|\langle \zeta, \textbf{v}'_{i'} \rangle| > \delta r$, and $\langle \zeta, \textbf{v}'_{i} \rangle$ and $\langle \zeta, \textbf{v}'_{i'} \rangle$ have opposite signs.

Furthermore, as $\langle \zeta, \textbf{v}_0 \rangle \sim \mathcal{N}(0, 1)$, using~\Cref{prop:gaussian-concentration}, we get that $|\langle \zeta, \textbf{v}_0 \rangle| \leq r$ with probability at least $1 - O(\epsilon)$. Combining this with the above, we conclude that with probability $1 - O\left(\frac{1}{r}\right)$, all the following hold:
\begin{enumerate}
    \item[(1)] $|\langle \zeta, \textbf{v}_0 \rangle| \leq r$.
    \item[(2)] $|\langle \zeta, \textbf{v}'_{j} \rangle| > \delta r$ for all $j \in \{i, i'\}$.
    \item[(3)] $\langle \zeta, \textbf{v}'_{i} \rangle$ and $\langle \zeta, \textbf{v}'_{i'} \rangle$ have opposite signs.    
\end{enumerate}

Note that for every $j \in [k]$, $|\alpha_j| = |\langle \textbf{v}_0, \textbf{v}_j \rangle| \leq 1$. Using (1), we have $|\delta \alpha_j \langle \zeta, \textbf{v}_0 \rangle| \leq r \delta$. Thus, from (2), for all $j \in \{i, i'\}$, we deduce that $x_j$ is set to $-1$ if $\langle \zeta, \textbf{v}'_j \rangle > 0$, and $+1$ otherwise.

Combining with (3), with probability at least $1 - O\left(\frac{1}{r}\right)$, $i$ and $i'$ are rounded to different values, which implies that the rounded solution satisfies $Q$.}

\smallskip \noindent \textbf{Case $2$.} Suppose that for every $i \in [k]$, we have $\norm{\textbf{v}'_i}_2 \leq \frac{\delta}{2r^2}$. 

As $\langle \zeta, \textbf{v}'_i \rangle \sim \mathcal{N}(0, \norm{\textbf{v}'_i}_2^2)$, using~\Cref{prop:gaussian-concentration},
we get that with probability at least $1-O(\frac{1}{r})$, for every $i \in [k]$, $|\langle \zeta, \textbf{v}'_i \rangle | \leq \frac{\delta}{2r}$. 
On the other hand, using~\Cref{prop:gaussian-anticoncentration}, we have that $|\langle \zeta, \textbf{v}_0 \rangle | \geq \frac{1}{r}$ with probability at least $1-\frac{1}{r}$. 
Furthermore, As $\alpha_i^2 + \norm{\textbf{v}'_i}_2^2=1$ for every $i \in [k]$, we get that $|\alpha_i| \geq 1-\delta \geq \frac{1}{2}$ for every $i \in [k]$.  
Thus, with probability at least $1-O(\frac{1}{r})$, for every $i \in [k]$, $x_i$ is set to be {$-1$} if $\alpha_i \leq 0$, and {$+1$} otherwise. Combining this with the fact that ${\sum_i \textbf{w}_i} \alpha_i = b + O(\sqrt{\epsilon})$, and that {$\textbf{w}.\sgn(\textbf{w})>b, -\textbf{w}.\sgn(\textbf{w})<b$,} for small enough $\epsilon$,
we get the rounded solution has variables assigned $+1$ and $-1$. {Here, we assumed that $\epsilon$ is at most $c := c(\Gamma)$ for some function $c$. This is without loss of generality, since when $\epsilon$ is larger than $c$, a uniformly random assignment satisfies the robust algorithm requirement.}

\smallskip \noindent \textbf{Completing the proof.}
We finish the proof by showing that with probability at least $1-O(\frac{1}{r})$, at least one of the above two cases holds. None of the above two cases hold {only} if for some $i \in [k]$, we have 
\[
\frac{\delta }{2r^2} < \norm{\textbf{v}'_i}_2 \leq k \delta r^2
\]
Or equivalently, 
\[
\frac{\norm{\textbf{v}'_i}_2}{k r^2} \leq \delta < \norm{\textbf{v}'_i}_2 2 r^2
\]
This holds with probability at most $O(\frac{1}{r})$ for every value of $\norm{\textbf{v}'_i}$ as we are picking $\delta$ from $\{ p, rp, \ldots, r^\kappa p \}$ uniformly at random with $\kappa=r$. 

\subsection{Proof of Theorem~\ref{thm:at-general}}
\label{subsec:at-general}

 For a predicate $P \subseteq {\B}^k$, we let $\Aff(P)\subseteq \R^k$ to denote the affine hull of $P$.
\[
\Aff(P):= \left\{ \sum_{\textbf{a} \in P} \lambda_{\textbf{a}} \textbf{a} : \sum_{\textbf{a} \in P}\lambda_\textbf{a}=1, \lambda_{\textbf{a}} \in \R\right\}.
\]
We recall that $O_{\AT}(P)$ denotes the set $\bigcup_{\textbf{x}_1,\ldots,\textbf{x}_L\in P,L \in \mathbb{N},\text{ odd},}\AT_L(\textbf{x}_1,\ldots,\textbf{x}_L)$ for a predicate $P \subseteq {\B}^k$.
{In order to} prove~\Cref{thm:at-general}, we need to use the following lemma implicit in~\cite{BrakensiekG21}.

\begin{lemma}\label{lem:AT}
Let $P\subseteq {\B}^k$ be a predicate such that there is non-trivial dependence in each coordinate, i.e., for every $i\in[k]$, there exist vectors $\textbf{x},\textbf{y}\in P$ with $x_i = -1$ and $y_i = +1$. Then, $O_{\AT}(P) = \{\sgn(\textbf{x} - \textbf{y}) : \textbf{x}, \textbf{y} \in \Aff(P), \forall i, x_i \neq y_i \}$.
\end{lemma}
We present the proof in~\Cref{sec:missing-proofs} for the sake of completeness. We also need the following claim.

\begin{claim}\label{claim:hyperplane}
    Let $H$ be a vector space in $\R^k$ such that there is no $\textbf{y} \in H$ with $y_i > 0$ for all $i$. Then, there exists {$\mathbf{w} \neq 0$ with $w_i \geq 0$} for all $i$ and ${ \langle \mathbf{w} , \mathbf{y} \rangle} = 0$ for all $\textbf{y} \in H$. \end{claim}

\newcommand{\by}{\textbf{y}}

\begin{proof}
Since $H$ and the positive orthant $\R_{+}^k:= \{ \textbf{y} \in \R^k : y_i >0\,\,\,\forall i \in [k]\}$  are {disjoint} convex bodies, {by the separating hyperplane theorem (e.g., \cite{boyd2023Convex}),} there exists {nonzero} ${\bw} \in \R^k$ and $b \in \R$ such that for all ${\bv}\in \R_{+}^k$, ${ \langle \bw , \bv \rangle \ge b}$ and for all $\textbf{y} \in H$, ${ \langle \by , \bw \rangle}\le b$. {Since $\bw$ is nonzero and $\R_{+}^k$ is open, we must more strongly have that ${ \langle \bw , \bv \rangle > b}$ and for all ${\bv}\in \R_{+}^k$.}

Taking the limit as ${\bv} \to 0$ {of ${ \langle \by , \bw \rangle}\le b$ for all $\textbf{y} \in H$}, we have that $b \le 0$. Further, since $(0,0,\ldots,0)\in H$, we must have that $b = 0$ and since $H$ is a vector space, ${ \langle \by , \bw \rangle } = 0$ for all $\by \in H$. Thus, ${\bw}$ is normal to $H$. Note that ${\bw}$ {is nonzero and} has all coordinates {non-negative} as ${ \langle \bv , \bw \rangle} >0$ for all ${\bv}$ in the positive orthant. 
\end{proof}

\begin{proof}[Proof of Lemma \ref{thm:at-general}.]
Fix a pair of predicates $(P,Q) \in \Gamma$. It suffices to show that  $(P,Q)$ is  ppp-definable from $\Gamma_{\AT} \cup \Gamma_{\const}$. If $P$ has coordinates of fixed value, we can use a gadget reduction from $\Gamma_{\const}$ to simulate these values. Thus, we assume that $P$ has non-trivial dependence in each coordinate, and thus we apply Lemma~\ref{lem:AT} to get that $Q \supseteq O_{\AT}(P) = \{\sgn(\textbf{x} - \textbf{y}): \textbf{x}, \textbf{y} \in \Aff(P), \forall i, x_i \neq y_i \}.$ We may without loss of generality assume that $Q = \{\sgn(\textbf{x} - \textbf{y}) : \textbf{x}, \textbf{y} \in \Aff(P), \forall i, x_i \neq y_i \}.$ 

For every $\textbf{x} \in {\B}^k\setminus Q$, we find $\textbf{w}, b$ such that 
\[
P_{\textbf{w},b} := \{\textbf{v} \in {\B}^k: { \langle \textbf{w} ,  \textbf{v} \rangle }= b\},
\quad Q_{\textbf{w},b} := {\B}^k \setminus \{\sgn(\textbf{w}), -\sgn(\textbf{w}\})
\]
satisfy $P \subseteq P_{\textbf{w},b}, {Q} \subseteq Q_{\textbf{w},b}$ and $\sgn(\textbf{w})=\textbf{x}$. By applying this for every $\textbf{x} \in {\B}\setminus Q$, we get a set of predicate pairs $(P_1, Q_1), (P_2, Q_2), \ldots, (P_L,Q_L)$ with $L\leq 2^k$ such that 
\begin{enumerate}
\item $P \subseteq P_i $ for every $i \in [L]$. 
\item $(P_i, Q_i) \in \Gamma_{AT}$ for every $i \in [L]$. 
\item $\bigcap_{i \in [L]} Q_i = Q$. 
\end{enumerate}
This shows that $(P,Q)$ is ppp-definable from $\{(P_1, Q_1), (P_2, Q_2), \ldots, (P_L,Q_L)\}\subseteq \Gamma_{AT}$.

Henceforth, our goal is to show that for every $\textbf{x} \in {\B}^k\setminus Q$, we can find $\textbf{w}, b$ such that $P_{\textbf{w},b} := \{x \in {\B}^k: { \langle \textbf{w} , x \rangle }= b\}$
satisfies $P \subseteq P_{\textbf{w},b}$, ${ \langle \textbf{w},  \sgn(\textbf{w})\rangle } >b, { \langle \textbf{w}, \sgn(\textbf{w})\rangle} >-b$ and $\sgn(\textbf{w})=\textbf{x}$. Without loss of generality, we can assume that $\textbf{x}=(+1,+1,\ldots,+1)$. Fix an arbitrary vector $\overline{\textbf{x}} \in P$ such that $\overline{\textbf{x}}\notin \{(-1,-1,\ldots,-1), (+1,+1,\ldots,+1)\}$. Such a vector is guaranteed to exist as $P$ does not contain {${\textbf{x}}$} and has non-trivial dependence on each coordinate. 
Let $H$ be a subspace of $\mathbb{R}^k$ defined as follows:
\[
H := \{ \textbf{y}-\overline{\textbf{x}} : \textbf{y} \in \Aff(P)\}
\]
As $\textbf{x} \notin O_{AT}(P)$, using~\Cref{lem:AT}, we get that for every $\textbf{z} \in H$, $\sgn(\textbf{z})\neq \textbf{x}$, or in other words, there is no $\textbf{z} \in H$ with $z_i >0$ for all $i \in [k]$.
Using~\Cref{claim:hyperplane}, we can obtain {nonzero $\textbf{w}\geq \textbf{0}$} such that ${ \langle \textbf{w}, \textbf{y}\rangle}=0$ for all $\textbf{y} \in H$. This shows that ${ \langle \textbf{w}, \textbf{y}\rangle} = b$ for every $\textbf{y} \in P$, where ${b := \langle \textbf{w}, \overline{\textbf{x}}\rangle}$ satisfies $\sum_i w_i >b$, $\sum_i w_i >-b$. 
\end{proof}

\section{Unique Games based Hardness}
\label{sec:ug-hardness}

In this section, we prove~\Cref{thm:main-hardness}. First, we use the analysis of $\AT$ and $\MAJ$ polymorphisms for symmetric {promise templates} with folding and idempotence in~\cite{BrakensiekG21} to show that we can relax $\Gamma$ into one of five candidate PCSP types. 

\begin{lemma}
\label{lem:reduction}
Let $\Gamma=(P,Q)$ be a Boolean {promise template} such that $
\MAJ_{L_1}, \AT_{L_2} \notin \Pol(\Gamma)$ for some odd integers $L_1, L_2$. Then, there exists a Boolean {template} $\Gamma'=(P,Q)$ that is {fi}ppp-definable from $\Gamma$ that is equal to {one} of the following:
\begin{enumerate}
    \item $k$ is even, and $\Gamma_1 = (P,Q), P=\Ham_{k}\{\frac{k}{2}\}, Q = \Ham_{k} \{0,1,\ldots,k\} \setminus \{b\}$ where $b \in \{1,k-1\}$. 
    \item $k$ is odd, $\Gamma_2 = (P,Q), P=\Ham_k \{l,\frac{k+1}{2}\}$, $Q=\Ham_k \{0,1,2,\ldots,k-1\}$, where $l \leq \frac{k-1}{2}$.
    \item $\Gamma_3 = (P,Q), P=\Ham_k \{l,k\}$, $Q=\Ham_k \{1,2,\ldots,k\}$, where $l\neq 0, l \leq \frac{k-1}{2}$. 
    \item $\Gamma_4 = (P,Q),P=\Ham_k \{l\}, Q=\Ham_k \{0,1,\ldots,k\}\setminus \{0,k-1\}$ where $l \in \{1,2,\ldots,k-1\}, l \leq \frac{k-1}{2}$. 
    \item $\Gamma_5=(P,Q), P=\Ham_k \{1,k\}, Q=\Ham_k \{0,1,\ldots,k\}\setminus\{b\}$ for arbitrary $b \in \{0,1,\ldots,k\}$.
\end{enumerate}
\end{lemma}
We defer the proof of~\Cref{lem:reduction} to~\Cref{sec:missing-proofs}.

Recall that if a {promise template} $\Gamma'$ is {fi}ppp-definable from another {promise template} $\Gamma$, if ${\fiPCSP(\Gamma)}$ has a polynomial time robust algorithm, then ${\fiPCSP(\Gamma')}$ has a polynomial time robust algorithm as well (\Cref{prop:fippp-defn}). Thus, to show~\Cref{thm:main-hardness}, it suffices to show Unique Games based hardness of obtaining robust algorithms for ${\fiPCSP(\Gamma_{i})}$ {for $i \in [5]$ from Lemma~\ref{lem:reduction}}.
We achieve {these hardness results} by showing integrality gaps for the basic SDP relaxation of {each ${\fiPCSP(\Gamma_{i})}$}. 
Raghavendra's result for CSPs~\cite{Raghavendra08} shows that integrality gaps for the basic SDP relaxation can be translated to Unique Games Conjecture (UGC)~\cite{Khot02} based inapproximability results. In fact, his result is verbatim applicable to Promise CSPs as well. 
\begin{theorem}[Special case of ~\cite{Raghavendra08} for Boolean {fiPCSPs} when the SDP is feasible]
\label{thm:raghavendra}
Suppose that for a Boolean {promise template} $\Gamma$, there is a finite integrality gap for the basic SDP relaxation {$\fiPCSP(\Gamma)$}, i.e., there is a finite instance $I$ of ${\fiPCSP}(\Gamma)$ on which the basic SDP relaxation is feasible but there is no assignment that weakly satisfies $I$. Then, there exists a constant $s<1$ that is a function of $\Gamma, I$ such that the following decision problem is NP-hard for sufficiently small $\epsilon,\delta >0$, assuming UGC. Given an instance $\Phi$ of ${\fiPCSP(\Gamma)}$, distinguish between the two cases:
\begin{enumerate}
    \item (Completeness.) There exists an assignment that strongly satisfies $1-\epsilon$ fraction of the constraints in $\Phi$.
    \item (Soundness.) No assignment weakly satisfies $s+\delta$ fraction of the constraints in $\Phi$.
\end{enumerate}
\end{theorem}

{
\begin{remark}\label{rem:Ragh}
It is worth noting that Raghavendra does not work directly with the basic SDP relaxation of $\fiPCSP(\Gamma)$ but rather with the basic SDP relaxation of $\PCSP(\Gamma')$ for a suitable promise template $\Gamma'$. Therefore, we consider for each relation $(P, Q) \in \Gamma$ of arity $k$ the variable set $V = \{x_1, \hdots, x_k\}$. For every possible clause $C = (y_1, \hdots, y_k) \in (V \cup \bar{V} \cup \B)^k$ using variables $S \subseteq V$, we construct a corresponding relation pair $(P_C, Q_C)$ of arity $S$ such that
\begin{align*}
    P_C &= \{z \in \B^{S} : C(z) \in P\}\\
    Q_C &= \{z \in \B^{S} : C(z) \in Q\}.
\end{align*}
That is, $\Gamma'$ consists of every $(P_C, Q_C)$ which can be defined from some $(P, Q) \in \Gamma$. The key observation is that every instance $\Psi$ of $\fiPCSP(\Gamma)$ can then be replaced by a suitable instance $\Psi'$ of $\PCSP(\Gamma')$ by replacing each use of $(P,Q)$ with the corresponding $(P_C, Q_C)$. One can verify that each strong and weak assignment to $\Psi$ has the same value as the corresponding assignment to $\Psi'$ and that the SDP relaxations have the same value. Thus, Raghavendra's theorem does indeed apply.
\end{remark}
}

Thus, to show~\Cref{thm:main-hardness}, our goal is to show the existence of finite integrality gaps for the basic SDP relaxations of the Boolean {promise templates} in~\Cref{lem:reduction}.
To obtain such an integrality gap for the basic SDP relaxation of a PCSP, we study colorings of the $n$ dimensional sphere $\mathbb{S}^n$ that satisfy certain properties. 
We start by defining a few notations that we need. 

\begin{definition}
\label{def:p-conf}
Fix a predicate $P \subseteq {\B}^k$. We say that a tuple of vectors $V=(\textbf{v}_1, \textbf{v}_2, \ldots, \textbf{v}_k), \textbf{v}_i \in \S^n\,\,\forall i \in [k]$ is a $P$-configuration with respect to another vector $\textbf{v}_0\in \S^n$ if the tuple of vectors can be assigned to a set of literals in a constraint of the basic SDP relaxation of an instance of ${\fiPCSP(P,Q)}$ with zero error, for some $Q \supseteq P$. 
In other words, there exists a probability distribution $\{\lambda(\textbf{a}): \textbf{a} \in P\}$ supported on $P$ that satisfies the following properties. 
\begin{enumerate}
    \item $0 \leq \lambda(\textbf{a})\leq 1$ for all $\textbf{a} \in P$, and 
    \[
    \sum_{\textbf{a} \in P}\lambda(\textbf{a})=1.
    \]
    \item First moments: 
    \[
    { \langle \textbf{v}_i , \textbf{v}_0 \rangle }= \sum_{ \textbf{a} \in P}\lambda(\textbf{a}) a_i\quad \forall i \in [k].
    \]
    \item Second moments: 
     \[{ \langle \textbf{v}_{i} , \textbf{v}_{i'} \rangle }= \sum_{\textbf{a} \in P}\lambda(\textbf{a}) a_i a_{i'} \quad \forall i,i' \in [k]. \]
\end{enumerate}
\end{definition}

We now define the notion of a function respecting a Boolean ${\fiPCSP}$. We refer to functions $f:\S^n \rightarrow {\B}$ as colorings of the sphere. 

\begin{definition}
Fix a vector $\textbf{v}_0 {\in \S^n.}$ {We} say that a coloring of the sphere $f:\mathbb{S}^n \rightarrow {\B}$ {is folded if} $f(-\bv) = -f(\bv)$ for every $\bv \in \S^n${.} {We say that $f$ is idempotent if $f(\bv_0) = +1$ and $f(-\bv_0) = -1$.} {We say that a coloring $f$ \emph{respects} the promise template} $(P,Q)$ with respect to a vector $\textbf{v}_0\in\S^n$ if the following condition holds. For every $P$-configuration $V=(\textbf{v}_1, \textbf{v}_2, \ldots, \textbf{v}_k)$ with respect to $\textbf{v}_0$, we have that the colors of the vectors satisfy $Q$, i.e.,
\[
 (f(\textbf{v}_1), f(\textbf{v}_2), \ldots, f(\textbf{v}_k)) \in Q
\]
More generally, we say that a coloring $f:\mathbb{S}^n \rightarrow {\B}$ respects a Boolean {promise template} $\Gamma$ with respect to a vector $\textbf{v}_0$ if it respects every predicate pair in ${(P,Q)\in \Gamma}$ with respect to $\textbf{v}_0$.
\end{definition}

Our key observation is that the absence of such a {folded} sphere coloring respecting $\Gamma$ for some finite $n$ gives an integrality gap for the basic SDP relaxation of ${\fiPCSP(\Gamma)}$. 
\begin{lemma}
\label{lem:sphere-coloring}
For every Boolean {promise template} $\Gamma=\{(P_1,Q_1),\ldots,(P_l,Q_l)\}$, the Basic SDP decides ${\fiPCSP(\Gamma)}$ if and only if for every integer $n \geq 1$, there exists a {folded, idempotent} coloring $f^{(n)}:\mathbb{S}^n \rightarrow {\B}$ that respects $\Gamma$ with respect to a vector $\textbf{v}_0^{(n)}$.
\end{lemma}
\begin{proof}
We slightly abuse the notation and say that a {``partial''} folded function $f: S\rightarrow {\B}, S \subseteq \S^n$ respects a Boolean {promise template} $(P,Q)$ {with respect to} $\textbf{v}_0$ if and only if for every $P$-configuration of vectors $V=(\textbf{v}_1,\ldots,\textbf{v}_k)$ {with respect to} $\textbf{v}_0$ with $\textbf{v}_i \in S\,\forall i \in [k]$, $f(V) \in Q$. More generally, for a Boolean {promise template} $\Gamma$, $f:S \rightarrow {\B}$ respects $\Gamma$ if and only if it respects every predicate pair in $\Gamma$.
Via a compactness\footnote{We assume the axiom of choice.} argument (e.g., like the De Brujin-Erdos theorem \cite{bruijn1951colour}, for more details see Remark 7.13 of \cite{BBKO21} or \cite{ciardo2022clap}), we can infer that there is a {folded} coloring $f:\mathbb{S}^n \rightarrow {\B}$ respecting $\Gamma$ {with respect to} $\textbf{v}_0$ if and only if for every finite subset $S \subset \mathbb{S}^n$, there exists a {folded} coloring $f_S:S\rightarrow {\B}$ that respects $\Gamma$ {with respect to} $\textbf{v}_0$. 

First, assume that the Basic SDP decides ${\fiPCSP}(\Gamma)$. Fix an arbitrary set of vectors $\textbf{v}_0^{(n)} {\in \mathbb{S}^n}, n \in \mathbb{Z}^{+}$. {Also pick a subset $T \subset \mathbb{S}^n$ such that for every $\bv \in \mathbb{S}^n$, exactly one of $\bv$ or $-\bv$ is in $T$.} For any finite subset $S \subset \mathbb{S}^n$, we construct an instance $\Phi_S$ of $\Gamma$ as follows. {By increasing the size of $S$, we may assume that $S$ is closed with respect to negation.} The variable set {of $\Phi_S$ is $\{x_{\textbf{v}} : \textbf{v}\in S \cap T\}$.} {We define the constraints as follows. For every $i \in [l]$ and every $P_i$-configuration $(\textbf{v}_1, \textbf{v}_2, \ldots, \textbf{v}_{k_i}) \in S^{k_i}$ with respect to $\bv_0^{(n)}$, we} add a constraint over $(P_i,Q_i)$ using the {literals $\ell_{\bv_{1}}, \ldots, \ell_{\bv_{k_i}}$, where
\[
    \ell_{\bv_j} = \begin{cases}
    +1 & \bv_j = \bv_0^{(n)}\\
    -1 & \bv_j = -\bv_0^{(n)}\\
    x_{\bv_j} & \bv_j \in T\\
    \bar{x}_{-\bv_j} & -\bv_j \in T.
    \end{cases}
\]
}
We have that $x_{\bv} \mapsto \bv$ is a basic SDP solution with zero error. Thus, there exists an assignment {$f_{S \cap T} : S \cap T \to \mathcal B$} to the variables that weakly satisfies all the constraints in $\Phi_S$, or equivalently, there exists {folded, idempotent} $f_S: S\rightarrow {\B}$ that respects the {promise template} $\Gamma$ {with respect to} $\textbf{v}_0^{(n)}${, where
\[
    f_S(\bv) = \begin{cases}
    f_{S \cap T}(\bv) & \bv \in S \cap T\\
    -f_{S\cap T}(\bv) & \text{otherwise.}
    \end{cases}
\]}
{Hence,} there exists a {folded, idempotent} coloring $f:\mathbb{S}^n \rightarrow {\B}$ that respects $\Gamma$ {with respect to} $\textbf{v}_0^{(n)}$ for every positive integer $n$.

Second, suppose that for every integer $n \geq 1$, there exists a {folded, idempotent} coloring $f^{(n)}:\mathbb{S}^n \rightarrow {\B}$ that respects $\Gamma$ {with respect to} some vector $\textbf{v}_0^{(n)}$. We seek to show that the Basic SDP decides ${\fiPCSP}(\Gamma)$. Take an arbitrary instance $\Phi$ of ${\fiPCSP}(\Gamma)$ such that there is a solution to Basic SDP with zero error. We solve the SDP relaxation of $\Phi$ and obtain a set of vectors $\textbf{v}_0$ and $ \textbf{v}_1,\ldots, \textbf{v}_n \in \S^n$ corresponding to the variables in $\Phi$ that satisfies all the constraints in $\Phi$ with zero error. {By \Cref{prop:gram-ortho}, there exists an orthogonal matrix $Q \in \R^{(n+1)\times(n+1)}$ for which $Q\bv_0 = \bv_{0}^{(n)}$. Since orthogonal matrices preserve the inner products of vectors, we have that multiplication on the left by $Q$ is an automorphism of $\mathbb{S}^n$. Define $f'^{(n)} : \S^n \to \B$ by $f'^{(n)}(\bv) = f^{(n)}(Q\bv)$. Thus, we have that $f'^{(n)}$ respects $\Gamma$ with respect to $\bv_0$ if and only if $f^{(n)}$ respects $\Gamma$ with respect to $\bv^{(n)}_0$. Therefore, the} assignment $f'^{(n)}(\textbf{v}_i)$ to the variable $u_i$ weakly satisfies all the constraints in $\Phi$. Thus, for every instance $\Phi$ of ${\fiPCSP}(\Gamma)$ with zero error on the basic SDP relaxation, there is an assignment that weakly satisfies all the constraints in $\Phi$, or equivalently, the basic SDP decides ${\fiPCSP(\Gamma)}$. \end{proof}

As a corollary, we get the following. 
\begin{corollary}
\label{cor:basic-bool}
Let $\Gamma$ be a Boolean {promise template}. Then, there is a finite integrality gap for the basic SDP relaxation of ${\fiPCSP(\Gamma)}$ if and only if for some positive integer $n$, there exists no folded{, idempotent} coloring $f: \S^n \rightarrow {\B}$ that respects $\Gamma$.
\end{corollary}

~\Cref{thm:raghavendra} together with~\Cref{cor:basic-bool} shows that if for a Boolean {promise template} $\Gamma$ does not admit a sphere coloring $f: \mathbb{S}^n \rightarrow {\B}$ that respects $\Gamma$ for some positive integer $n$, then, ${\fiPCSP(\Gamma)}$ does not admit a polynomial time robust algorithm, assuming the Unique Games Conjecture. Thus, our goal is to show that the {promise templates} mentioned in~\Cref{lem:reduction} do not admit sphere coloring that respects them, and use~\Cref{cor:basic-bool} to prove~\Cref{thm:main-hardness}.

In the rest of this section, we first prove a couple of lemmas regarding sphere Ramsey theory. Then, we show that the earlier mentioned {promise templates} $\Gamma_{1\text{---}5}$ do not have folded sphere coloring respecting them using the sphere Ramsey results. 
\subsection{Sphere Ramsey theory}

We start with {some more notation}. For a tuple of vectors $S=(\textbf{v}_1, \textbf{v}_2,\ldots,\textbf{v}_k)$ with $\textbf{v}_i \in \S^{d}$, we use $\rho(S)$ to denote the {the smallest radius of a sphere} that contains $S$ as a subset. 
\[
\rho(S):=\min\{r: \exists\, \textbf{c}\in \R^d, \norm{\textbf{c}-\textbf{v}_i}_2=r\,\forall i \in [k]\}.
\]
Let $S_1 = (\textbf{u}_1, \textbf{u}_2,\ldots,\textbf{u}_k)$, $S_2 = (\textbf{v}_1, \textbf{v}_2,\ldots,\textbf{v}_k)$ with $\textbf{u}_i \in \S^{d_1}$, $\textbf{v}_i \in \S^{d_2}$ be two tuples with the same arity. We say that $S_1$ and $S_2$ are congruent if they have the same pairwise inner products, i.e., ${ \langle \textbf{u}_i , \textbf{u}_{j} \rangle =\langle \textbf{v}_i , \textbf{v}_{j} \rangle }$ for all $i,j\in [k]$.
Matoušek and Rödl~\cite{matouvsek1995ramsey} proved the following:

\begin{theorem}[\cite{matouvsek1995ramsey}]
\label{thm:sphere-ramsey}
Let $S=(\textbf{u}_1, \textbf{u}_2,\ldots,\textbf{u}_k)$ be a tuple of affinely independent vectors with $\rho(S)<1$. Then, for every positive integer $r \geq 2$, there exists an integer $n_0 := n_0(S,r)$ such that for every $n \geq n_0$, for every {coloring} $f:\mathbb{S}^n \rightarrow [r]$, there exists a tuple of vectors $S'=(\textbf{v}_1, \textbf{v}_2,\ldots,\textbf{v}_k), \textbf{v}_i\in \mathbb{S}^n\,\forall i \in [k]$ that is congruent to $S$, and is monochromatic, i.e., $f(\textbf{v}_i)=f(\textbf{v}_j)$ for every $i,j \in [k]$.
\end{theorem}

We will use this to show the following lemma regarding sphere colorings.
\begin{lemma}
\label{lem:sphere-coloring-1}
Fix an integer $k \geq 3$ and $r \geq 2$. There exists $n_0 := n_0(k{, r})$ such that for every $n \geq n_0$ and coloring $f: \mathbb{S}^n \rightarrow [r]$ and $\gamma \in \R$ with $ \frac{-1}{k-1} <\gamma <1$, there exists a monochromatic set of vectors $V = \{ \textbf{v}_1, \textbf{v}_2, \ldots, \textbf{v}_k \} \subseteq \mathbb{S}^n$ such that ${ \langle \textbf{v}_i , \textbf{v}_j \rangle} = \gamma$ for every $i \neq j$. 
\end{lemma}
\begin{proof}
Consider an arbitrary set $S = \{ \textbf{u}_1, \textbf{u}_2, \ldots, \textbf{u}_k \}$ of $k$ unit vectors in $\mathbb{S}^n$ such that ${ \langle \textbf{u}_i , \textbf{u}_j \rangle }= \gamma$ for every $i \neq j$. Such a set $S$ is guaranteed to exist when $n$ is large enough. We show that the vectors are affinely independent: suppose for contradiction that there exists reals $c_1, c_2, \ldots, c_k$ not all zero, $\sum_i c_i = 0$ and $\sum_i c_i \textbf{u}_i = 0$.
We have 
\[
0={ \left\langle \textbf{u}_1 , \left( \sum_i c_i \textbf{u}_i \right) \right\rangle }  = c_1 + \gamma ( c_2 + \ldots + c_k) = c_1 + \gamma ( - c_1) 
\]
implying that $c_1=0$. The same argument shows that $c_i=0$ for all $i \in [k]$, a contradiction. 

The set of vectors can be embedded on a sphere of radius strictly smaller than $1$. Let $\alpha \in \R$ be such that $0<\alpha <\frac{2}{k}$, and let $\textbf{u}_s = \sum_{i\in [k]}\textbf{u}_i$, $\textbf{c} = \alpha \textbf{u}_s$. We have 
\[
\norm{\textbf{u}_s}_2^2 = \sum_i \norm{\textbf{u}_i}_2^2 + 2 \sum_{i \neq j} { \langle \textbf{u}_i , \textbf{u}_j \rangle }= {k +k(k-1) \gamma }
\]
Note that 
\begin{align*}
\norm{ \textbf{u}_i - \textbf{c} }_2^2 &=  \norm{\textbf{u}_i}_2^2 + \norm{\textbf{c}}_2^2 - 2 { \langle \textbf{c} , \textbf{u}_i \rangle} \\ 
&= 1+\alpha^2 ({ k +k(k-1) \gamma })- 2 \alpha (1+(k-1)\gamma) \\ 
&= 1 - k(1+(k-1)\gamma)\alpha \left( \alpha - \frac{2}{k}\right) 
\end{align*}
which is strictly smaller than $1$ when $0 < \alpha <\frac{2}{k}$. Thus, all the vectors are on a sphere centered at $\textbf{c}$ and radius strictly smaller than $1$, implying that $\rho(S)<1$. Now, we can use~\Cref{thm:sphere-ramsey} on $S$ and $f$ to obtain the required set of vectors $V$. 
\end{proof}

While~\Cref{thm:sphere-ramsey} is applicable to a wide range of sets $S$, we sometimes need to apply it to sets $S$ that do not form a simplex or have $\rho(S)=1$. Towards this, we use the ``Spreads'' based idea in~\cite{matouvsek1995ramsey} to obtain a version of~\Cref{thm:sphere-ramsey} directly for certain sets $S$ where~\Cref{thm:sphere-ramsey} is not applicable. 

We use the following notion of $\Spread$ vectors from~\cite{matouvsek1995ramsey}. For an integer $n$, a vector $\textbf{a} \in \R^k$, and a set $J \subseteq [n]$ of cardinality $k$ with $J=\{ j_1, j_2, \ldots, j_k\}$ {where $j_1 < \cdots < j_k$},
we let 
\[
\SpreadSingle_n(\textbf{a}, J) = { \sum_{i=1}^k a_i \textbf{e}_{j_i} }
\]
where {$\textbf{e}_1, \textbf{e}_2, \ldots, \textbf{e}_n$} is an orthonormal basis of $\R^n$. 
For a set $I \subseteq [n]$, we let 
\[
\Spread_n(\textbf{a}, I) = \{ \SpreadSingle_n (\textbf{a}, J) : J \subseteq I, |J|=k \}
\]
{Recall the hypergraph Ramsey theorem: for any $n, k, r$ there exists $N$ such that in any $r$-coloring of the hyperedges of the complete $k$-uniform hypergraph on $N$ vertices, there exists an induced subgraph on $n$-vertices which is monochromatic.} {As a direct consequence, we get the following application to spreads.} 
\begin{lemma}(~\cite{matouvsek1995ramsey})
\label{lem:spreads-ramsey}
For every $n, k, \textbf{a} \in \R^k$, there exists $N$ such that in any coloring $f: \Spread_N(\textbf{a}, [N]) \rightarrow [r]$, there exists $I$ with $|I|=n$ such that $\Spread_N(\textbf{a},I)$ is monochromatic with respect to $f$, i.e., $\exists p \in [r]$ such that $f(\textbf{v})=p$ for all $\textbf{v} \in \Spread_N(\textbf{a},I)$. 
\end{lemma}

~\Cref{lem:spreads-ramsey} implies the following immediately. 
\begin{corollary}
\label{coroll:spreads-ramsey}
Let $U=\{\textbf{u}_1, \textbf{u}_2, \ldots, \textbf{u}_k \}$ be a set of $k$ unit vectors such that $\textbf{u}_i \in \Spread_N(\textbf{a},[N]) $ for all  $i \in [k]$ for an integer $N$, and a vector $\textbf{a} \in \R^N$ with $\norm{\textbf{a}}_2=1$. Then there exists $n_0 := n_0(U,\textbf{a},N)$ such that for every $n \geq n_0,r$, for every sphere coloring $f: \S^n \rightarrow [r]$, there exists a set of $k$ vectors $V=\{\textbf{v}_1, \ldots, \textbf{v}_k\}$ that are all colored the same, and ${ \langle \textbf{v}_i , \textbf{v}_j \rangle= \langle \textbf{u}_i , \textbf{u}_j\rangle }$ for every $i,j \in [k]$. 
\end{corollary}

We use~\Cref{coroll:spreads-ramsey} to prove a lemma regarding sphere colorings. For ease of notation, we call a set of $k$ unit vectors $V=\{ \textbf{v}_1, \textbf{v}_2, \ldots, \textbf{v}_k \}$ to be \textit{$k$-regular} if ${ \langle \textbf{v}_i , \textbf{v}_j \rangle }= -\frac{1}{k-1}$ for every $i \neq j$. 

\begin{lemma}
\label{lem:sphere-coloring-2}
Fix an integer $k \geq 2$. There exists $n_0 := n_0(k)$ such that for every $n \geq n_0$ and folded coloring $f: \mathbb{S}^n \rightarrow {\B}$, there exist a $k$-regular set of vectors $V = \{ \textbf{v}_1, \textbf{v}_2, \ldots, \textbf{v}_k \} \subseteq \mathbb{S}^n$ such that exactly $k-1$ vectors in $V$ are colored $-1$. 
\end{lemma}

\begin{proof}
We construct a set of $k$ unit vectors $V=\{\textbf{v}_1, \textbf{v}_2, \ldots, \textbf{v}_k \}$ in $\Spread_N(\textbf{a}, [N])$ such that  the set of vectors $\{\textbf{v}_1, \textbf{v}_2, \ldots, \textbf{v}_{k-1}, -\textbf{v}_k\}$ is a $k$-regular set, where $N$ and $\textbf{a}$ depend only on $k$, and $\norm{\textbf{a}}_2=1$. Using~\Cref{coroll:spreads-ramsey}, we can infer that in the coloring $f$, there exist $k$ vectors $V=\{\textbf{v}_1, \textbf{v}_2, \ldots, \textbf{v}_k \}$ that are all assigned the same color, such that $\{\textbf{v}_1, \textbf{v}_2, \ldots, \textbf{v}_{k-1}, -\textbf{v}_k\}$ is a $k$-regular set. As $f$ is folded, this implies that there is a $k$-regular set in which exactly $k-1$ vectors are assigned the color $-1$. 

Thus our goal is to construct $k$ unit vectors $V=\{\textbf{v}_1, \textbf{v}_2, \ldots, \textbf{v}_k \}$ in $\Spread_N(\textbf{a}, [N])$ such that the set of vectors $\{\textbf{v}_1, \textbf{v}_2, \ldots, \textbf{v}_{k-1}, -\textbf{v}_k\}$ is a $k$-regular set. Or equivalently, we construct the vectors $\textbf{v}_1, \textbf{v}_2, \ldots, \textbf{v}_{k-1} $ in $\Spread_N(\textbf{a},[N])$ and $\textbf{v}_k$ in $\Spread_N(-\textbf{a},[N])$ such that $\{ \textbf{v}_1, \ldots, \textbf{v}_k \}$ is a $k$-regular set. 
We set $\gamma = \frac{1}{\sqrt{2(k-1)}}$ and $\textbf{a} = (\gamma, -\gamma, \gamma, -\gamma, \ldots, -\gamma) \in \R^{2(k-1)}$.
We set $\textbf{v}_i = \Spread_n (\textbf{a}, J_i), i \in [k-1], \textbf{v}_k = \Spread(-\textbf{a}, J_k)$ where $J_1, J_2, \ldots, J_k$ such that $|J_i|=2(k-1)$ for every $i \in [k]$.  
We obtain these sets by induction on $k$. First, we consider the base case when $k=2$. In this case, we set $J_1 = J_2 = \{ 1,2\}$ and $N=2$ suffices. The vectors are the following:
\begin{align*}
    \textbf{v}_1 &= (\gamma, -\gamma) \\
    \textbf{v}_2 &= (-\gamma, \gamma)
\end{align*}
where $\gamma = \frac{1}{\sqrt{2}}$. Note that the above two vectors are a $2$-regular set, and letting $\textbf{a}=(\gamma, -\gamma)$, we have  $\textbf{v}_1 \in \Spread_2(\textbf{a},[2])$, and $ \textbf{v}_2 \in \Spread_2(-\textbf{a},[2])$.
Now, suppose that $J_1, J_2, \ldots, J_k, N$ are such that $\textbf{v}_i = \Spread_N(\textbf{a},J_i), i \in [k-1], \textbf{v}_k = \Spread_N(-\textbf{a},J_k)$ satisfy the property that $\{ \textbf{v}_1, \textbf{v}_2, \ldots, \textbf{v}_k \}$ is a $k$-regular set with $\textbf{a}=(\gamma, -\gamma, \ldots, \gamma, -\gamma) \in \R^{2(k-1)}$, $\gamma = \frac{1}{\sqrt{2(k-1)}}$. 
We construct $J'_1, \ldots, J'_{k+1}$ such that $\textbf{v}'_i = \Spread_{N'}(\textbf{a}',J'_i)$ for all $i \in [k]$, $\textbf{v}'_{k+1} = \Spread_{N'}(-\textbf{a}',J_{k+1})$ satisfy the property that $\{ \textbf{v}'_1, \textbf{v}'_2, \ldots, \textbf{v}'_{k+1} \}$ is a $(k+1)$-regular set with $\textbf{a}'=(\gamma', -\gamma', \ldots, \gamma', -\gamma') \in \R^{2k}$, $\gamma' = \frac{1}{\sqrt{2k}}$.
\begin{enumerate}
    \item For every $i \in [k-1]$, we obtain $J'_i$ from $J_i$ by adding two new elements. 
    \[
    J'_i = J_i \cup \{ N+2i, N+2i+1 \} 
    \]
    This ensures that ${ \langle \textbf{v}'_i , \textbf{v}'_j \rangle }= -(\gamma')^2$ for every $i,j \in [k-1], i \neq j$. 
    \item We obtain $J_{k+1}$ from $J_k$ by adding two new elements. 
    \[
    J_{k+1} = J_k \cup \{ N+1, N+2k\}
    \]
    This ensures that ${ \langle \textbf{v}'_i , \textbf{v}'_{k+1} \rangle }= -(\gamma')^2$ for every $i \in [k-1]$. 
    \item Finally, we set $J_k$. 
    \[
    J_k = \{ N+1, N+2, \ldots, N+2k\}
    \]
    This ensures that ${ \langle \textbf{v}'_i , \textbf{v}'_{k} \rangle }= -(\gamma')^2$ for every $i \in [k+1],i \neq k$.
\end{enumerate}
We illustrate our construction by obtaining the vectors for the case when $k=3$ and $k=4$: 
{
\begin{align*}
    \begin{alignedat}{8}
        \textbf{v}_1 &&= 
        (&&\alpha, &&-\alpha, &&0, &&\alpha, &&-\alpha, &&0) \\
        \textbf{v}_2 &&= 
        (&&0, &&0, &&\alpha, &&-\alpha, &&\alpha, &&-\alpha) \\
        \textbf{v}_3 &&= 
        (&&-\alpha, &&\alpha, &&-\alpha, &&0, &&0, &&\alpha)
    \end{alignedat}
\end{align*}
}
{
\begin{align*}
    \begin{alignedat}{14}
        \textbf{v}_1 &&= 
        (&&\beta, &&-\beta, &&0, &&\beta, &&-\beta, &&0, &&0, &&\beta, &&-\beta, &&0, &&0, &&0) \\
        \textbf{v}_2 &&= 
        (&&0, &&0, &&\beta, &&-\beta, &&\beta, &&-\beta, &&0, &&0, &&0, &&\beta, &&-\beta, &&0) \\
        \textbf{v}_3 &&= 
        (&&0, &&0, &&0, &&0, &&0, &&0, &&\beta, &&-\beta, &&\beta, &&-\beta, &&\beta, &&-\beta) \\
        \textbf{v}_4 &&= 
        (&&-\beta, &&\beta, &&-\beta, &&0, &&0, &&\beta, &&-\beta, &&0, &&0, &&0, &&0, &&\beta)
    \end{alignedat}
\end{align*}
}

where $\alpha=\frac{1}{2}$ and $\beta=\frac{1}{\sqrt{6}}$.

As the pairwise inner product of every pair in $\{\textbf{v}'_1, \textbf{v}'_2, \ldots, \textbf{v}'_{k+1}\}$ is equal to $-(\gamma')^2=-\frac{1}{k}$, we get that these set of vectors are a $(k+1)$-regular set, completing the inductive proof.
\end{proof}

\subsection{Absence of sphere coloring via sphere Ramsey theory}

First, we show the absence of sphere coloring respecting $\Gamma_1$ using~\Cref{lem:sphere-coloring-2}. 
\begin{lemma}
\label{lem:gamma1}
Fix an even integer $k \geq 4$. There exists an integer $n_0$ such that for every $n \geq n_0$, there is no folded $f:\mathbb{S}^n \rightarrow {\B}$ that respects  $\Gamma_1 = (P,Q)$, $P=\Ham_{k}\{\frac{k}{2}\}$, $Q = \Ham_{k} \{0,1,\ldots,k\} \setminus \{b\}$ where $b \in \{1,k-1\}$.
\end{lemma}
\begin{proof}
Consider a large integer $n$ and suppose for the sake of contradiction that there is a folded function $f: \S^n \rightarrow {\B}$ that respects $\Gamma_1$ with respect to a vector $\textbf{v}_0 \in \mathbb{S}^n$. 
We get the $P$-configuration of vectors $\textbf{v}_1, \textbf{v}_2, \ldots , \textbf{v}_k$ where we set $\lambda(\textbf{a})=\frac{1}{|P|}$ for every $\textbf{a}\in P$ in~\Cref{def:p-conf}.
The vectors satisfy the following properties. 
\begin{enumerate}
    \item (First moments.) ${ \langle \textbf{v}_i , \textbf{v}_0 \rangle }= 0$ for every $i \in [k]$. 
    \item (Second moments.) ${ \langle \textbf{v}_i , \textbf{v}_j \rangle }= \frac{2 \binom{{k/2}}{2} - \frac{k^2}{4}}{\binom{k}{2}}  = \frac{-1}{k-1}$ {for every $i \neq j \in [k]$}.
\end{enumerate}
Our goal is to show that there is a $P$-configuration of such vectors such that exactly $b$ of them are colored $+1$ according to $f$. 
Consider the set of vectors 
\[
\textbf{v}_0^{\perp}:= \{ \textbf{u} \in \S^n : { \langle \textbf{u}, \textbf{v}_0 \rangle }= 0\}
\]
Using~\Cref{lem:sphere-coloring-2}, we can obtain a set of $k$ vectors $\textbf{u}_1, \textbf{u}_2, \ldots, \textbf{u}_{k} \in \textbf{v}_0^{\perp}$ such that ${ \langle \textbf{u}_i , \textbf{u}_j \rangle }= \frac{-1}{k-1}$ and exactly $k-1$ of $\{ \textbf{u}_1, \textbf{u}_2, \ldots, \textbf{u}_k \}$ are colored $-1$. 
\[
\left|\{i:f(\textbf{u}_i)=-1\}\right|=k-1.
\]
Note that both $\{\textbf{u}_1, \textbf{u}_2, \ldots, \textbf{u}_k \}$ and $\{-\textbf{u}_1, -\textbf{u}_2, \ldots, -\textbf{u}_k \}$ are $P$-configurations with respect to $\textbf{v}_0$. Since $f$ is folded, in at least one of these two $P$-configurations, 
there are exactly $b$ vectors that are colored $+1$, a contradiction.
\end{proof}

We show the absence of sphere coloring respecting $\Gamma_2$, $\Gamma_3$, and $\Gamma_4$ using~\Cref{lem:sphere-coloring-1}.

\begin{lemma}
\label{lem:gamma2}
Fix an odd integer $k \geq 3$ and integer $l : 0 \leq l \leq \frac{k-1}{2}$. There exists an integer $n_0$ such that for every $n \geq n_0$, there is no folded $f:\mathbb{S}^n \rightarrow{\B}$ that respects $\Gamma_2 =  (P,Q)$, $P=\Ham_k \{l,\frac{k+1}{2}\}$, $Q=\Ham_k \{0,1,2,\ldots,k-1\}$.
\end{lemma}
\begin{proof}
Consider a large integer $n$ and suppose for contradiction that there is a folded function $f : \S^n \rightarrow {\B}$ that respects $\Gamma_2$ with respect to a vector $\textbf{v}_0 \in \mathbb{S}^n$.  
The $P$-configuration that we consider is a set of vectors $\textbf{v}_1, \textbf{v}_2, \ldots, \textbf{v}_k$ that are obtained by setting $\lambda(\textbf{a})$ in~\Cref{def:p-conf} as follows. We first sample an integer $t \in \{l, \frac{k+1}{2}\}$ as below.
\[
t = \begin{cases}
l,\text{ with probability}\frac{1}{1-s}. \\ 
\frac{k+1}{2},\text{ with probability}\frac{-s}{1-s}.
\end{cases}
\]
where $s = l-(k-l)<0$.
The probability distribution $\lambda$ is obtained by sampling a uniform element of $\Ham_k\{t\}$. In other words, we have 
\[
\lambda(\textbf{a}):= \begin{cases} 
\dfrac{1}{(1-s)\dbinom{k}{l}}, \text{ if }\textbf{a} \in \Ham_k \{l\}.\\
\dfrac{-s}{(1-s)\dbinom{k}{\frac{k+1}{2}}}, \text{ else if }\textbf{a} \in \Ham_k \{\frac{k+1}{2}\}.\\
0, \text{ otherwise.}
\end{cases}
\]
We obtain the following properties: 
\begin{enumerate}
    \item (First moments).  { By symmetry, $\langle \textbf{v}_i, \textbf{v}_0 \rangle = \langle \textbf{v}_j, \textbf{v}_0 \rangle $ for all $i, j \in [k]$.}
    
We have     
    ${ \langle \textbf{v}_i , \textbf{v}_0 \rangle = \frac{1}{k} \langle (\sum_{i \in [k]} \textbf{v}_i , \textbf{v}_0 \rangle = 
    \frac{1}{k} \left(\frac{1}{1-s}(l-(k-l))+\frac{-s}{1-s}1\right)=0}$ for every $i \in [k]$. 
    \item (Second moments). By symmetry of variables, we get that ${ \langle \textbf{v}_i , \textbf{v}_j \rangle}= \gamma$ for every $i \neq j$, for some $\gamma:=\gamma(k,l)$. 
    We have 
    \begin{align*}
        k+k(k-1)\gamma=\norm{\sum_i \textbf{v}_i}^2 = {\sum_{\textbf{a}\in P}\lambda(\textbf{a}) \left(\sum_{i \in [k]}a_i\right)^2}>0.
    \end{align*}
    Thus, we get that ${\frac{-1}{k-1}} < \gamma < 1$. 
\end{enumerate}
Now, restricting ourselves to the vectors in $\mathbb{S}^n$ that are orthogonal to $\textbf{v}_0$, and using~\Cref{lem:sphere-coloring-1}, we get that there exists a $P$-configuration of vectors that are all colored the same. By taking the negation of these vectors if needed, we get our required claim.
\end{proof}

\begin{lemma}
\label{lem:gamma3}
Fix integers $k, l $ such that $0 < l \leq \frac{k-1}{2}$. Then, there exists an integer $n_0$ such that for every $n \geq n_0$, there is no folded $f:\mathbb{S}^n \rightarrow {\B}$ that respects $\Gamma_3 = (P,Q)$, $P=\Ham_k \{l,k\}$, $Q=\Ham_k \{1,2,\ldots,k\}$, where $l\neq 0, l \leq \frac{k-1}{2}$. 
\end{lemma}
\begin{proof}
Consider a large integer $n$ and suppose for contradiction that there is a folded function $f : \S^n \rightarrow {\B}$ that respects $\Gamma_3$ with respect to a vector $\textbf{v}_0 \in \mathbb{S}^n$.  
 We pick the $P$-configuration along the same lines as in~\Cref{lem:gamma2}. We sample $t \in \{l,k\}$ with 
\[
t = \begin{cases}
l \text{ with probability }\frac{k}{k-s}\\ 
k \text{ with probability }{\frac{-s}{k-s}}
\end{cases}
\]
where $s = l-(k-l)<0$. As before, the probability distribution $\lambda$ is obtained by sampling a uniform element of $\Ham_k\{t\}$. We get 
\begin{enumerate}
    \item (First moments). ${ \langle \textbf{v}_i , \textbf{v}_0 \rangle = \frac{1}{k}\left\{\left( \frac{k}{k-s} \right) \frac{s}{k} + \left( \frac{-s}{k-s} \right) 1 \right\} }=0$ for every $i \in [k]$. 
    \item (Second moments). As in~\Cref{lem:gamma2}, we have 
    ${ \langle \textbf{v}_i , \textbf{v}_j \rangle }= \gamma$ for every $i \neq j$, for some $\gamma:=\gamma(k,l)$ with
    \begin{align*}
        k+k(k-1)\gamma=\norm{\sum_i \textbf{v}_i}^2 = {\sum_{\textbf{a}\in P}\lambda(\textbf{a}) (\sum_{i \in [k]}a_i)^2}>0.
    \end{align*}
    Thus, we get that ${\frac{-1}{k-1}} < \gamma < 1$.
\end{enumerate}
    We restrict ourselves to vectors in $\mathbb{S}^n$ that are orthogonal to $\textbf{v}_0$, and applying~\Cref{lem:sphere-coloring-1}, we get that for any coloring $f:\mathbb{S}^n \rightarrow {\B}$, there is a monochromatic $P$-configuration that we described. By negating the vectors if needed, we get our required proof. 
\end{proof}

For the proof of the next case $\Gamma_4$ we will need~\Cref{lem:sphere-coloring-1} applied to $4$-colorings of the sphere. 
\begin{lemma}
\label{lem:gamma4}
Fix integers $k \geq 3, l \in \{1,\ldots,k-1\}, l \leq \frac{k-1}{2}$. 
There exists integer $n_0$ such that for every $n \geq n_0$, there does not exist coloring $f:\mathbb{S}^n \rightarrow \{0,1\}$ that is folded and respects the {promise template} $\Gamma_4 = (P,Q),P=\Ham_k \{l\}, Q=\Ham_k \{0,1,\ldots,k\}\setminus \{0,k-1\}$.
\end{lemma}
\begin{proof}
As before, consider a large integer $n$ and suppose for contradiction that there is a folded function $f : \S^n \rightarrow {\B}$ that respects $\Gamma_4$ with respect to a vector $\textbf{v}_0 \in \mathbb{S}^n$.  
We partition the predicate $P$ into $P_1$ and $P_{-1}$ depending on the value of the first element, i.e., 
\[
P_i = \{ \textbf{x} \in P: x_1 = i \}, i \in {\B}
\]
We pick the probability distribution $\lambda$ as follows: sample $i$ from ${\B}$ uniformly at random, then, sample a uniformly random element from $P_i$. We get 
\begin{enumerate}
    \item (First moments). By our choice of $P_i$s, we get that 
    \[
    { \langle \textbf{v}_1 , \textbf{v}_0 \rangle}= 0
    \]
    By using the symmetry of the rest of the variables and the fact that $\sum_{i=1}^k { \langle \textbf{v}_i , \textbf{v}_0\rangle }=l-(k-l)=2l-k,$ we get that 
    \[
    { \langle \textbf{v}_i , \textbf{v}_0 \rangle }= \frac{2l-k}{k-1} \forall i \in \{2,3,\ldots,k\}.
    \]
    For ease of notation, let $\alpha = \frac{2l-k}{k-1}$. 
    \item (Second moments). 
    We have 
    \begin{align*}
        { \left\langle \textbf{v}_1 , \left( \sum_{i=2}^k \textbf{v}_i\right) \right\rangle }&= \frac{1}{2} \left( 1 \cdot (2l-k-1) \right)+ \frac{1}{2}\left( (-1) \cdot (2l-k+1)\right)\\
        &=-1
    \end{align*}
    Thus, we get 
    \begin{align*}
    { \left\langle \textbf{v}_1 , \textbf{v}_i \right\rangle } &= \frac{-1}{k-1}\,\forall i \in \{2,3,\ldots,k\}.
    \end{align*}
    Note that we have 
    \[
    \norm{\sum_{i=1}^k \textbf{v}_i}_2^2 = (2l-k)^2.
    \]
    Using this, we get that 
    \begin{align*}
    { \langle \textbf{v}_i , \textbf{v}_j \rangle } &=\frac{(2l-k)^2-(k-2)}{(k-1)(k-2)} \,\forall i,j \in \{2,3,\ldots,k\}, i \neq j
    \end{align*}
    For ease of notation, let $\beta = \frac{(2l-k)^2-(k-2)}{(k-1)(k-2)}$. 
\end{enumerate}

Our goal is to show that there exists $n_0$ such that for every $n \geq n_0$, for every folded sphere coloring $f: \mathbb{S}^n \rightarrow {\B}$ and $\textbf{v}_0 \in \mathbb{S}^n$, there exists a set of $k$ vectors $V=\{ \textbf{v}_1, \textbf{v}_2, \ldots, \textbf{v}_k \}$ that satisfy the above first and second moments, and exactly $b$ vectors in $V$ are colored $+1$, where $b \in \{0,k-1\}$. 

{
Let \(\gamma = \frac{\beta - \alpha^2}{1 - \alpha^2}\). We expand this as follows:
\begin{align*}
\gamma &= \frac{\frac{(2l-k)^2 - (k-2)}{(k-1)(k-2)} - \frac{(2l-k)^2}{(k-1)^2}}{1 - \alpha^2} \\[8pt]
&= \frac{\frac{(2l-k)^2}{(k-1)^2 (k-2)} - \frac{1}{k-1}}{1 - \alpha^2} \\[8pt]
&= \frac{\frac{\alpha^2}{k-2} - \frac{1}{k-1}}{1 - \alpha^2} = \frac{\frac{1 - (1 - \alpha^2)}{k-2} - \frac{1}{k-1}}{1 - \alpha^2}.
\end{align*}

Simplifying further, we obtain:
\begin{equation}
\label{eq:gamma-bound-1}
\gamma = \frac{1}{(1 - \alpha^2)(k-1)(k-2)} - \frac{1}{k-2} > -\frac{1}{k-2}.
\end{equation}

}
We apply~\Cref{lem:sphere-coloring-1} on the following coloring of the sphere. For a vector $\textbf{u} \in \S^n$ such that ${ \langle \textbf{u}, \textbf{v}_0\rangle }=0$, 
let $f':\S^{n-1} \rightarrow {\B}^2$ be defined as 
\[
f'(\textbf{u})=\left(f\left(\alpha \textbf{v}_0 + \sqrt{1-\alpha^2} \textbf{u}\right), f\left(\alpha \textbf{v}_0 - \sqrt{1-\alpha^2} \textbf{u}\right)\right)
\]
Using~\Cref{lem:sphere-coloring-1} on $f'$ combined with the fact that $\gamma > \frac{-1}{k-2}$ obtained from~\Cref{eq:gamma-bound-1}, we can infer that there exist $k-1$ unit vectors $\textbf{u}_1, \textbf{u}_2, \ldots, \textbf{u}_{k-1} \in \S^{n-1}$ such that ${\langle \textbf{u}_i , \textbf{v}_0 \rangle }=0$ for all $i$, ${\langle \textbf{u}_i , \textbf{u}_j \rangle }=\gamma$ for all $i \neq j$ and $f'(\textbf{u}_i)=f'(\textbf{u}_j)$ for all $i \neq j, i,j \in [k-1]$. 

We define $\textbf{v}^{(1)}_1, \textbf{v}^{(1)}_2, \ldots, \textbf{v}^{(1)}_k, \textbf{v}^{(2)}_1, \ldots, \textbf{v}^{(2)}_k$ as follows. 
For $i \in \{2,3,\ldots, k\}$, we let 
\begin{align*}
    \textbf{v}^{(1)}_i = \alpha \textbf{v}_0 + \sqrt{1-\alpha^2}\textbf{u}_{i-1}\\
    \textbf{v}^{(2)}_i = \alpha \textbf{v}_0 - \sqrt{1-\alpha^2}\textbf{u}_{i-1}\\
\end{align*}
We let 
\[
\textbf{v}^{(1)}_1 = -\frac{\sum_{i=1}^{k-1} \textbf{u}_i}{\norm{\sum_{i=1}^{k-1} \textbf{u}_i}}
\]
and $\textbf{v}^{(2)}_1=-\textbf{v}^{(1)}_1$. We now prove that the set of vectors $\textbf{v}^{(1)}_1, \ldots, \textbf{v}^{(1)}_k$ and the set of vectors $\textbf{v}^{(2)}_1, \ldots, \textbf{v}^{(2)}_k$ are a $P$-configuration with first and second moments as computed earlier, where we sampled $i$ from ${\B}$ uniformly at random and set the probability distribution $\lambda$ as the uniform distribution over $P_i$.
\begin{enumerate}
    \item (First moments). As ${\langle \textbf{u}_i , \textbf{v}_0 \rangle }= 0$ for all $i \in [k-1]$, we get that 
    \[
    {\langle \textbf{v}^{(1)}_1 , \textbf{v}_0 \rangle }= {\langle \textbf{v}^{(2)}_1 , \textbf{v}_0 \rangle }= 0
    \]
    and 
    \[
    {\langle \textbf{v}^{(1)}_i , \textbf{v}_0 \rangle }= {\langle \textbf{v}^{(2)}_i , \textbf{v}_0 \rangle }= \alpha \,\,\, \forall i \in \{2,\ldots,k\}
    \]
    \item (Second moments). We have $\forall i \in \{2,3,\ldots,k\}$,
    \begin{align*}
        { \langle \textbf{v}^{(1)}_1 , \textbf{v}^{(1)}_i \rangle }&= - \frac{{ \langle (\sum_{j=1}^{k-1}\textbf{u}_j) , (\alpha \textbf{v}_0 + \sqrt{1-\alpha^2}\textbf{u}_{i-1})\rangle }}{\norm{\sum_{j=1}^{k-1} \textbf{u}_j}}  \\ 
        &= \frac{\sqrt{1-\alpha^2}}{k-1} -\frac{{ \langle (\sum_{j=1}^{k-1}\textbf{u}_j) , (\sum_{j=1}^{k-1}\textbf{u}_j) \rangle }}{\norm{\sum_{j=1}^{k-1} \textbf{u}_j}} \\ 
        &= -\frac{\sqrt{1-\alpha^2}}{k-1}\norm{\sum_{j=1}^{k-1} \textbf{u}_j} \\ 
        &= -\frac{\sqrt{1-\alpha^2}}{k-1} \sqrt{ k-1+ 2 \frac{(k-1)(k-2)}{2}\gamma} \\ 
        &= -\frac{\sqrt{1-\alpha^2}}{k-1} \sqrt{ \frac{1}{1-\alpha^2}}  \text{                Using Equation}~(\ref{eq:gamma-bound-1}) \\
        &=\frac{-1}{k-1}  \ . 
    \end{align*}
    Furthermore, $\forall i,j \in \{2,3,\ldots,k\}$, $i \neq j$, 
    \begin{align*}
        { \langle \textbf{v}^{(1)}_i , \textbf{v}^{(1)}_j \rangle } &= { \langle ( \alpha \textbf{v}_0 + \sqrt{1-\alpha^2}\textbf{v}_i), ( \alpha \textbf{v}_0 + \sqrt{1-\alpha^2}\textbf{v}_j) \rangle }\\ &= \alpha^2 + (1-\alpha^2)\gamma \\ 
        &= \beta  \ .
    \end{align*}
    Similarly, we have  
    \begin{align*}
        { \langle \textbf{v}^{(2)}_1 , \textbf{v}^{(2)}_i \rangle } &=  \frac{{ \langle (\sum_{j=1}^{k-1}\textbf{u}_j) , (\alpha \textbf{v}_0 - \sqrt{1-\alpha^2}\textbf{u}_{i-1})\rangle }}{\norm{\sum_{j=1}^{k-1} \textbf{u}_j}}  = \frac{-1}{k-1}  \, \, \forall i \in \{2,3,\ldots,k\}.\\
        { \langle \textbf{v}^{(2)}_i , \textbf{v}^{(2)}_j \rangle } &= { \langle ( \alpha \textbf{v}_0 - \sqrt{1-\alpha^2}\textbf{v}_i), ( \alpha \textbf{v}_0 - \sqrt{1-\alpha^2}\textbf{v}_j)\rangle } \\ &= \alpha^2 + (1-\alpha^2)\gamma  = \alpha^2 + (1-\alpha^2)\gamma = \beta \, \, \forall i,j \in \{2,3,\ldots,k\}, i \neq j
    \end{align*}
\end{enumerate}
Thus, both the set of vectors $\textbf{v}^{(1)}_1, \textbf{v}^{(1)}_2, \ldots, \textbf{v}^{(1)}_k$ and the set of vectors $\textbf{v}^{(2)}_1, \textbf{v}^{(2)}_2, \ldots, \textbf{v}^{(2)}_k$ are $P$-configurations with respect to $\textbf{v}_0$. As $f'(\textbf{u}_1)=f'(\textbf{u}_2)=\ldots=f'(\textbf{u}_{k-1})$, we can infer that $f(\textbf{v}^{(1)}_2)=f(\textbf{v}^{(1)}_3)=\ldots=f(\textbf{v}^{(1)}_{k})$ and $f(\textbf{v}^{(2)}_2)=f(\textbf{v}^{(2)}_3)=\ldots=f(\textbf{v}^{(2)}_{k})$. Furthermore, as $\textbf{v}^{(1)}_1=-\textbf{v}^{(2)}_1$ and $f$ is folded, we can infer that $f(\textbf{v}^{(1)}_1)=-f(\textbf{v}^{(2)}_1)$. Thus, there exists $p \in \{1,2\}$ such that $f(\textbf{v}^{(p)}_1)=-1$. Thus, there are either $0$ or $k-1$ vectors among $\textbf{v}^{(p)}_1, \textbf{v}^{(p)}_2, \ldots, \textbf{v}^{(p)}_k$ that are colored $+1$ according to $f$, contradicting the fact that $f$ respects the PCSP $(P,Q)$. 
\end{proof}

\subsection{Absence of sphere coloring via connectivity of configurations}

Finally, we show the absence of sphere coloring for $\Gamma_5$ using a connectivity lemma. 
\begin{lemma}
\label{lem:gamma5}
\label{thm:sphere-coloring-gamma}
Fix integers $k \geq 3, b \in \{0,1,\ldots,k\}\setminus\{1,k\}$. 
There exists an integer $n_0$ such that for every $n \geq n_0$, there does not exist coloring $f:\mathbb{S}^n \rightarrow \{0,1\}$ that is folded and respects the {promise template} $\Gamma_5 =  (P,Q)$, $P=\Ham_k \{1,k\}$, $Q=\Ham_k \{0,1,\ldots,k\}\setminus\{b\}$.
\end{lemma}
We dedicate the rest of the section to proving~\Cref{lem:gamma5}. 
We pick the configuration of vectors along the same lines as in~\Cref{lem:gamma2}. Fix $\textbf{v}_0 \in \mathbb{S}^n$. The $P$-configuration that we study is a set of vectors $\textbf{v}_1, \textbf{v}_2, \ldots, \textbf{v}_k$ that is obtained by first sampling $t \in \{1,k\}$ such that 
\[
t = \begin{cases}
1\text{ with probability}\frac{k}{2k-2} \\ 
k\text{ with probability}\frac{k-2}{2k-2}
\end{cases}
\]
Then, we sample a uniform element from $\Ham_k \{t\}$. We get the following properties: 
\begin{enumerate}
    \item (First moments). ${ \langle \textbf{v}_i , \textbf{v}_0 \rangle }= \left( \frac{k}{2k-2} \right) \frac{2-k}{k} + \left( \frac{k-2}{2k-2} \right) 1 =0$ for every $i \in [k]$. 
    \item (Second moments). For every $i \neq j \in [k]$, we get 
    \begin{align*}
        { \langle \textbf{v}_i , \textbf{v}_j \rangle }&= \left( \frac{k}{2k-2} \right) \frac{\binom{k-1}{2}-(k-1)}{\binom{k}{2}} + \left( \frac{k-2}{2k-2} \right)1 = \frac{k-3}{k-1}
    \end{align*}
\end{enumerate}

{
For ease of notation, let \(\alpha = \frac{k-3}{k-1}\). We restrict our attention to vectors in the unit sphere \(\mathbb{S}^n\) that are orthogonal to \(\mathbf{v}_0\). This allows us to focus on \(P\)-configurations which are sets of \(k\) unit vectors with pairwise inner products equal to \(\alpha\). More formally, we define such a set of vectors as an $\alpha$-configuration.

\begin{definition}
A set of vectors \(V = \{\mathbf{v}_1, \mathbf{v}_2, \dots, \mathbf{v}_k\} \subset \mathbb{S}^n\) is called an \(\alpha\)-configuration if the vectors \(\mathbf{v}_i\) satisfy:
\[
\langle \mathbf{v}_i, \mathbf{v}_j \rangle = \alpha \quad \text{for all } i \neq j.
\]
\end{definition}

Given a folded sphere coloring \(f\), our objective is to prove that there exists an \(\alpha\)-configuration \(V\) such that exactly \(b\) of the vectors in \(V\) are assigned the value \(+1\).}

Unlike the earlier studied {promise template}, here, the setting when $b=0$ is relatively straightforward, simply because $\alpha \geq 0$. $\alpha \geq 0$ implies that there are an arbitrarily large number of unit vectors (as we can pick $n$ to be large enough) all of whose pairwise inner product is equal to $\alpha$. Specifically, we pick a set of $2k-1$ unit vectors all of whose pairwise inner product is equal to $\alpha$. Among those, $k$ of them are colored the same according to $f$. By taking the negation of these if needed, we can infer that there are $\alpha$-configurations that are all colored $+1$, and also $\alpha$-configurations that are all colored $-1$. 

Before delving further, we handle the case when $\alpha =0$, i.e., when $k=3$. In this case, we just pick a set of $k$ unit vectors that are all orthogonal to each other and their negations. Note that these are $2k$ pairwise orthogonal vectors where exactly $k$ of them are colored $+1$ according to $f$. Thus, we can pick $k$ pairwise orthogonal vectors from this set where exactly $b$ of them are colored $+1$ according to $f$. Henceforth, we assume that $\alpha >0$. 

To show that there are $\alpha$-configurations that have exactly $b$ vectors that are colored $+1$, we show a connectivity lemma (\Cref{lem:connectivity}) where we prove that between any two $\alpha$-configurations, there exists a path using $O_{k,\alpha}(1)$ $\alpha$-configurations where we change a single vector at each step in the path. As there is an $\alpha$-configuration where all are $k$ vectors are colored $+1$, and the $\alpha$-configuration obtained by negating these vectors where all the vectors are colored $-1$, the connectivity lemma then shows that for every $b \in \{0,1,\ldots, k\}$, there exists an $\alpha$-configuration that has exactly $b$ vectors that are colored $+1$. 
 
We first prove a simplified version of the connectivity lemma that we use to prove~\Cref{lem:connectivity}. 

\begin{lemma}
\label{lem:vectors-connected}
Given an $\alpha$-configuration $U=\{\textbf{u}_1,\textbf{u}_2, \ldots, \textbf{u}_k\}\subseteq \mathbb{S}^n$, and a unit vector $\textbf{w} \in \mathbb{S}^n$ that is orthogonal to each vector in $U$, there exists $L := L(k, \alpha)$ and a set of $\alpha$-configurations $V_1, V_2, \ldots, V_L$ such that 
\begin{enumerate}
    \item The consecutive configurations differ in a single vector i.e.,$ |V_i \cap V_{i+1}|=k-1$ for every $i \in [L-1]$. 
    \item Final configuration contains $\textbf{w}$, i.e., $\textbf{w} \in V_L$, and the initial configuration $V_1$ is equal to $U$. 
\end{enumerate}
\end{lemma}
\begin{proof}
We prove the lemma by studying the inner product of $\textbf{w}$ with an $\alpha$-configuration $V$, which is equal to all zeroes initially when $V=U$, and changing $V$ one vector at a time such that the inner product of $V$ with $\textbf{w}$ eventually reaches all $\alpha$s. Towards this end, for an $\alpha$-configuration $V$, we define { $I(V,\textbf{w})$ as the Gram matrix of $\{\textbf{v}_1, \textbf{v}_2, \ldots, \textbf{v}_k, \textbf{w}\}$.}
\[
I(V,\textbf{w})_{i,j} =\begin{cases}
 \langle \textbf{v}_i, \textbf{v}_j \rangle \text{ if }1 \leq i,j \leq k. \\ 
 \langle \textbf{v}_i, \textbf{w} \rangle \text{ if }i = k+1, 1 \leq j \leq k. \\ 
 \langle \textbf{w}, \textbf{v}_j \rangle \text{ if } j = k+1, 1 \leq i \leq k. \\ 
 \langle \textbf{w}, \textbf{w} \rangle = 1, \text{ if }i=j=k+1.
\end{cases}
\]

Starting with $I(V,\textbf{w})$ where $V=U$, our goal is to change one vector in $V$ at a time so that we eventually reach a configuration where the last column in $I(V,\textbf{w})$ is equal to $(\alpha, \alpha, \ldots, \alpha, 1)$. Note that changing one vector in $V$ corresponds to changing a single value in the last column (and the corresponding value in the last row) in $I(V,\textbf{w})$. We show that the opposite direction also holds, i.e., by changing a single value in the last column (and the corresponding value in the last row) of $I(V,\textbf{w})$, we obtain a new matrix that is $I(V',\textbf{w})$ with $V'$ being different from $V$ only in a single vector, as long as the new matrix is positive semidefinite.

\begin{claim}
\label{cla:psd-change}
Suppose that $A$ is a $m \times m$ real symmetric positive semidefinite matrix with $A=U^TU$ with $U=[\textbf{u}_1 \textbf{u}_2 \ldots \textbf{u}_m]$ where $\textbf{u}_i \in \mathbb{R}^n$ with $n\geq m$, and $A'$ is another real symmetric positive semidefinite matrix such that $A'$ and $A$ differ only in {the entries $(1,2)$ and $(2,1)$}. Then, there exists {$\bu'_1 \in \R^n$ such that for} $U' {:}=[\textbf{u}'_1 \textbf{u}_2 \ldots \textbf{u}_m]$ {we have} that $A'=(U')^TU'$.
\end{claim}
\begin{proof}
As $A'$ is a positive semidefinite matrix, there exist $\textbf{v}_1, \textbf{v}_2, \ldots, \textbf{v}_m \in \mathbb{R}^n$ such that $A'=V^TV$ where $V=[\textbf{v}_1 \textbf{v}_2 \ldots \textbf{v}_m]$. {Observe that $\langle \bu_i, \bu_j\rangle = \langle \bv_i, \bv_j\rangle$ for all $i,j \in \{2, \hdots, n\}$. Therefore, by Proposition~\ref{prop:gram-ortho}, there exists an orthogonal matrix $Q \in \R^{n \times n}$ such that $\bv_i = Q\bu_i$ for all $i \in \{2, \hdots, n\}$. Thus, $A' = (V Q^{T})^T (Q^{T} V)$, where the last $n-1$ columns of $Q^T V$ are $\bu_2, \hdots, \bu_n$. We set $\bu'_1$ to be the first column of $Q^T V$.}
\end{proof}
Thus, our goal is to obtain a series of $(k+1) \times (k+1)$ real symmetric positive semidefinite matrices $M_1, M_2, \ldots, M_L$ such that 
\begin{enumerate}
    \item $M_1 = I(U, \textbf{w})$. 
    \item The diagonal entries of $M_i$ for every $i \in [L]$ are all equal to $1$. 
    \item All the off-diagonal entries of $M_L$ are equal to $\alpha$.
    \item For every $i \in [L-1]$, $M_i$ and $M_{i+1}$ differ only in one element in the last column (and the corresponding element in the last row). 
\end{enumerate}
Towards this end, for $\epsilon \geq 0$, $0 \leq \gamma \leq \alpha$, and $d \in [k]$, we define the $(k+1)\times (k+1)$ matrix $M(\gamma, \epsilon, d)$ as follows: 
\[
M(\gamma, \epsilon, d)_{i,j} = \begin{cases} 
1, \text { if }i=j.\\
\alpha, \text { if }1\leq i,j\leq k.\\
\gamma + \epsilon, \text{ if }i=k+1, 1\leq j \leq d \text{ or } j=k+1, 1 \leq i\leq d.\\
\gamma, \text{ if }i=k+1, d+1\leq j \leq k \text{ or } j=k+1, d+1 \leq i\leq k.
\end{cases}
\]

Note that $I_{U,\textbf{w}}=M(0,0,k)$, and our goal $M_L$ is equal to $M(\alpha, 0,k)$. We consider the below sequence of positive semidefinite matrices $M(0,0,k), M(0,\epsilon,1), M(0,\epsilon,2), \ldots, M(0,\epsilon,k), M(\epsilon, \epsilon,1), M(\epsilon, \epsilon, 2),\ldots,$ $M(\epsilon, \epsilon,k)$, $M(2\epsilon,\epsilon,1), \ldots, M(\alpha-\epsilon, \epsilon, k)$. At each step, we change a single element in the last column (and the corresponding element in the last row). 

The final step is to show that when we set $\epsilon\leq \frac{1-\alpha}{k}$, $M(\gamma, \epsilon,d)$ is positive semidefinite for every $d\in [k], 0\leq \gamma\leq \alpha$. This follows from a simple calculation. 

\begin{align*}
    x^TM(\gamma, \epsilon, d)x &= \gamma ( \sum_{i=1}^{k+1}x_i)^2 + (\alpha - \gamma) ( \sum_{i=1}^{k}x_i)^2 + (1-\alpha) \sum_{i=1}^kx_i^2 + (1-\gamma)x_{k+1}^2 + { 2\epsilon} (x_{k+1})(\sum_{i=1}^d x_i) \\ 
    &\geq (1-\alpha) \sum_{i=1}^{k+1} x_i^2 + {2\epsilon } (x_{k+1})(\sum_{i=1}^d x_i) \\
    &\geq {(1-\alpha) \sum_{i=1}^{k+1} x_i^2 -2\epsilon \left|(x_{k+1})(\sum_{i=1}^d x_i)\right|}\\
    &\geq {(1-\alpha) \sum_{i=1}^{k+1} x_i^2 -2\left(\frac{1-\alpha}{k}\right) \left|(x_{k+1})(\sum_{i=1}^d x_i)\right|}\\
    &\geq \left( \frac{1-\alpha}{k} \right) \left(k \sum_{i=1}^{k+1} x_i^2 - {2}\left|(x_{k+1})(\sum_{i=1}^d x_i)\right| \right) \\ 
    &\geq {\left( \frac{1-\alpha}{k} \right) \left(k \sum_{i=1}^{k+1} x_i^2 - \sum_{i=1}^d (x_i^2 + x_{k+1}^2) \right)}\\ 
    &\geq 0. \qedhere
\end{align*}
\end{proof}

Now, we prove the connectivity lemma. 
\begin{lemma}
\label{lem:connectivity}
Fix an integer $k\geq 2$ and $0 < \alpha < 1$. Suppose that $U$ and $V$ are two $\alpha$-configurations in $\mathbb{S}^n$ { with cardinality $k$ each}. Then, there exists $n_0 := n_0(k,\alpha)$, and $L := L(k, \alpha)$ such that as long as $n \geq n_0$, there exist $\alpha$-configurations $V_1, V_2, \ldots, V_L$ {each with cardinality $k$ }such that 
\begin{enumerate}
    \item The endpoints are $U$ and $V$ i.e., $U=V_1, V=V_L$. 
    \item Any two consecutive configurations differ in exactly one vector i.e., $| V_i \cap V_{i+1}|=k-1$ for every $i \in [L-1]$. 
\end{enumerate}

\end{lemma}

\begin{proof}
We use induction on $k$. First, we consider the case when $k=2$. Let $U=\{\textbf{u}_1, \textbf{u}_2\}$, and $V=\{\textbf{v}_1, \textbf{v}_2\}$ be two $\alpha$-configurations. Consider an arbitrary vector $\textbf{w}$ that is orthogonal to all the vectors in $U$ and $V$. Such a $\textbf{w}$ is guaranteed to exist when $n$ is large enough. 
Now, using~\Cref{lem:vectors-connected}, we can infer that there exists a configuration $W=\{\textbf{w}, \textbf{w}'\}$ such that there is a path of length $O_{\alpha}(1)$ from $U$ to $W$, and from $V$ to $W$. Thus, there exists a path of length $O_{\alpha}(1)$ from $U$ to $V$. 

Assume that the proof holds for $k-1$, and we are given the configurations $U$ and $V$ consisting of $k$ vectors each. We choose a vector $\textbf{w}$ that is orthogonal to each of the vectors in $U$ and $V$. Using~\Cref{lem:vectors-connected}, there are {$\alpha$-}configurations $X=\{\mathbf{w}, \mathbf{x}_1, \mathbf{x}_2, \ldots, \mathbf{x}_{k-1}\}$ and $Y=\{\mathbf{w}, \mathbf{y}_1, \mathbf{y}_2, \ldots, \mathbf{y}_{k-1}\}$ such that there is a path of length $O_{\alpha, k}(1)$ from $U$ to $X$ and from $V$ to $Y$.
Now, our goal is to show that there is a path from $X$ to $Y$ of length $O_{\alpha, k}(1)$. We achieve this by restricting ourselves to components orthogonal to $\mathbf{w}$ of the $(k-1)$-sized configurations $X'=\{\mathbf{x'}_1, \mathbf{x'}_2, \ldots, \mathbf{x'}_{k-1}\}$ and $Y'=\{\mathbf{y'}_1, \mathbf{y'}_2, \ldots, \mathbf{y'}_{k-1}\}$, where {$\mathbf{x'}_i = \frac
{\mathbf{x}_i - \langle \mathbf{x}_i, \mathbf{w}\rangle \mathbf{w}}{\norm{\mathbf{x}_i - \langle \mathbf{x}_i, \mathbf{w}\rangle \mathbf{w}}}$} for each $i \in [k-1]$ (and similarly for $\mathbf{y'}_i$). 
{
Note that \( X' \) and \( Y' \) are \(\frac{\alpha}{1+\alpha}\)-configurations, 
each of which is a subset of 
\[
\{\mathbf{v} \in \mathbb{S}^n : \langle \mathbf{v}, \mathbf{w} \rangle = 0 \} \cong \mathbb{S}^{n-1}.
\]
By the induction hypothesis, there exists a path from \( X' \) to \( Y' \) 
using \(\frac{\alpha}{1+\alpha}\)-configurations in 
$\{\mathbf{v} \in \mathbb{S}^n : \langle \mathbf{v}, \mathbf{w} \rangle = 0 \}$.
}
Adding the component along $\mathbf{w}$, we get a path from $X$ to $Y$ of length $O_{\alpha,k}(1)$, finishing the proof. 
\end{proof}

We are now ready to prove~\Cref{thm:sphere-coloring-gamma}. 
\begin{proof}
Suppose for contradiction that there exists a coloring $f:\mathbb{S}^n \rightarrow \{0,1\}$ that is folded and respects the PCSP $\Gamma_5$. Consider an arbitrary set of vectors $\{ \mathbf{v}_1, \mathbf{v}_2, \ldots, \mathbf{v}_{2k-1} \}$ that are all orthogonal to $\textbf{v}_0$, and have pairwise inner product $\alpha$. Such a set is guaranteed to exist as $\alpha \geq 0$. 
There exists a set of $k$ vectors among these that are all assigned the same color in $f$. Let these form the configuration $U$, and the set of negations of these vectors be the configuration $V$. Using~\Cref{lem:connectivity}, there exists a path from $U$ to $V$ where we change a single vector in each step. Note that the endpoints of the path have $0$ and $k$ vectors assigned $+1$ respectively. Since we change at most one vector at a time, there exists a configuration where we have exactly $b$ $1$s, a contradiction. 
\end{proof}

\Cref{lem:gamma1,lem:gamma2,lem:gamma3,lem:gamma4,lem:gamma5} for templates ${\Gamma_{1-5}}$, together with~\Cref{lem:reduction},~\Cref{thm:raghavendra} and~\Cref{lem:sphere-coloring} finish the proof of our main hardness result~\Cref{thm:main-hardness}. 

\medskip \noindent \textbf{Explicit Construction.} We give an explicit construction of an integrality gap instance for {$\fiPCSP(\Gamma_5)$ where} $\Gamma_5=(P, Q), P=\Ham_k \{1,k\}, Q=\Ham_k \{0,1,\ldots,k\}\setminus\{b\}$ for arbitrary $b \in \{0,1,\ldots,k\}$. 
Let $L$ be a large constant (depends on $k$, to be set later). We have $n=2k-1+\binom{2k-1}{k}L$ variables $x_i: i \in [2k-1], x_S^{(i)}: i \in [L], S \subseteq [2k-1], |S|=k$. 
{For every subset $S \subseteq [2k-1]$ with $|S|=k$ and $S=\{i_1, i_2, \ldots, i_k\}$,  we add the following $L+k+1$ constraints to the instance:} 
\begin{align*}
   \{x_{i_1}, x_{i_2}, \ldots, x_{i_k}\}, \{x_{i_2}, x_{i_3}, \ldots, x_{i_{k}}, x_S^{(1)}\},  
   \{x_{i_3}, \ldots, x_S^{(1)}, x_S^{(2)}\}, \ldots, \{x_S^{(L-k+1)}, \ldots, x_S^{(L)}\}, \\ 
   \{x_S^{(L-k+2)}, \ldots, x_S^{(L)}, \overline{x_{i_1}}\}, \{x_S^{(L-k+3)}, \ldots, \overline{x_{i_1}}, \overline{x_{i_2}}\}, \ldots, \{\overline{x_{i_1}}, \overline{x_{i_2}}, \ldots, \overline{x_{i_k}}\}.
\end{align*}
We choose $L$ to be the constant factor from~\Cref{lem:connectivity} with $\alpha =\frac{k-3}{k-1}$. 
 The idea is that when all the {literals} in the constraint $\{x_{i_1}, x_{i_2}, \ldots, x_{i_k}\}$ are {set} to be True, all the {literals} in $\{\overline{x_{i_1}}, \overline{x_{i_2}}, \ldots, \overline{x_{i_k}}\}$ {are set to be False.} {Since} there {is} a series of constraints between them where we alter a single variable, there must exist a constraint where there are exactly $b$ variables that are set to True. 

Formally, we show that {this instance lacks a weakly satisfying assignment for} $Q=\Ham_k \{0,1,\ldots, k\} \setminus \{b\}$. 
Suppose for contradiction that there is an assignment that {weakly} satisfies all the constraints. Since there are $2k-1$ variables $x_1, x_2, \ldots, x_{2k-1}$, at least $k$ of them are set to be true. If not, then at least $k$ of the negated variables are set to be true. This implies that there is a sequence of constraints where the endpoints are assigned all True and all False, and at every point, we change a single variable. This implies that there is a constraint where there are exactly $b$ variables that are set to True, a contradiction. 

We now show that the instance has a basic SDP solution with zero error. Set $\textbf{v}_1, \textbf{v}_2, \ldots, \textbf{v}_{2k-1} \in \mathbb{S}^n$ to the variables $x_1, x_2, \ldots, x_{2k-1}$ such that ${ \langle \textbf{v}_i , \textbf{v}_0 \rangle }= 0$, and ${ \langle \textbf{v}_i , \textbf{v}_j \rangle}= \alpha$ for every $i \neq j$.  {
We can prove the existence of such a set of vectors by showing that the Gram matrix of the vectors \( \{\mathbf{v}_1, \mathbf{v}_2, \ldots, \mathbf{v}_{2k-1}\} \) is positive semidefinite. 

The Gram matrix \( G \) has \( 1 \) on the diagonal entries and \( \alpha \geq 0 \) on the off-diagonal entries:
\[
G =
\begin{bmatrix}
1 & \alpha & \alpha & \cdots & \alpha \\
\alpha & 1 & \alpha & \cdots & \alpha \\
\alpha & \alpha & 1 & \cdots & \alpha \\
\vdots & \vdots & \vdots & \ddots & \vdots \\
\alpha & \alpha & \alpha & \cdots & 1
\end{bmatrix}.
\]

We can write \( G \) as:
\[
G = \alpha J_{2k-1} + (1-\alpha) I_{2k-1},
\]
where \( I_{2k-1} \) is the \( (2k-1) \times (2k-1) \) identity matrix, and \( J_{2k-1} \) is the \( (2k-1) \times (2k-1) \) matrix with all entries equal to \( 1 \).

Since both \( I_{2k-1} \) and \( J_{2k-1} \) are positive semidefinite, and \( 0 \leq \alpha \leq 1 \), the matrix \( G \) is positive semidefinite as well.}
Finally, we use~\Cref{lem:connectivity} to set the vectors $\textbf{v}_S^{(i)}$ for every $S \subseteq [2k-1], |S|=k, i \in [L]$. 

\section{The SDP minion}
\label{sec:minion}

As previously mentioned, polymorphisms are a powerful tool for understanding the computational complexity of PCSPs. However, beyond some of the simplest classes of PCSPs, individually classifying the complexity based on specific polymorphisms can be unwieldy. Instead, one often looks to higher-level structure between classes of polymorphisms, which is captured by the notion of \emph{minions} (also called clonoids) \cite{BBKO21}.

In this section, we describe the structure of the ``basic SDP minion,'' which gives a precise algebraic condition for when the basic SDP decides a Promise CSP. We note that a similar minion was concurrently and independently discovered by Ciardo-Zivny \cite{cz22-minion}. 

{We prove our minion characterization for ``ordinary'' promise constraint satisfaction problems over any domain. In particular, we do not consider instances with folding or the setting of constants.}

\subsection{Minion preliminaries}

\newcommand{\Set}{\mathsf{Set}}
\newcommand{\NeFinSet}{\mathsf{NeFinSet}}

We now give the definition of a minion. {This is most easily stated in the language of category theory, however we shall also give ``concrete'' interpretations of any category-theoretic statements we make. We let $\Set$ denote the category of sets, and $\NeFinSet$ denote the category of \emph{non-empty} finite sets.}

\begin{definition}
A minion $\mathcal M$ is {a functor from $\NeFinSet$ to $\Set$. More concretely, for every non-empty finite set $A$, we let $\cM^A$ denote the image of $A$ with respect to the functor. For every pair of non-empty finite sets $A$ and $B$ and function $\pi : A \to B$, we let the \emph{minor map} $\cM^{\pi} : \cM^{A} \to \cM^{B}$ denote the image of $\pi$ with respect to the functor. Further, for any maps $\pi : A \to B$ and $\tau : B \to C$, we have that $\cM^{\tau} \circ \cM^{\pi} = \cM^{\tau \circ \pi}$. We also assert for the identity map $\id_A : A \to A$, we have that $\cM^{\id_A} : \cM^A \to \cM^A$ is also the identity map.}
\end{definition}

To streamline the minor map notation $\cM^{\pi}$, for all $M \in \cM^A$, we let $\pi \oslash M := \cM^{\pi}(M) \in \cM^B$.

{
\begin{remark}
Some works only use \emph{minion} (or \emph{clonoid}) to refer to sets of the form $\Pol(\Gamma)$ for some (possibly) promise template $\Gamma$, and refer to all other minions as an \emph{abstract minions}. In this work, all such objects are minions.

In addition, $\cM$ is typically only defined as a functor from the subcategory of $\NeFinSet$ of sets of the form $[k]$ for some $k \in \N$. Our definition is equivalent as every non-empty finite set $A$ has a bijection with $[|A|]$. However, we consider this richer definition of a minion so that we can discuss how the minion acts on sets without an explicit enumeration. 
\end{remark}
}

\subsubsection{{Examples}}

{We now walk through a few examples.}
\paragraph{Polymorphisms.} {Given a promise template $\Gamma = \{(P_1, Q_1), \hdots, (P_l, Q_l)\}$ over the domain $(D_1, D_2)$, one can view $\Pol(\Gamma)$ as a minion. Formally, for every non-empty finite set $A$, we let $\Pol(\Gamma)^A$ denote the set of functions $f : D_1^A \to D_2$ (where $D_1^A$ is the set of functions from $A$ to $D_1$) such that for any $i \in [l]$ and $\alpha : A \to P_i$, we have that 
\begin{align}
    (f(\pi_1 \circ \alpha), \hdots, f(\pi_{k_i} \circ \alpha)) \in Q_i,\label{eq:pol-cond}
\end{align}
where $\pi_j : P_i \to D_1$ is the projection of $P_i \subseteq D^{k_i}$ onto the $j$th coordinate. Observe that if $A = \{1,\hdots, L\}$ then $\Pol(\Gamma)^A$ can be identified with the arity $L$ polymorphisms.

Likewise, for any map $\tau : A \to B$ for non-empty finite sets $A$ and $B$, we define $\Pol(\Gamma)^\tau$ such that for any $f \in \Pol(\Gamma)^A$  and $\beta : B \to D_1$, we have that
\begin{align}
    (\tau \oslash f)(\beta) = f(\beta \circ \tau).\label{eq:pol-min}
\end{align}
We It's clear that $\id_{A} \oslash f = f$. Furthermore, for any pair of maps $\tau : A \to B$, $\tau' : B \to C$, and $\gamma : C \to D_1$, we have that
\[
  (\tau' \oslash (\tau \oslash f))(\gamma) = (\tau \oslash f)(\gamma \circ \tau') = f((\gamma \circ \tau') \circ \tau)) = f(\gamma \circ (\tau' \circ \tau)) = ((\tau' \circ \tau) \oslash f)(\gamma),  
\]
so associativity also holds.
}

{
\paragraph{Trivial minion.} The simplest example of a minion $\cM_{\operatorname{triv}}$ maps every non-empty finite set to the same set $\{e\}$. All functions $\tau : A \to B$ map to the same trivial map from $e$ to istelf.
}

{
\paragraph{Identity minion.} {The \emph{identity minion} $\cM_{\id}$ is the identity functor from $\NeFinSet$ to itself (when viewed as a subcategory of $\Set$). In other words, for every non-empty finite set $A$ we have that $\cM_{\id}^A = A$, and for every map $\tau : A \to B$, we have that $\cM_{\id}^\tau = \tau$.}
}

{
\paragraph{Convex combinations.} For every non-empty finite set $A$, let $\Delta_A$ be the set of probability distributions over $A$. That is, $\Delta_A = \{f : A \to [0,1] : \sum_{a\in A} f(a) = 1\}.$ We let $\cQ_{\conv}$ denote the minion which maps $A$ to $\Delta_A$~\cite{BBKO21}. Given a map $\tau : A \to B$ and a probability distribution $f \in \Delta_A$, we let $\tau \oslash f \in \Delta_B$ be defined by 
\[
    (\tau \oslash f)(b) = \sum_{a \in \tau^{-1}(b)} f(a),
\]
where the empty sum is equal to $0$. The main motivation for studying $\cQ_{\conv}$ is that it is closely related to the use of basic LP \cite{BBKO21}.
}

\subsubsection{{Minion homomorphisms}}

{
A key approach of \cite{BBKO21} is that to better understand $\PCSP(\Gamma)$, we can ask how $\Pol(\Gamma)$ relates to other (abstract) minions. The most natural way to do that is via a \emph{minion homomorphism}, which in category-theoretic language is a \emph{natural transformation} between two minions.}

\begin{definition}
A \emph{minion homomorphism} $\psi: \mathcal M \to \mathcal N$ between two minions {is a natural transformation between the respective functors. More concretely, for every non-empty finite set $A$, we have a map $\psi_A : \cM^A \to \cN^A$ such that for every map $\tau : A \to B$ and $M \in \cM^A$, we have that}
\[
    {\tau \oslash \psi_A(M) = \psi_B(\tau \oslash M).}
\]
In other words, the following diagram commutes:
\begin{center}
\begin{tikzcd}
    \cM^A \arrow[r,"\psi_A"] \arrow[d,"\cM^{\tau}"] &\cN^A \arrow[d,"\cN^{\tau}"]\\
    \cM^B \arrow[r,"\psi_B"] &\cN^B
\end{tikzcd}
\end{center}
\end{definition}

{To assert the existence of a minion homomorphism from $\mathcal M$ and $\cN$, we write $\cM \to \cN$. A key result of \cite{BBKO21} is that minion homomorphisms govern the computational complexity of PCSPs.}

{\begin{theorem}[\cite{BBKO21}]
Let $\Gamma_1, \Gamma_2$ be promise templates. If $\Pol(\Gamma_1) \to \Pol(\Gamma_2)$ then decision version of $\PCSP(\Gamma_2)$ is polynomial-time reducible to $\PCSP(\Gamma_1)$.
\end{theorem}}

{Another result of \cite{BBKO21} is that for any promise template $\Gamma$, $\cQ_{\conv} \to \Pol(\Gamma)$ if and only if the basic LP solves $\PCSP(\Gamma)$. Our goal in the remainder of this section is to prove an analogous result for the basic SDP.} 

\subsubsection{The free structure}

{Let $\Gamma$ be a promise template.} One of the most important tools for understanding minion {homomorphisms of} the form $\mathcal M \to \Pol({\Gamma})$ {is} the \emph{free structure} \cite{BBKO21}. {Since in this paper we parameterize our promise constraint satisfaction problems using promise templates rather than a pair of relational structures, we restate the definition of a free structure in our language, which we call a \emph{free template}.}  {Given a relation $R \subseteq D^k$, we define the \emph{free relation} $\F_{\cM}(R)$ to be a predicate of arity $k$ over $\cM^D$ with the following property, we have that $(M_1, \hdots, M_k) \in (\cM^D)^k$ is an element of $R$ if and only if there exists $N \in \cM^{R}$ such that for all $i \in [k]$, we have that 
\[
    \pi_i \oslash N = M_i,
\]
where $\pi_i : R \to D$ is the projection of $R$ onto the $i$th coordinate.

Given a promise template $\Gamma = \{(P_1, Q_1), \hdots, (P_l, Q_l)\}$, we define the \emph{free template} 
\[
\F_{\cM}(\Gamma) := \{(\F_{\cM}(P_1), Q_1), \hdots, (\F_{\cM}(P_l), Q_l)\}.
\]
Recall that for a promise template $\Gamma$ over domain $(D_1, D_2)$, there exists a map $h : D_1 \to D_2$ such for every $i \in [\ell]$ and $\bx \in P$, we have that $h(\bx) \in Q$. It is not obvious (nor necessarily true) that for $\F_{\cM}(\Gamma)$ an analogous map $h^{\cM} : \cM^{D_1} \to D_2$ necessarily exists. In fact, constructing $h^{\cM}$ asserts the existence of a minion homomorphism. 
}
{This is t}he fundamental property of the free structure{.}

{
\begin{lemma}[Lemma~4.4 of \cite{BBKO21}]\label{lem:free-structure}
$\mathcal M \to \Pol(\Gamma)$ if and only if $\F_{\cM}(\Gamma)$ is a promise template.
\end{lemma}
}

{In the language of category theory, the free structure can be thought of as an adjoint functor. See the discussion preceding Remark 4.5 in \cite{BBKO21}.}

\subsection{SDP {m}inion {d}efinition}

{
We now define the minion corresponding to the behavior of the basic SDP. We call this minion $\cM_{\SDP}.$ Recall that for $\cQ_{\conv}$, for any non-empty finite set $A$, we have that $\cQ_{\conv}^A = \Delta_A$, where $\Delta_A$ is the set of probability distributions over $A$. We define $\cM_{\SDP}^A$ in an analogous manner which captures orthogonal configurations of vectors index by $A$.}

{S}ince the SDP minion is a universal object, we need to be able to represent vectors of arbitrary large dimensions. We achieve this using infinite-dimensional vectors that are eventually zero. Similar techniques have been used in other minion constructions \cite{ciardo2022clap}.

Let $\R^{\omega}$ be infinite sequences of real numbers {indexed by $\N$} which are eventually $0${.} {We note} that $\R^{\omega}$ is an inner product space {with the bilinear form}
\[
    {\langle \bx, \by \rangle := \sum_{i=1}^{\infty} x_i y_i,}
\]
{which is well-defined as only finitely many terms in the summation are nonzero.}

{We now formally define $\cM_{\SDP}$. For each non-empty finite set $A$, we define $\cM_{\SDP}^{A}$ to be the set of functions $\bw : A \to \R^{\omega}$ with the following properties.}

\begin{enumerate}
\item For all {distinct $a, a' \in A$, $\langle \bw(a), \bw(a')\rangle = 0$.}
\item $\sum_{{a \in A}} \|{\bw(a)}\|^2_2 = 1$.
\end{enumerate}
Observe that the {first condition implies that the} second condition is equivalent to $\|\sum_{{a \in A}} {\bw(a)}\|_2 = 1$.

The minors of $\cM_{\SDP}$ are {defined analogously to $\cQ_{\conv}$. Let $A$ and $B$ be non-empty finite sets and consider $\bw \in \cM_{\SDP}^A$ and $\tau : A \to B$. We define $\tau \oslash \bw : B \to \R^{\omega}$ to be for all $b \in B$,}
\[
    {(\tau \oslash \bw)(b) := \sum_{a \in \tau^{-1}(b)} \bw(a),}
\]
{where the empty sum equals $\mathbf{0} \in \R^{\omega}$. We now finish the remaining details which prove that $\cM_{\SDP}$ is a minion.}

{\begin{lemma}
$\cM_{\SDP}$ is a minion.
\end{lemma}}
\begin{proof}
First, for each {$\bw \in \cM_{\SDP}^{A}$} and {$\tau : A \to B$} we verify that ${\bw' := \tau \oslash \bw \in \cM_{\SDP}^{B}}$. First, fix {distinct $b, b' \in B$}. We have that
\begin{equation*}
    { \langle \bw'(b) , \bw'(b') \rangle }= { \left\langle \Bigl(\sum_{a \in \tau^{-1}(b)} \bw(a)\Bigr), \Bigl(\sum_{a' \in \tau^{-1}(b')} \bw(a')\Bigr)\right\rangle}
    = \sum_{{\substack{a \in \tau^{-1}(b)\\a' \in \tau^{-1}(b')}}} { \langle \bw(a) , \bw(a')\rangle}
    = 0{,}
\end{equation*}
{since $b$ and $b'$ are distinct and $\tau$ is a function.} Further,
\begin{align*}
{ \left\langle \Bigl( \sum_{b \in B} \bw'(b)\Bigr), \Bigl(\sum_{b \in B} \bw'(b)\Bigr) \right\rangle} &= { \left\langle \Bigl(\sum_{a\in A} \bw(a)\Bigr), \Bigl(\sum_{a\in A} \bw(a)\Bigr) \right\rangle}= 1.
\end{align*}
Thus, ${\bw' \in \cM_{\SDP}^{B}}.$ It is trivial to observe that {$\cM^{\id_A}$ is the identity map.}

The only remaining condition to check is that the minors commute. Consider {$\pi : A \to B$} and {$\tau : B \to C$. Consider $\bu \in \cM_{\SDP}^A.$}  We seek to verify that $\tau \oslash (\pi \oslash \bu) = (\tau \circ \pi) \oslash \bu$. For all {$c \in C$}, we have that
{
\begin{align*}
(\tau \oslash (\pi \oslash \bu))(c) = \sum_{b \in \tau^{-1}(c)} (\pi \oslash \bu)(b) = \sum_{b \in \tau^{-1}(c)} \sum_{a \in \pi^{-1}(b)} \bu(a) = \sum_{a \in (\tau \circ \pi)^{-1}(c)} \bu(a) = ((\tau \circ \pi) \oslash \bu)(a),
\end{align*}
as desired.
}
\end{proof}

The goal of the rest of this section is to prove the following theorem.

\begin{theorem}\label{thm:basic-SDP}
{For any promise template $\Gamma$, t}he basic SDP decides $\PCSP({\Gamma})$ if and only if $\cM_{\SDP} \to \Pol({\Gamma})${.}
\end{theorem}

\subsection{On the basic SDP}

{As we are considering the decision version of the basic SDP, we assume that $\eps_j = 0$ for each clause index $j \in [m]$. For completeness, we state the simplified SDP below. Note that we omit an objective since we are concerned only with feasibility. Any undefined notation is the same as in Section~\ref{sec:prelims}. Given a clause index $j \in [m]$, we let $\cF_j$ denote the set of assignments $f : S_j \to D$ for which $f(C_j) \in P^{(j)}$.}
{
\begin{align*}
\text{(a1)} && \lambda_j(f) &\geq 0  & \forall j\in [m], ~ f \in \cF_j \\ 
\text{(a2)} && \sum_{f \in \cF_j}\lambda_j(f)&=1  & \forall j \in [m] \\
\text{(a3)} && \langle \bv_0, \bv_0\rangle &= 1\\
\text{(a4)} && { \langle \textbf{v}_{i,a},   \textbf{v}_0 \rangle } &= \sum_{\substack{f \in \cF_j\\f(x_i) = a}}\lambda_j(f)&  \forall j\in [m],~x_i \in C_j,~ a \in D \\ 
\text{(a5)} && { \langle \textbf{v}_{i,a} , \textbf{v}_{i',a'} \rangle } &= \sum_{\substack{f \in \cF_j\\f(x_i) = a, f(x_{i'}) = a'}}\lambda_j(f)  & \forall j\in [m],~x_i,x_{i'} \in C_j,~ a,a' \in D 
\end{align*}
}

{We call this formulation of the basic SDP the ``traditional'' basic SDP, in contrast to the ``alternative'' basic SDP we soon present.}

\begin{remark}
Although we specify here that the vectors can have an arbitrarily large dimension, it is known that an SDP with $n$ vector or scalar variables is feasible if and only if it is feasible in $n$ dimensions. Thus, nothing is lost from Section~\ref{sec:prelims} by making the vectors have an arbitrarily large (but finite) dimension.
\end{remark}

{\begin{remark}
Some formulations of the basic SDP also impose that for every pair of vectors appearing in the basic SDP has a nonnegative pairwise dot product. Raghavendra's theorem does not use this in their basic SDP formulation, so we do not include it in our formulation.
\end{remark}}

\subsubsection{An alternative basic SDP}

{In order to prove Theorem~\ref{thm:basic-SDP}, w}e now present a modified basic SDP which is more convenient to work with. {Instead of having an LP variable $\lambda_j(f)$ for each potential clause assignment $f : S_j \to D$, we have a vector $\bw_{j,f} \in \R^{\omega}.$ For each $i \in [n]$, let $\bv_i : D \to \R^{\omega}$ denote the function $\bv_i(a) = \bv_{i,a}$. Likewise, for every $j \in [m]$, let $\bw_j : \cF_j \to R^{\omega}$ be the function $\bw_j(f) = \bw_{j,f}$. With this notation, we can succinctly impose our SDP conditions as follows.}

{
\begin{align*}
\text{(b1)} &&  \bv_i &\in \cM_{\SDP}^{D} & \forall i \in [n] \\ 
\text{(b2)} && \bw_j &\in \cM_{\SDP}^{\cF_j}   & \forall j \in [m] \\ 
\text{(b3)} && \bv_{i,a} &= \sum_{\substack{f \in \cF_j\\f(x_i) = a}} \bw_{j,f} & \forall j \in [m], i \in S_j, a \in D,
\end{align*}
}
where we use the fact that membership in $\cM_{\SDP}^S$ can be expressed in $O(|S|^2)$ SDP constraints.

\begin{remark}
{Note that the special ``truth'' vector $\bv_0$ is omitted from this formulation. As we shall soon see, this does not change the power of the SDP relaxation, but it does make the analysis more straightforward.}
\end{remark}

{
\begin{remark}\label{rem:b3}
Note that (b3) can also be rewritten as $\bv_i = \tau_{i,j} \oslash \bw_j$, where $\tau_{i,j} : \cF_j \to D$ is the map $\tau_{i,j}(f) = f(x_i)$. As such, $\cM_{\SDP}$ can be swapped with \emph{any} minion $\cM$ in conditions (b1--b3). However, such a relaxation is not algorithmically tractable in general.
\end{remark}
}

\subsubsection{Equivalence with traditional basic SDP}

We now prove that these two formulations of the basic SDP are equivalent. 
That is, a solution to either SDP can be transferred to the other. 

{\begin{lemma}\label{lem:SDP-equiv}
For any promise template $\Gamma$ and any instance of $\PCSP(\Gamma)$, the traditional basic SDP for the instance is feasible if and only if the alternative basic SDP for the instance is feasible.
\end{lemma}}

{Before we begin, we note that we often apply linear maps to the vectors in order to transform vectors from one SDP solution to another, extending the ideas in Section~\ref{subsec:gram-ortho}. We let $\R^{\omega \times \omega}$ denote the set of matrices such that each row and column has finite support. These finite support conditions are such that for any $M \in \R^{\omega \times \omega}$ and $v \in \R^{\omega}$ we have that $Mv \in \R^{\omega}$ and that $\R^{\omega \times \omega}$ is closed under the transpose operator.}

{The identity matrix $I_{\omega} \in \R^{\omega \times \omega}$ has $(I_{\omega})_{i,j} = 1$ if $i = j$ and $0$ otherwise. We say that $Q \in \R^{\omega \times \omega}$ is \emph{orthogonal} if $Q Q^{T} = Q^T Q = I_{\omega}.$ We now extend Proposition~\ref{prop:gram-ortho} to our setting.}

{\begin{proposition}\label{prop:ortho}
Let $n$ be a positive integer and $\bu_1, \hdots, \bu_n, \bv_1, \hdots, \bv_n \in \R^{\omega}$ be vectors. We have that $\langle \bu_i, \bu_j \rangle = \langle \bv_i, \bv_j\rangle$ for all $i,j \in [n]$ if and only if there exists an orthogonal matrix $Q \in \R^{\omega \times \omega}$ such that $\bv_i = Q \bu_i$ for all $i \in [n]$. 
\end{proposition}
\begin{proof}
Note that for any $\bu, \bu' \in \R^{\omega}$ and $Q \in \R^{\omega \times \omega}$ orthogonal, we have that $\langle \bu, \bu'\rangle = \langle Q\bu, Q\bu' \rangle$. This proves the ``if'' direction. We now focus on proving the converse.

Assume that $\langle\bu_i, \bu_j\rangle = \langle \bv_i, \bv_j\rangle$ for all $i,j \in [n]$. Let $N$ be a positive integer such that $(\bu_i)_j = 0$ and $(\bv_i)_j = 0$ for all $i \in [n]$ and $j > N$. We can thus assume that $\bu_1, \hdots, \bu_n, \bv_1, \hdots, \bv_n \in \R^N$. By Proposition~\ref{prop:gram-ortho} there exists an orthogonal matrix $Q \in \R^{N \times N}$ such that $Q\bu_i = \bv_i$ for all $i \in [n]$. We can extend this $Q$ to an orthogonal $Q' \in \R^{\omega \times \omega}$ in the following straightforward manner.
\[
    Q'_{i,j} = \begin{cases}
    Q_{i,j} & i, j \le N\\
    1 & i = j > N\\
    0 & \text{otherwise}
    \end{cases}.\qedhere
\]
\end{proof}
}

\medskip\textbf{Alternative to traditional.} First, we show that a solution to the alternative basic SDP is a solution to the traditional basic SDP. For each ${i \in [n]}$, define ${\bv_{i,D} = \sum_{a \in D} \bv_{i,a}}$. By (b1), we know that each ${\bv_{i,D}}$ is a unit vector. Further by (b3), we can deduce that ${\bv_{i,D} = \bv_{i',D}}$ whenever ${x_i}$ and ${x_{i'}}$ appear in a common constraint ${C_j}$. Thus, ${\bv_{i,D} = \bv_{i',D}}$ whenever ${x_i}$ and ${x_{i'}}$ are in the same connected component of the {hypergraph with vertices $\{x_1, \hdots, x_n\}$ and hyperedges $\{C_1, \hdots, C_m\}$}. {Let $I \subseteq [n]$ and $J \subseteq [m]$ be the indices of the variables and constraints appearing in some connected component. For any orthogonal matrix $Q \in \R^{\omega \times \omega}$,} each of the constraints (b1), (b2), and (b3) are preserved when we {multiply each vector $\bv_{i,a}$ for $i \in I$ and $\bw_{j,f}$ for $j \in J$ on the left by $Q$.} Thus, by {multiplying by suitable orthogonal matrices (i.e., apply Proposition~\ref{prop:ortho} for $n=1$)}, we can achieve a solution to (b1-3) such that {each $\bv_{i,D}$} is the same unit vector for all ${i \in [n]}$, which we shall call ${\bv_0}$. {We call this an \emph{aligned} solution to the alternative basic SDP.}

{We now transform this aligned solution into a solution to the traditional basic SDP. Each $\bv_{x,a}$ in the traditional basic SDP is actually the same as it appeared in the aligned solution to the alternative basic SDP. We set $\bv_0$ to the newly-defined $\bv_0$. Finally, we set $\lambda_j(f) = \|\bv_{j,f}\|_2^2$ for all $f \in \cF_j$.}

To see that the traditional basic SDP is satisfied, note that condition (a1) follows by inspection, and condition (a2) follows from (b2). {Conditions (a3) and (a4) follow} by substituting ${\bv_{i,D}}$ to ${\bv_0}$ and applying (b3). Condition {(a5)} follows {by} combining (b2) and (b3).

\medskip\textbf{Traditional to alternative.} We now show that any solution to the traditional basic SDP is a solution to the alternative basic SDP. Assume we have a traditional basic SDP solution. Recall that ${\bv_{i,a}} \in \R^\omega$ {for each $i \in [n]$ and $a \in D$}. We shall use the same vectors ${\bv_{i,a}}$ in {the solution to the alternative basic SDP}. Note that (b1) then follows from {(a5)} and (a2). 

Since there are finitely many vectors, there exists some $N \in \N$ such that each ${\bv_{i,a}}$ {(and $\bv_0$)} has its support within the first $N$ coordinates. {For each $j \in [m]$ and $f \in \cF_j$, p}rovisionally assign ${\hat{\bw}_{j,f}}$ to be $\sqrt{\lambda_{{j}}({f})} \cdot e_i$, where {$e_i \in \R^{\omega}$ is the $i$th standard unit vector for some $i > N$} chosen uniquely for each constraint-assignment pair. As there are finitely many constraint-assignment pairs, all of these {new vectors} are in $\R^{\omega}$.

However, as written the ${\hat{\bw}_{j,f}}$'s are not compatible with the ${\bv_{i,a}}$'s. {To rectify this, for each $j \in [m], i \in S_j,$ and $a \in D$ we define}
\[{\hat{\bv}_{j,i,a}} := \sum_{\substack{{j \in \cF}\\{f(x_i)} = a}} {\hat{\bw}_{j,f}}\]
{Using (a5), we can readily verify that for any $j \in [m]$, $i,i' \in S_j$, and $a,a' \in D$, we have that $\langle \bv_{i,a}, \bv_{i,a'}\rangle = \langle \hat{\bv}_{j,i,a}, \hat{\bv}_{j,i',a'}\rangle$. Thus, by Proposition~\ref{prop:ortho}, for each $j \in [m]$, there exists an orthogonal matrix $Q_j \in \R^{\omega \times \omega}$ such that $Q_j \hat{\bv}_{j,i,a} = \bv_{i,a}$ for all $i \in S_j$ and $a \in D$.}

{For each $j \in [m]$ and $f \in \cF_j$, define $\bw_{j,f} := Q_j \hat{\bw}_{j,f}$. Observe that for each $j \in [m]$, $i \in S_j,$ and $a \in D$, we have that
\[
    \bv_{i,a} = Q_j \hat{\bv}_{j,i,a} = \sum_{\substack{{j \in \cF}\\{f(x_i)} = a}} Q_j {\hat{\bw}_{j,f}} = \sum_{\substack{{j \in \cF}\\{f(x_i)} = a}} \bw_{j,f}, 
\]}
as desired. Thus, conditions (b2) and (b3) hold. {Therefore,} the two SDP formulations are equivalent{, proving Lemma~\ref{lem:SDP-equiv}.}

\subsection{From minion homomorphism to SDP rounding algorithm}

Recall that we are trying to prove that {for any promise template $\Gamma$,} the basic SDP decides $\PCSP({\Gamma}) $ if and only if $\cM_{\SDP} \to \Pol({\Gamma})$.  We begin by showing the {easier} direction that the minion homomorphism implies that the basic SDP decides $\PCSP({\Gamma})$. {In other words, for any instance of $\PCSP(\Gamma)$, the existence of a vector solution to alternative basic SDP of the instance implies that there exists a weakly satisfying assignment to the instance.}

\begin{theorem}\label{thm:hom_to_alg}
If $\cM_{\SDP} \to \Pol({\Gamma})$, then the basic SDP decides $\PCSP({\Gamma})$.
\end{theorem}
{
\begin{proof}
Assume that $(D_1, D_2)$ is the domain of $\Gamma$ with homomorphism $h : D_1 \to D_2$. Consider an instance of $\PCSP(\Gamma)$ on variables $V = \{x_1, \hdots, x_n\}$ and clauses $\{C_1, \hdots, C_m\}$ such that for all $j \in [m]$, the clause $C_j$ corresponds to $(P^{(j)}, Q^{(j)}) \in \Gamma$ of arity $k_j$. Let $S_j$ be the set of variables appearing at least once in $C_j$. We also define $\cF_j$ as the set of assignments $f : S_j \to D_1$ for which $f(C_j) \in P^{(j)}$ and $\cG_j$ as the set of assignments $g : S_j \to D_2$ for which $g(C_j) \in Q^{(j)}$.

By Lemma~\ref{lem:SDP-equiv}, we may without loss of generality assume that there exist vectors $\{\bv_{i,a} : i \in [n], a \in D_1\}$ and $\{\bw_{j,f} : j \in [m], f \in S_j\}$ such that (b1), (b2), and (b3) hold, where we define $\bv_i(a) = \bv_{i,a}$ and $\bw_j(f) = \bw_{j,f}$.
If no such vectors exist, we know that our instance is not strongly satisfiable by the contrapositive of Proposition~\ref{prop:basic-SDP-completeness}.

Let $\psi : \cM_{\SDP} \to \Pol(\Gamma)$ be the guaranteed minion homomorphism. Let $\id_{D_1} : D_1 \to D_1$ be the identity map. We define our weak assignment $\sigma : V \to D_2$ as follows. For all $i \in [n]$, we set
\[
    \sigma(x_i) := \psi(\bv_i)(\id_{D_1}).
\]
Unpacking this, $\psi(\bv_i) \in \Pol(A, B)^{D_1}$, so $\psi(\bv_i)$ has type $\psi(\bv_i) : D_1^{D_1} \to D_2$. The ``canonical'' element of $D_1^{D_1}$ is the identity map, so we plug that into the function to get a value in $D_2$.

We now seek to show that for every $j \in [m]$, that $\sigma|_{S_j} \in \cG_j$. First, we need to eliminate any repetition of variables. To do this, observe that $(\cF_j, \cG_j)$ can be viewed as a promise relation of arity $S_j$ since it is obtained from $(P^{(j)}, Q^{(j)})$ by setting some tuples in the relation equal. As such, we have that $\Pol({\Gamma}) \subseteq \Pol(P^{(j)}, Q^{(j)}) \subseteq \Pol(\cF_j, \cG_j)$. Therefore, we can think of $\psi$ as a minion homomorphism from $\cM_{\SDP}$ to $\Pol(\cF_j, \cG_j)$. Thus, since $\bw_j \in \cM_{\SDP}^{\cF_j}$, we have that $\psi(\bw_j) \in \Pol(\cF_j, \cG_j)^{\cF_j}$. Thus, $\psi(\bw_j)$ has type $\cF_j^{D_1} \to D_2$.

For $x_i \in S_j$, let $\pi_i : \cF_j \to D_1$ be the evaluation map $\pi_i(f) = f(x_i)$. We check that for all $x_i \in S_j$ and $a \in D_1$, we have that
\begin{align}
(\pi_i \oslash \bw_{j})(a) = \sum_{\substack{f \in \cF_j\\f(x_i) = a}} \bw_{j,f} = \bv_{i}(a),\label{eq:pi-map}
\end{align}
by definition of the alternative basic SDP, so $\pi \oslash \bw_j = \bv_i$. Thus, for all $x_i \in S_j$ we have that
\begin{align*}
    \sigma(x_i) &= \psi(\bv_i)(\id_{D_1}) & \text{(definition of $\sigma$)}\\
    &= \psi(\pi_i \oslash \bw_j)(\id_{D_1})  & \text{(\ref{eq:pi-map})}\\
    &= (\pi_i \oslash \psi(\bw_j))(\id_{D_1}) & \text{($\psi$ is minion homomorphism)}\\
    &= \psi(\bw_j)(\id_{D_1} \circ \pi_i) & \text{(\ref{eq:pol-min})}\\
    &= \psi(\bw_j)(\pi_i \circ \id_{\cF_j}),
\end{align*}
where for the last line $\id_{D_1} \circ \pi_i = \pi_i = \pi_i \circ \id_{\cF_j}$.

Apply (\ref{eq:pol-cond}) with $A = \cF_i$ and $(P_i, Q_i) := (\cF_j, \cG_j)$. Since $\psi(\bw_j) \in \Pol(\cF_j, \cG_j)^{\cF_j}$ and $\id_{\cF_j} : \cF_j \to \cF_j$, we have that
\[
    (\sigma(x_i) : x_i \in S) = (\psi(\bw_j)(\pi_i \circ \id_{\cF_j}) : x_i \in S) \in \cG_j,
\]
as desired. Thus, $\sigma$ is indeed a satisfying assignment to the original instance. Therefore, the basic SDP decides $\PCSP(\Gamma)$.
\end{proof}
}

\subsection{From SDP rounding algorithm to minion homomorphism}
We now prove the converse. {To do this, we adapt} the proof techniques of {Theorem 4.11 of }\cite{ciardo2022clap}. 

\begin{theorem}
If the basic SDP decides $\PCSP({\Gamma})$, then $\cM_{\SDP} \to \Pol({\Gamma})$.
\end{theorem}

\begin{proof}
{Assume that $(D_1,D_2)$ is the domain of $\Gamma = \{(P_1,Q_1), \hdots, (P_l, Q_l)\}$ with a homomorphism $h : D_1 \to D_2$ such that $h(P_i) \subseteq Q_i$ for all $i \in [l]$. To simplify the proof, we make the following technical assumption similar to the reduction in Remark~\ref{rem:Ragh}. If a relation $(P,Q) \in \Gamma$ of arity $k$, then every possible relation which can be created by repeating variables is also in $\Gamma$. For example, if $P, Q$ have arity $3$, then $(P',Q') \in \Gamma$ where $P' = \{(x,y) \in D_1^2 : (x,y,x) \in P\}$ and $Q' = \{(x,y) \in D_2^2 : (x,y,x) \in Q\}$. Note that $\Pol(P,Q) \subseteq \Pol(P',Q')$, so adding $(P',Q')$ to $\Gamma$ does not make $\cM_{\SDP} \to \Pol({\Gamma})$ harder to verify. Likewise, since repeating variables is allowed in PCSP instances, whether the basic SDP decides $\PCSP({\Gamma})$ is unaffected by these extra relations.

Assume now that the (alternative) basic SDP decides $\PCSP(\Gamma)$. By Lemma~\ref{lem:free-structure}, it suffices to prove that $\F_{\cM_{\SDP}}(\Gamma)$ is a promise structure. In other words, we must construct a map $\hat{h} : \cM_{\SDP}^{D_1} \to D_2$ such that $\hat{h}(\F_{\cM_{\SDP}}(P_i)) \subseteq Q_i$ for all $i \in [l]$. }

{
By the De Bruijn-Erd{\H o}s Theorem\footnote{{See \cite{ciardo2022clap} or Remark 7.13 of \cite{BBKO21} for additional context.}}~\cite{bruijn1951colour}, it suffices to prove that any finite subsets $F_i \subseteq \F_{\cM_{\SDP}}(P_i)$ for $i \in [l]$, we have that there exists $\hat{h} : \cM_{\SDP}^{D_1} \to D_2$ such that $\hat{h}(F_i) \subseteq Q_i$ for all $i \in [l]$.

Note that each $F_i \subseteq (\cM_{\SDP}^{D_1})^{k_i}$, let $V$ be the subset of $\cM_{\SDP}^{D_1}$ which appears in some tuple of $F_i$ for some $i \in [l].$ It suffices to construct $\hat{h} : V \to D_2$ rather than all of $\cM_{\SDP}^{D_1}$. For each $t \in F := F_1 \cup \cdots F_l$, let $(P_t, Q_t)$ be promise pair corresponding to $t$, and let $S_t \subseteq V$ be the set of distinct elements which appear in $t$. Since $\Gamma$ is closed under identification of variables, we can impose on $t$ a clause $C_t$ corresponding to $(P'_t, Q'_t)$ such that each variable of $S_t$ appears exactly once.

We view $\Phi := (V, \{C_t : t \in F\})$ as an instance of $\PCSP(\Gamma)$. If we can show this instance is accepted by the basic SDP, then we know there is an assignment $\hat{h} : V \to D_2$ which weakly satisfies every clause of $F_1 \cup \cdots \cup F_l$ since replacing $(P_t, Q_t)$ with $(P'_t, Q'_t)$ does not change whether the clause $\hat{h}|_{S_t}(t) \in Q_t$ if and only if $\hat{h}|_{S_t}(S_t) \in Q'_t$. Thus, $\hat{h}$ has precisely the property we want.

Thus, we conclude the proof by constructing an alternative basic SDP solution for $\Phi$. For convenience, let $V = \{x_1, \hdots, x_n\}$ and for each $t \in F$, let $\cF_t$ be the set of possible satisfying assignments $f : S_t \to D_1$ to $C_t$.

For each $i \in [n]$, we let $\bv_i \in \cM_{\SDP}^{D_1}$ be $x_i \in \cM_{\SDP}^{D_1}$ itself. This ensures that (b1) is satisfied. For each $t \in F$, let $\cF_t$ be the set of satisfying assignments to $f : S_t \to D_1$ to $P'_t$. By definition of $\F_{\cM_{\SDP}}(P'_t)$, there exists $\bw_t \in \cM_{\SDP}^{\cF_t}$ such that $\pi_i \oslash \bw_t = x_{i} = \bv_{i}$ whenever $x_i \in S_t$, where $\pi_i : \cF_t \to D_1$ is the evaluation map $\pi_i(f) = f(x_i)$. The condition $\bw_t \in \cM_{\SDP}^{\cF_j}$ is precisely (b2) and the equation $\pi_i \oslash \bw_t = \bv_i$ is precisely (b3) (see Remark~\ref{rem:b3}). Thus, $\Phi$ is indeed accepted by the alternative basic SDP, completing the proof.}
\end{proof}

This completes the proof of Theorem~\ref{thm:basic-SDP}, giving a minion homomorphism characterization of when the basic SDP works for deciding a PCSP.

\section{Conclusion}
\label{sec:conclusion}

We studied the robust satisfiability of promise CSPs, and specifically the power of SDPs in this context, revealing a number of new phenomena on both the algorithmic and integrality gap fronts. Our work brings to the fore a number of intriguing questions and directions. We list a few below.

\begin{itemize}
\item Can we get a robust {$\fiPCSP$} dichotomy result for all Boolean symmetric {promise templates}? We conjecture that every Boolean symmetric folded idempotent PCSP without $\MAJ$ or $\AT$ polymorphisms does not admit a robust algorithm. In our hardness result~\Cref{thm:main-hardness}, we showed the same when there is a single predicate pair. Extending our result to multiple predicate pairs is an interesting challenge. 

As a concrete example, consider the {promise template} $\Gamma=\{(P_1,Q_1),(P_2,Q_2)\}$ where $P_1 = \Ham_k \{\frac{k+1}{2}\},$ $Q_1=\Ham_k \{0,1,\ldots,k-1\},P_2=\Ham_t \{1\},Q_2=\Ham_t\{0,1,\ldots,t-2,t\}$ for odd integers $k,t\geq 3$. We have $\AT \subseteq \Pol(P_1,Q_1), \MAJ \subseteq \Pol(P_2,Q_2)$ while both $\AT,\MAJ \nsubseteq \Pol(\Gamma)$. To obtain the integrality gap for {the basic SDP relaxation of $\fiPCSP(\Gamma)$}, we need to show that there is no folded sphere coloring that respects $\Gamma$. There is a {folded} sphere coloring $f_1 : \S^n \rightarrow {\B}$ respecting $(P_1,Q_1)$, and a {folded} sphere coloring $f_2 : \S^n \rightarrow {\B}$ respecting $(P_2,Q_2)$ for every positive integer $n$. The challenge lies in showing that there is an integer $n$ for which there is no single {folded} coloring $f:\S^n \rightarrow {\B}$ that respects both $(P_1,Q_1)$ and $(P_2,Q_2)$ simultaneously. We are able to prove this for small values of $k,t$ under an assumption that ~\Cref{thm:sphere-ramsey} extends to the density setting: for any constant $0<\sigma<1$, and any tuple $S$ of vectors in $\S^n$ satisfying the conditions in~\Cref{thm:sphere-ramsey}, in any subset $T$ of $\S^n$ with spherical measure at least $\sigma$, there is a tuple of vectors $S'\subseteq T$ that is congruent to $S$, as long as $n \geq n_0:= n_0(S,\sigma)$. Both proving the density sphere Ramsey theorem, and showing that $\Gamma$ does not admit a sphere coloring respecting it are interesting open problems, and the former problem could have applications elsewhere as well. Very recently, density sphere Ramsey theorems for certain configurations were established by Guruswami and Li~\cite{guruswami2025density}.

On the other hand, extending the {robust $\fiPCSP$ classification to $\PCSP$ (i.e., no idempotence or folding) for all Boolean symmetric promise templates would require} a better understanding of polymorphisms of arbitrary Boolean symmetric {promise templates}. We remark that for the decision version of {PCSPs for} Boolean symmetric {promise tempaltes}, a dichotomy was first proved for the folded case~\cite{BrakensiekG21}, and later, the restriction was removed~\cite{FicakKOS19}, where the authors showed that the decision version of a Boolean symmetric PCSP can be solved in polynomial time if and only if it has Threshold, $\AT$ or Parity polymorphisms. Our algorithm for the $\MAJ$ polymorphisms extends to the setting when there are threshold polymorphisms, similar to the algorithmic result of~\cite{FicakKOS19}.
Combining the polymorphic ideas in~\cite{FicakKOS19} with our sphere coloring results is a potential venue to generalize our hardness results to general Boolean symmetric PCSPs. 

Finally, can we get a robust {PCSP} dichotomy result for general Boolean {promise templates}? We believe new algorithmic techniques are needed to understand what polymorphic families lead to robust algorithms. For the symmetric folded case, $\MAJ$ and $\AT$ polymorphisms resulted in robust algorithms while Parity polymorphisms resulted in just the decision version being solved in polynomial time. This suggests that the existence of a suitable notion of noise stable family of polymorphisms could be the key to robust algorithms for PCSPs. 

\item Can we improve the quantitative aspects of our robust approximation algorithm based on the Majority polymorphism to satisfy a $1-O(\sqrt{\eps})$ fraction of the constraints (which would be optimal~\cite{KKMO07} under the UGC)? Even for {robust} Max-Cut, such an algorithm that makes black-box use of the Majority polymorphism is not known. More generally, we do not know how to translate a polymorphism for the CSP into an approximation or robust satisfaction algorithm for it. For a suitable notion of  ``approximate" polymorphisms, there is such a connection~\cite{BR16}. {Of note, in the case of CSP templates of arity $k$ (over an arbitrary domain) with a ``near-unanimity operator'' (of which $\MAJ_3$ is a special case), a robustness algorithm of $1-\eps$ versus $1-O(\eps^{1/k})$ is known~\cite{dalmau2019Robust}.}

\item For the case of AT polymorphisms, our algorithm incurs an exponential loss, and only satisfies $1-O\left(\frac{\log \log \frac{1}{\epsilon}}{\log \frac{1}{\epsilon}}\right)$ fraction of the constraints in a $(1-\eps)$-satisfiable instance. Is this inherent (which is known to be the case for Horn-SAT~\cite{GuruswamiZ12})? Can one show (Unique Games) hardness of how robustly one can approximate concrete {promise templates} like $1$-in-$3$-SAT vs. NAE-$3$-SAT?

\item For CSPs it is known that if Sherali Adams LP correctly ascertains exact satisfiability, then the CSP has a robust satisfaction algorithm (albeit not LP based). Does this connection remain true for PCSPs (similar to~\Cref{conj:sdp-robust})? 

\item We have minion characterizations for the solvability of PCSPs by SDPs (this work) and the Sherali Adams hierarchy~\cite{ciardo2022sherali}. What is the relationship between these minions?  Do there exist PCSPs whose exact satisfiability is decided by the basic SDP but which cannot be decided by $O(1)$ levels of the Sherali Adams hierarchy? Note that for CSPs, both these classes are the same and coincide with the class of bounded width CSPs.
        
\item As mentioned earlier, our algorithm for Majority polymorphism can be generalized for arbitrary threshold polymorphisms. Together with the hardness results in \cite{BGS21}, this gives a conditional dichotomy w.r.t robust {PCSP} satisfiability for the class \emph{ordered} Boolean {promise templates} (whose polymorphisms are monotone functions). We leave establishing even more ambitious dichotomies as an intriguing open question.
\end{itemize}

\appendix %
\section*{Appendix}
\appendix
{
\section{Efficient SDP Solution}\label{app:efficient-sdp}

\newcommand{\fcirc}{\bullet}
\newcommand{\hull}{\operatorname{hull}}

To begin, state a reformulated basic SDP for which it is easier to reason about approximate solutions. Let $A \fcirc X$ denote the entry-wise product of matrices. If $\cF$ is a collection of matrices, we let $\hull(\cF)$ denote the convex hull of these matrices--the set of all convex combinations. If $\cF$ is a finite set, then $\hull(\cF)$ can be described by a finite family of inequalities (and thus with bounded coefficients).

Given a clause $C_j$ with variable set $S_j$, we let $\mathcal F_j$ be the set of square matrices with $|S_j|+2$ rows and columns with the following properties. For every $\sigma : S_j \to \B$, we have a matrix $M_\sigma \in \cF_j$ corresponding to the outer product of the vector $(1, \sigma, \eta_{\sigma})$, where $\eta_{\sigma} = 0$ if $\sigma$ is a \emph{satisfying} assignment to $C_j$ and $\eta_{\sigma}=1$ otherwise. We also define $\mathcal G_j$ to be a set of square matrices with $|S_j|+1$ rows corresponding to the outer products of the form $(1, \sigma)$ where $\sigma : S_j \to \B$.

Let $X$ be a symmetric matrix with $n+m+1$ rows/columns---the last $m$ columns correspond to the ``error'' in each of the $m$ clauses. Given a set $S$ of indices, we let $X|_{S}$ denote the square submatrix whose rows and columns come from $S$.  We consider the following semidefinite program similar to one considered by Karloff and Zwick for MAX 3-SAT~\cite{KZ}.
\begin{align*}
    \textbf{minimize: } &\sum_{j=1}^m X_{j+n,j+n}\\
    \textbf{subject to: } &X|_{S_j \cup \{0, j+n\}} \in \hull(\mathcal F_j), \forall j \in [m]\\
        &X \succeq 0.
\end{align*}

By Theorem 5.1.1 of \cite{gartner2012approximation}, for any $\delta = (n/\eps)^{-O(1)}$ we can efficiently find a realization $M$ of $X$ such that $M$ has rational entries, $M$ is PSD, the objective is optimized up to an additive $\delta$, and every linear constraint is satisfied up to an additive $\delta$. For the last condition, this means we ``only'' have that
\[
M|_{S_j \cup \{0, j+n\}} + E_j \in \hull(\mathcal F_j),
\]
where $E_j$ is an ``error'' matrix in which all entries have absolute value at most $\delta$. We only know that $M_{i,i} \in [1-\delta, 1+\delta]$ for all $i \in \{0,1,\hdots, n\}$. Let $M'$ be a matrix which is the same as $M$ except $M'_{i,i} = 1+\delta$ for all $ i \in \{0, 1, \hdots, n\}$. By definition, $D' := M' - M$ is a diagonal matrix with nonnegative entries, so $M' \succeq 0$. Let $N = \frac{1-\lambda}{1+\delta} \cdot M'|_{\{0,1,\hdots, n\}} + \lambda I$ for some $\lambda \approx \Theta(\delta)$. Note that $N \succeq 0$ as well and that $N_{i,i} = 1$ for all $i \in \{0,1, \hdots, n\}$.

We claim that for all $j \in [m]$ that $N|_{\{0\} \cup S_j} \in \hull(\cG_j)$. If not, there exist a matrix $A$ and number $b$ such that such that $A \fcirc X \le b$ for all $X \in \hull(\cG_j)$ but $A \fcirc N > b$. Since $N_{i,i} = 1$ for all $i \in \{0,1,\hdots, n\}$ and $X_{i,i} = 1$ for all $X \in \hull(\cG_j)$, we may assume without loss of generality that $A_{i,i} = 0$ for all $i \in \{0,1,\hdots, n\}$. Likewise, since all matrices in $\hull(\cG_j)$ are symmetric, we may assume that $A$ is symmetric.  Thus, since $I \in \hull(\cG_j)$, we may assume that $b \ge 0$. In fact, we now show for nonzero $A$ we need $b > 0$.

\begin{claim}
Let $A$ be a symmetric matrix with an all-zero diagonal such that $A \fcirc X \le 0$ for all $X \in \hull(\cG_j)$. Then, $A$ is the zero matrix.
\end{claim}
\begin{proof}
Consider indices $i < j$ and let $\cG'_j \subseteq \cG_j$ be the matrices $Y \in \cG_j$ for which $Y_{i,j} = 1$. Let $Z$ be the (uniform) average of all matrices in $\cG'_j$. It is straightforward to verify that $Z_{i,j} = Z_{j,i} = 1$, but all other off-diagonal entries are zero. Since $Z, -Z \in \hull(\cG_j)$, we have that
\begin{align*}
    A_{i,j} &= A \fcirc Z \le 0,\\
    -A_{i,j} &= A \fcirc (-Z) \le 0.
\end{align*}
Therefore, $A_{i,j} = 0$ for all $i < j$. Thus, $A$ is the zero matrix.
\end{proof}

Since $\hull(\cG_j)$ is a projection of $\hull(\cF_j)$, we have that this inequality is approximately satisfied by $M$: $A \fcirc M|_{\{0\} \cup S_j} \le b + O(\delta)$. Thus, 
\begin{align*}
b &< A \fcirc N\\
&=  A \fcirc (\frac{1-\lambda}{1+\delta} \cdot (M + D')|_{\{0\} \cup S_j} + \lambda I)\\
&= \frac{1-\lambda}{1+\delta} \cdot (A \fcirc M|_{\{0\} \cup S_j} + \frac{1-\lambda}{1+\delta} \cdot (A \fcirc D'|_{\{0\} \cup S_j}+ \lambda (A \fcirc I)\\
&\le \frac{1-\lambda}{1+\delta} b + O(\delta),
\end{align*}
where we use the fact that every entry of $D'$ is bounded in absolute value by $\delta$ and $A \fcirc I = 0$. 
Thus, for $\lambda = \Omega(\delta)$, we get a contradiction. This proves that $N|_{\{0\} \cup S_j} \in \hull(\cG_j)$ for all $j \in [m]$.

Using Gaussian elimination (see O'Donnell~\cite{odonnell2016SOS} and references therein, including~\cite{lovasz2003Semidefinite}), we can efficiently compute rational square matrices $L$ and $D$, where $L$ is lower triangular and $D$ is a diagonal, such that $N = LDL^{T}$. Then, we can compute vectors $\bv_0, \hdots, \bv_n$ whose coordinates are square roots of rational numbers such that $N_{i,j} = \langle \bv_i,\bv_j\rangle$. Since $N|_{\{0\} \cup S_j} \in \hull(\cG_j)$ for all $j \in [m]$, by solving a suitable linear program, we can efficiently find a probability distribution $\lambda_j$ over assignments $S_j \to D$ which is consistent with the first and second moment conditions of the fiPCSP basic SDP, where the objective value is to minimize the total weight $\epsilon_j$ on non-assignments to $C_j$.

To finish, it suffices to prove that $\epsilon_j \le M_{j+n,j+n} + O(\delta)$ for all $j \in [m]$. Recall that $M|_{S_j \cup \{0, j+n\}} + E_j \in \hull(\mathcal F_j)$. Thus, $N' := (M|_{S_j \cup \{0, j+n\}} + E_j)|_{\{0\}\cup S_j} \in \hull(\mathcal G_j)$. Furthermore, if we solve the same ``error'' LP for $N'$, we will get an error $\epsilon'_j$ which must be at most $M_{j+n,j+n}+\delta.$ One can see that every entry in $N$ and $N'$ differ by at most $O(\delta)$. Let $\epsilon : \hull(\mathcal G_j) \to [0,1]$ be the value of the LP on all possible input matrices. It is clear that $\epsilon$ is a bounded convex function, so it must be Lipschitz. Thus, $\epsilon_j - \epsilon'_j \le O(\delta)$, so $\epsilon_j \le M_{j+n,j+n} + O(\delta)$, as desired. 
}

\section{Missing Proofs}
\label{sec:missing-proofs}

\subsection{{Missing proofs from Section 4}}

\begin{proof}[Proof of Lemma~\ref{lem:equal-case}]
We first consider the case when $k=1$. Without loss of generality, let $P=Q=\{+1\}$, and we use $\textbf{v}_1 = \alpha \textbf{v}_0 + \textbf{v}'_1$ to denote the SDP vector corresponding to the variable used in the constraint. 
As the basic SDP has error at most $\sqrt{\epsilon}$, we get that 
\[
\alpha_1 \geq 1-\sqrt{\epsilon}
\]
As $\alpha_1^2 + \norm{v'_1}_2^2 = 1$, $\norm{v'_1}_2 \leq O(\epsilon^{0.25})$. Thus, using~\Cref{prop:gaussian-concentration}, we get that $\langle \zeta, \textbf{v}'_1 \rangle \leq O(\epsilon^{0.25}r)$ with probability at least $1-e^{\frac{-r^2}{2}} \geq 1-\sqrt{\epsilon}$. On the other hand, using~\Cref{prop:gaussian-anticoncentration}, we get that $|\langle \zeta, \textbf{v}_0 \rangle | \geq \frac{1}{r}$ with probability at least $1-\frac{1}{r}$. This implies that 
\[
\delta \alpha_1 | \langle \zeta, \textbf{v}_0 \rangle | \geq \frac{\delta }{2r}
\]
Thus, with probablity at least $1-O(\frac{1}{r})$, we have 
\[
\langle \zeta, \textbf{v}'_1 \rangle \leq O(\epsilon^{0.25}r) < \frac{\delta}{2r} \leq \delta \alpha_1 | \langle \zeta, \textbf{v}_0 \rangle |
\]
Hence, with probability at least $1-O(\frac{1}{r})$, we round the variable to $+1$. 

We now consider the general case when $k \geq 2$. Note that the above proof for $k=1$ holds when $P = \Ham_k \{0\}$ or when $P=\Ham_k \{k\}$. We are left with the setting when $P=\Ham_k \{0,k\}$.
In order to show that our algorithm is a robust algorithm for this {fi}PCSP, it suffices to show that all the elements in the predicates are rounded to the same value with high probability. Consider $i,j \in [k]$. We show that the probability that the variables $x_i$ and $x_j$ get rounded to different values is at most $O\left( \frac{1}{r}\right)$. Using the union bound over all the $\binom{k}{2}$  pairs of indices, we get our required claim. 

We first collect useful properties using the fact that the basic SDP is supported with a probability of at least $1-c$ on $P$. 
\begin{enumerate}
    \item (First moment.) We have 
    \[
    | \mu_i - \mu_j | \leq 2c. 
    \]
    \item (Second moment.) We have 
    \[
    \langle \textbf{v}_i, \textbf{v}_j \rangle \geq 1-2c
    \]
    Using this, we get 
    \begin{align*}
        \norm{\textbf{v}'_i - \textbf{v}'_j}_2^2 &= \textbf{v}_i - \textbf{v}_j + (\alpha_j - \alpha_i)\textbf{v}_0 \\ 
        &\leq \norm{\textbf{v}_i - \textbf{v}_j}_2^2 + (\alpha_i - \alpha_j)^2 \\ 
        &\leq O(c). 
    \end{align*}
    Thus, $\norm{\textbf{v}'_i - \textbf{v}'_j} \leq O(\sqrt{c})$.
\end{enumerate}
As earlier, we assume that $c$ is at most $\sqrt{\epsilon}$.

Recall that our goal is to upper bound the probability that $x_i$ and $x_j$ are rounded to different values. Without loss of generality, suppose that $x_i$ is rounded to $+1$, and $x_j$ is rounded to $-1$. 
We get that 
\begin{align*}
    \langle \zeta, \textbf{v}'_i \rangle &\geq \delta \alpha_i | \langle \zeta, \textbf{v}_0 \rangle | \label{eq:1}\\ 
    \langle \zeta, \textbf{v}'_j \rangle &< \delta \alpha_j | \langle \zeta, \textbf{v}_0 \rangle | 
\end{align*}
Using~\Cref{prop:gaussian-concentration}, we can infer that 
\[
| \langle \zeta, \textbf{v}'_i \rangle - \langle \zeta, \textbf{v}'_j \rangle | \leq O(r\sqrt{c})
\]
with probability at least $1-\frac{1}{r}$. 

We consider two cases: first, when $|\alpha_i| \leq \frac{1}{2}$. As $\alpha_i^2 + \norm{\textbf{v}'_i}_2^2=1$, and $|\alpha_i - \alpha_j | \leq 2c$, we get that $\norm{\textbf{v}'_j}=\Omega(1)$. 
In this case, we have 
\begin{align*}
    \langle \zeta, \textbf{v}'_j \rangle &\geq \langle \zeta, \textbf{v}'_i \rangle - O(r \sqrt{c}) \\ 
    &\geq \delta \alpha_i | \langle \zeta, \textbf{v}_0 \rangle | - O(r\sqrt{c}) \\ 
\end{align*}
Thus, we have 
\[
\langle \zeta, \textbf{v}'_j \rangle \in [\delta \alpha_i | \langle \zeta, \textbf{v}_0 \rangle | - O(r\sqrt{c}),\delta \alpha_j | \langle \zeta, \textbf{v}_0 \rangle | ]
\]
Here, $\langle \zeta, \textbf{v}'_j \rangle \in [p,q]$ where $q-p \leq O(\delta r)+O(r \sqrt{c}) \leq O(\delta r)$. However, as $\norm{\textbf{v}'_j} \geq \Omega(1)$, this happens with probability at most $O(\delta r)$. 

Now, suppose that $|\alpha_i|\geq \frac{1}{2}$. 
We have 
\begin{align*}
    \delta \alpha_i | \langle \zeta, \textbf{v}_0 \rangle | &\geq  \delta \alpha_j | \langle \zeta, \textbf{v}_0 \rangle |- 2\delta c | \langle \zeta, \textbf{v}_0 \rangle |
\end{align*}
However, as $\langle \zeta, \textbf{v}_0 \rangle \sim \mathcal{N}(0,1)$, we have $|\langle \zeta, \textbf{v}_0 \rangle | \leq r$ with probability at least $1-\sqrt{\epsilon}$. 
Thus, with probability at least $1-\sqrt{\epsilon}$, we have 
\[
\langle \zeta, \textbf{v}'_i \rangle \geq \delta \alpha_i | \langle \zeta, \textbf{v}_0 \rangle | \geq  \langle \zeta, \textbf{v}'_j \rangle  - 2\delta c r 
\]
We have $\delta \alpha_i |\langle \zeta, \textbf{v}_0 \rangle | \in [p,q]$ where $q-p \leq O(\delta c r)+ O(r\sqrt{c})$. However, this happens with probability at most $O(\frac{r\sqrt{c}}{\delta}) \leq O(\frac{1}{r})$.
\end{proof}

\begin{proof}[Proof of Lemma~\ref{lem:AT}]
Suppose that $\textbf{a} = \sgn(\textbf{x}-\textbf{y})$ for $\textbf{x}, \textbf{y} \in \Aff(P)$, and $x_i \neq y_i$ {for all} $i \in [k]$. By modifying the affine combinations slightly, we can assume that $\textbf{x}$ and $\textbf{y}$ are rational affine combinations of $P$ while still preserving the fact that $\textbf{a}=\sgn(\textbf{x}-\textbf{y})$. In other words, there exist $p_1, p_2, \ldots, p_K, q_1, q_2, \ldots, q_K \in \mathbb{Q}$ such that $\sum_{i \in [K]}p_i = \sum_{i \in [K]}q_i = 1$, and $\textbf{x} = \sum_{i \in [K]}p_i \textbf{a}_i$, $\textbf{y}=\sum_{i \in [K]}q_i \textbf{a}_i$,  where ${\B}^k = \{\textbf{a}_1, \textbf{a}_2, \ldots, \textbf{a}_K\}$. Let $N$ be a positive integer such that we can write $p_i = \frac{p'_i}{N}, q_i = \frac{q'_i}{N}$ where $p'_i, q'_i$ are integers for every $i \in [K]$. 

\newcommand{\bs}{\mathbf{s}}

Let $S$ be a {list of elements of} of ${\B}^k$ where {for each} $i \in [K]$, {there are $|p'_i|+|q'_i|$} copies of $\textbf{a}_i${.} {Let $\bs \in \B^{|S|}$ be a $\pm 1$ vector such that for each $i \in [K]$ we have that $|p'_i|$ of the $j \in [|S|]$ with $S_j = \ba_i$ have $s_j = \sgn(p'_i)$ and the other $|q'_i|$ copies have $s_j = \sgn(q'_i)$. Since $\sum_{i \in [K]}p'_i = \sum_{i \in [K]}q'_i$, we have that $s_j$ has an equal number of each of $+1$ and $-1$.}

Let $\textbf{z}$ denote the signed sum of all vectors (including repetitions) in $S$ {according to the signs of $\bs$}. Note that $\sgn(\textbf{z})=\sgn(\textbf{x}-\textbf{y})$. As each element of $\mathbf{z}$ is an integer, we get that the absolute value of each coordinate in $\mathbf{z}$ is at least $1$. Furthermore, we can take multiple copies of $S$ to ensure that the absolute value of each coordinate in $\mathbf{z}$ is at least $2$. Now, we add an arbitrary element of $P$ with sign $+1$ to $S$. Note that we still have that the signed sum of $S$, i.e., the updated $\mathbf{z}$ satisfies $\sgn(\mathbf{z})=\sgn(\mathbf{x}-\mathbf{y})$. Furthermore, $\mathbf{z}=\textbf{x}_1 - \textbf{x}_2 + \ldots + \textbf{x}_L$ where each $\textbf{x}_i \in P$. Thus, $\textbf{a}=\sgn(\textbf{x}-\textbf{y})=\sgn(w)=\sgn(\textbf{x}_1 - \textbf{x}_2 + \ldots + \textbf{x}_L) \in O_{AT}(P)$. Thus, 
\[
\{\sgn(\textbf{x} - \textbf{y}) : \textbf{x}, \textbf{y} \in \Aff(P), \forall i, x_i \neq y_i \} \subseteq O_{AT}(P)
\]

To prove the other direction, suppose that $\textbf{a} \in O_{AT}(P)$. That is, $\textbf{a}=\sgn(\textbf{x}_1 - \textbf{x}_2 + \ldots + \textbf{x}_L)$. Let $S$ be a multiset of $\textbf{x}_1, \textbf{x}_2, \ldots, \textbf{x}_L$ with the corresponding sign as in the summation. As $P$ is non-trivial in every coordinate i.e., for every $i \in [k]$, there exist assignments $\textbf{x}$ in $P$ where $x_i=+1$, and similarly, $\textbf{x}' \in P$ where $x'_i=-1$. By adding vectors with both signs $+1$ and $-1$, we can assume that $S$ is non-trivial in every coordinate while still preserving the fact that the sign vector of the signed sum of $S$ is equal to $\sgn(\textbf{a})$. We modify $S$ while still preserving this property to ensure that the signed sum of vectors in $S$ has an absolute value of at least $2$ in every coordinate. 

As there are an odd number of vectors in $S$, the signed sum of the vectors has an absolute value of at least $1$ in every coordinate. 
Fix a vector $\textbf{x}_i \in S$. Create two copies of every other vector in $S$ (with the same sign as the original). Note that this operation does not alter the sign vector of the signed sum of the vectors in $S$. 
We can repeat this process at most $2k$ times to ensure that in the final multiset $S$, the signed sum has an absolute value of at least $2$ in every coordinate. Finally, we add an arbitrary vector with sign $-1$ to $S$, to ensure that the number of vectors with $+1$ sign and the number of vectors with $-1$ sign are equal in $S$. Overall, we get that there are $\textbf{x}_1, \textbf{x}_2, \ldots, \textbf{x}_N \in P$ and $\textbf{y}_1, \textbf{y}_2, \ldots, \textbf{y}_N \in P$ such that
\[\textbf{a}=\sgn(\textbf{x}_1+\ldots +\textbf{x}_N-\textbf{y}_1-\ldots -\textbf{y}_N) = \sgn \left( \frac{1}{N} \textbf{x}_1 + \ldots + \frac{1}{N}\textbf{x}_N-\frac{1}{N} \textbf{y}_1 - \ldots -\frac{1}{N}\textbf{y}_N\right)
\]
Thus, we get that $\textbf{a} {\in} \{\sgn(\textbf{x} - \textbf{y}) : \textbf{x}, \textbf{y} \in \Aff(P), \forall i, x_i \neq y_i \}$, completing the proof that 
\[
O_{AT}(P)\subseteq \{\sgn(\textbf{x} - \textbf{y}) : \textbf{x}, \textbf{y} \in \Aff(P), \forall i, x_i \neq y_i \}
\qedhere \]
\end{proof}

\subsection{{Missing proofs from Section 5}}

\begin{proof}[Proof of~\Cref{lem:reduction}]
We extensively use the properties of $\AT, \MAJ$ polymorphisms of Boolean symmetric folded idempotent PCSPs proved in~\cite{BrakensiekG21}.\footnote{The notion of ppp-reduction used in~\cite{BrakensiekG21} implicitly allows for equality constraints. However, the ppp-reductions we use in this proof (particularly those in Claims 4.2 and 4.4 of ~\cite{BrakensiekG21}) do not use equalities.}
We recall that $O_{\AT}(P)$ ({respectively} $O_{\MAJ}(P)$) denotes the set $\bigcup_{\textbf{x}_1,\ldots,\textbf{x}_L\in P,L \in \mathbb{N},\text{ odd},}\AT_L(\textbf{x}_1,\ldots,\textbf{x}_L)$ ({respectively} $\MAJ_L$) for a predicate $P$.

Let $k$ denote the arity of $P,Q$, i.e., $P \subseteq Q \subseteq {\B}^k$.
Note that $P \nsubseteq \{ (-1, \ldots, -1), (+1,\ldots, +1)\}$ as in that case $O_{\MAJ}(P)  =P \subseteq Q$, contradicting the fact that $\MAJ \nsubseteq \Pol(P,Q)$. Thus, there exists $l \in \{1,2,\ldots, k-1\}$ such that $\Ham_k \{l\} \subseteq P$. 

\smallskip \noindent \textbf{Case 1.} We first consider the case when $P=\Ham_k \{l\}$ for some $l \in \{1,2,\ldots, k-1\}$. As $P$ is symmetric, $O_{\MAJ}(P)$ is symmetric as well~\cite{BrakensiekG21}.
Furthermore, as $\MAJ \nsubseteq \Pol(P,Q)$, there exists $b \in \{0,1,\ldots, k\}$ such that $\Ham_k \{b\} \cap Q=\phi$ and $\Ham_k \{b\} \subseteq O_{\MAJ}(P)$. 

Suppose that $b \notin \{0,k\}$. 
Let $Q' = {\B}^k \setminus \Ham_k \{b\}$. By definition, $\MAJ \nsubseteq \Pol(P,Q')$. 
Using the fact that $O_{\AT}(\Ham_k \{l\}) = \Ham_k \{1,2,\ldots,k-1\}$, we get that $\AT \nsubseteq \Pol(P,Q')$. 
Thus, we get a {promise template} $(P,Q)$ that is {fi}ppp-definable from original {promise template} where $P=\Ham_k \{l\}$, $Q=\Ham_k \{0,1,\ldots,k\} \setminus \{b\}$ where $b \in \{1,2,\ldots, k-1\} \setminus \{l\}$. 
Note that $\MAJ, \AT \notin \Pol(P,Q)$. 
We now modify this {promise template further}, updating $P,Q,l,k,b$ while preserving the following two properties:
\begin{enumerate}
    \item At every step, $P = \Ham_k \{l\}, Q = \Ham_k \{0,1,\ldots,k\} \setminus \{b\}$ where $b \in \{1,2,\ldots,k-1\}\setminus\{l\}$. 
    \item $\MAJ, \AT \notin \Pol(P,Q)$. 
\end{enumerate}
As $O_{\MAJ}(P) = \Ham_k \{ 0,1,\ldots, k\} \cap \{2l-k+1,\cdots,2l-1\}$, and $\Ham_k \{b\} \in O_{\MAJ}(P)$,  we get that $b \in \{2l-k+1,\cdots,2l-1\} \cap \{0,\cdots,k\}$. Furthermore, as $b>0$, we get that $l>1$. Similarly, we get that $l <k-1$. This also implies that $k \geq 4$ as $l \in \{1,\ldots,k-1\}$.

We use the following two tools to repeatedly obtain a new {promise template} that is {fi}ppp-definable from the previous one. 
\begin{enumerate}
    \item Given a {promise template} $P = \Ham_k \{l\}, Q = \Ham_k \{0,1,\ldots,k\} \setminus \{b\}$ where $b \in \{1,2,\ldots,k-1\}\setminus\{l\}$, then the {promise template} $P' = \Ham_{k-1} \{l\}, Q' = \Ham_{k-1} \{0,1,\ldots,k-1\} \setminus \{b\}$ is {fi}ppp-definable from $(P,Q)$ (Claim $4.2$ of~\cite{BrakensiekG21}). As long as $b<k-1$ and $b \neq 2l-k+1$, this update preserves the above two properties. 
    \item Given a {promise template} $P = \Ham_k \{l\}, Q = \Ham_k \{0,1,\ldots,k\} \setminus \{b\}$ where $b \in \{1,2,\ldots,k-1\}\setminus\{l\}$, then the {promise template} $P' = \Ham_{k-1} \{l-1\}, Q' = \Ham_{k-1} \{0,1,\ldots,k-1\} \setminus \{b-1\}$ is {fi}ppp-definable from $(P,Q)$ (Claim $4.4$ of~\cite{BrakensiekG21}). As long as $b >1$ and $b \neq 2l-1$, this update preserves the above two properties. 
\end{enumerate}
Now, we update the {promise template} using the above two steps. As $k$ is decreasing at every step, this procedure terminates at some point. Then, either of the two conditions holds:
\begin{enumerate}
    \item $b=1, b=2l-k+1$. In this case, we get that $l = \frac{k}{2}$ and $b=1$. Thus, $(P',Q')$ is {fi}ppp-definable from $\Gamma$ where $P'=\Ham_{k}\{\frac{k}{2}\}, Q = \Ham_{k} \{0,1,\ldots,k\} \setminus \{1\}$ where $k$ is even and is at least $4$. 
    \item $b=k-1, b=2l-1$. In this case, we get that $l = \frac{k}{2}$ and $b=k-1$. Thus, $(P',Q')$ is {fi}ppp-definable from $\Gamma$ where $P'=\Ham_{k}\{\frac{k}{2}\}, Q = \Ham_{k} \{0,1,\ldots,k\} \setminus \{k-1\}$ where $k$ is even and is at least $4$. 
\end{enumerate}

Suppose that there is no $b \notin \{0,k\}$ such that $\Ham_k \{b\} \subseteq O_{\MAJ} (P) \setminus Q $. As $O_{\MAJ}(P) \nsubseteq Q$, by negating the variables if needed, we can assume that $\Ham_k \{0\} \in O_{\MAJ}(P) \setminus Q$. Furthermore, there exists $b \in \{1,2,\ldots, k-1\}$ such that $\Ham_k \nsubseteq  Q$ as $O_{\AT}(P) \nsubseteq Q$. Thus, we obtain $(P,Q)$ that is {fi}ppp-definable from the original {promise template} such that $P=\Ham_k \{l\}, Q=\Ham_k \{1,\ldots,k\}\setminus \{b\}$ where $l,b \in \{1,2,\ldots,k-1\}, b >2l-1$. By using the first type of update used above (Claim $4.2$ of~\cite{BrakensiekG21}), we obtain a new {promise template} $(P,Q)$ that is {fi}ppp-definable from the original {promise template} such that $P=\Ham_k \{l\}, Q=\Ham_k \{0,1,\ldots,k\}\setminus \{0,k-1\}$, where $l \in \{1,2,\ldots,k-1\}, l \leq \frac{k-1}{2}$. 

\smallskip \noindent \textbf{Case 2}. There exist $l \neq l'$ such that $\Ham_k \{l,l'\}\subseteq P$. Recall that $P \nsubseteq \Ham_k \{0,k\}$. This implies that $O_{\AT}(P)=\Ham_k \{0,1,\ldots,k\}$. Hence, we can get a {promise template} $(P,Q')$ that is {fi}ppp-definable from the original {promise template} such that $Q' = \Ham_k \{0,1,\ldots, k\} \setminus \{b\}$ and $\Ham_k \{b\} \nsubseteq O_{\MAJ}(P)$. 

For ease of notation, let $P=\Ham_k S$, where $S \subseteq \{0,1,\ldots,k\}$. First, consider the case when $\min S=0, \max S=k$. As mentioned earlier, we know that there exists $l \in \{1,2,\ldots,k-1\}$ such that $\Ham_k \{l\} \subseteq P$. Thus, we can obtain a new PCSP $(P,Q)$ that is {fi}ppp-definable from the original {promise template} where $P=\Ham_k \{0,l,k\}, Q = \Ham_k \{0,1,\ldots,k\}\setminus \{b\}$ and $b \notin \{0,l,k\}$. We consider three cases separately:
\begin{enumerate}
    \item Suppose that $l \leq \frac{k-1}{2}$. In this case, we have a {promise template} $(P,Q)$ that is {fi}ppp-definable from the original {promise template} where $P=\Ham_k \{l,k\}$ and $Q = \Ham_k \{0,1,\ldots,k\}\setminus \{b\}$ and $b \notin \{l,k\}$. Note that this {promise template} does not contain $\AT$ or $\MAJ$ as polymorphisms. 
    \item Suppose that $l = \frac{k}{2}$. In this case, we have a {promise template} $(P,Q)$ that is {fi}ppp-definable from the original {promise template} where $P=\Ham_k \{l\}$ and $Q = \Ham_k \{0,1,\ldots,k\}\setminus \{b\}$ where $b \notin \{0,l,k\}$. This also doesn't have $\AT$ and $\MAJ$ as polymorphisms. We have already shown that we can relax this further to the earlier mentioned three {promise templates}. 
    \item Suppose that $l \geq \frac{k+1}{2}$. In this case, we have a {promise template} $(P,Q)$ that is {fi}ppp-definable from the original PCSP where $P=\Ham_k \{0,l\}$ and $Q = \Ham_k \{0,1,\ldots,k\}\setminus \{b\}$ where $b \notin \{0,l\}$. Note that this {promise template} does not contain $\AT$ or $\MAJ$ as polymorphisms. 
\end{enumerate}
Thus, there is a {promise template} $(P,Q)$ that is {fi}ppp-definable from the original PCSP where $P=\Ham_k \{l,l'\}, Q = \Ham_k \{0,1,\ldots,k\}\setminus \{b\}$ such that $l <l',\{l, l'\} \neq \{0,k\}$, $b \in O_{\MAJ}(P)$. 
We end up with the same PCSP when $\{\min S, \max S\} \neq \{0,k\}$. 

If $\{l,l'\} = \{1,k\}$ or $\{0,k-1\}$, we get that the {promise template} $P=\Ham_k\{1,k\},$\linebreak $Q=\Ham_k \{0,1,\ldots,k\}\setminus\{b\}$ is {fi}ppp-definable from the original {promise template}, and we are done. 
If not, we update the {promise template} while maintaining the two below properties:
\begin{enumerate}
    \item $P=\Ham_k \{l,l'\}$ with $l <l'$ and $\{ l, l' \} \neq \{0,k\}$ and $Q = \Ham_k \{0,1,\ldots, k\} \setminus \{b\}$. We also assume that $\{l,l'\} \neq \{1,k\}$ and $\{l,l'\}\neq \{0,k-1\}$. 
    \item $\AT, \MAJ \nsubseteq \Pol(P,Q)$. 
\end{enumerate}
As with the earlier case, we obtain a series of new {promise templates} that are {fi}ppp-definable from the previous {promise templates} using the below tools. 
\begin{enumerate}
    \item We get $P'= \Ham_{k-1} \{l,l'\}$ and $Q=\Ham_{k-1} \{0,1,\ldots,k\} \setminus \{b\}$ using Claim $4.2$ of~\cite{BrakensiekG21}. For this to preserve the above properties, we need that $l' \neq k, b \neq k$ and $b \neq 2l-k+1$. 
    \item We get $P'=\Ham_{k-1} \{l-1,l'-1\}$ and $Q=\Ham_{k-1} \{0,1,\ldots,k\} \setminus \{b-1\}$ using Claim $4.4$ of~\cite{BrakensiekG21}. For this to preserve the above properties, we need that $l \neq 0, b \neq 0$ and $b \neq 2l'-1$. 
\end{enumerate}
As the arity of the predicates is decreasing at each step, this process terminates in finite steps. 
When we are unable to update the {promise template} using the above procedures, one of the following must be true.
\begin{enumerate}
    \item $l'=k, b=0$. In this case, we have a {promise template} $(P,Q)$ where $P=\Ham_k \{l,k\}$, $Q=\Ham_k \{1,2,\ldots,k\}$, where $l\neq 0, l \leq \frac{k-1}{2}$. 
    \item $b=k,l=0$. By negating the variables, we can observe that the above {promise template} is {fi}ppp-definable from this {promise template}. 
    \item $b=k, b=2l'-1$. We have $l'=\frac{k+1}{2}$. In this case, we have {promise template} $(P,Q)$ where $P=\Ham_k \{l,\frac{k+1}{2}\}$, $Q=\Ham_k \{0,1,2,\ldots,k-1\}$, where $l \leq \frac{k-1}{2}$. 
    \item We have $b=2l-k+1, b=0$. By negating the variables, we can observe that the above {promise template} is {fi}ppp-definable from this {promise template}. 
\end{enumerate}

Thus, we have obtained a new {promise template} that is {fi}ppp-definable from the original {promise template} and is equal to either of the following {promise templates}. 
\begin{enumerate}
    \item $k$ is even, and $P=\Ham_{k}\{\frac{k}{2}\}, Q = \Ham_{k} \{0,1,\ldots,k\} \setminus \{b\}$ where $b \in \{1,k-1\}$. 
    \item $k$ is odd, $P=\Ham_k \{l,\frac{k+1}{2}\}$, $Q=\Ham_k \{0,1,2,\ldots,k-1\}$, where $l \leq \frac{k-1}{2}$.
    \item $P=\Ham_k \{l,k\}$, $Q=\Ham_k \{1,2,\ldots,k\}$, where $l\neq 0, l \leq \frac{k-1}{2}$. 
    \item $P=\Ham_k \{l\}, Q=\Ham_k \{0,1,\ldots,k\}\setminus \{0,k-1\}$ where $l \in \{1,2,\ldots,k-1\}, l \leq \frac{k-1}{2}$. 
    \item $P=\Ham_k \{1,k\}, Q=\Ham_k \{0,1,\ldots,k\}\setminus\{b\}$ for arbitrary $b$.
\end{enumerate}

Finally, we note that if a Boolean {promise template} $\Gamma'=\{(P_1, Q_1),\ldots,(P_l, Q_l)\}$ is {fi}ppp-definable from another Boolean {promise template} $\Gamma$, then $\Gamma'$ remains {fi}ppp-definable from $\Gamma$ even when we disallow the constraints {of} $\Gamma'$ {in the fippp-definition} to use constants{.}
\end{proof}

\section*{Acknowledgments} %
We profusely thank anonymous reviewers for numerous corrections and suggestions which drastically improved the quality and accuracy of this manuscript. We thank Ryan O'Donnell and Anupam Gupta for helpful discussions on SDP solvability which led to Appendix~\ref{app:efficient-sdp}. 

\providecommand{\bysame}{\leavevmode\hbox to3em{\hrulefill}\thinspace}
\begin{dajauthors}
\begin{authorinfo}[josh]
  Joshua Brakensiek\\
  Department of Electrical Engineering and Computer Sciences\\
  University of California, Berkeley\\
  Berkeley, California\\
  josh\imagedot{}brakensiek\imageat{}berkeley\imagedot{}edu \\
  \url{https://jbrakensiek.github.io/}
\end{authorinfo}
\begin{authorinfo}[venkat]
  Venkatesan Guruswami \\
  Departments of EECS \& Mathematics \\
  Simons Institute for the Theory of Computing \\ 
  University of California, Berkeley\\
  Berkeley, California\\
  venkatg\imageat{}berkeley\imagedot{}edu\\
  \url{https://people.eecs.berkeley.edu/~venkatg/}
\end{authorinfo}
\begin{authorinfo}[sandeep]
  Sai Sandeep\\
  Department of Electrical Engineering and Computer Sciences\\
  University of California, Berkeley\\
  Berkeley, California\\
  saisandeep\imageat{}berkeley\imagedot{}edu
\end{authorinfo}
\end{dajauthors}

\end{document}